\newif\ifFull\Fulltrue
\documentclass[a4paper,UKenglish,cleveref, autoref]{lipics-v2019}

\usepackage{graphicx}
\usepackage{amsthm}
\usepackage[dvipsnames]{xcolor}
\usepackage{wrapfig}
\usepackage{golang}
\usepackage[nomargin,draft,inline]{fixme}
\usepackage{tikz}
\usetikzlibrary{positioning}
\usepackage{pgf}
\usetikzlibrary{arrows,automata}
\usepackage{stmaryrd}
\usepackage{mathpartir}
\usepackage{mathtools}
\usepackage{xspace}
\usepackage{xifthen}
\usepackage{listings}
\usepackage{amssymb}
\usepackage{nicefrac}
\usepackage{booktabs}
\usepackage{refcount}
\usepackage{ougi}

\title{Static Race Detection and Mutex Safety and Liveness for Go Programs\full{(extended version)}{}} 

\author{Julia Gabet}{Imperial College London, United Kingdom}{j.gabet18@imperial.ac.uk}{https://orcid.org/0000-0001-9944-9497}{}

\author{Nobuko Yoshida}{Imperial College London, United Kingdom}{n.yoshida@imperial.ac.uk}{https://orcid.org/0000-0002-3925-8557}{}

\authorrunning{J. Gabet and N. Yoshida}

\Copyright{Julia Gabet and Nobuko Yoshida}



\category{}

\full{}{\relatedversion{A full version of the paper is available at \cite{JY20:full}, \url{http://arxiv.org/abs/2004.12859}.}}

\supplement{The source code for the tool presented in this paper and instructions to run it are available at \cite{web:godel2,web:benchmark,web:migoinfer}.}

\funding{The work is partially supported by VeTSS, EPSRC EP/K011715/1, EP/K034413/1, 
EP/L00058X/1, EP/N027833/1,  EP/N028201/1, EP/T006544/1 and EP/T014709/1.}

\acknowledgements{The authors want to thank Nicholas Ng for his initial collaboration on the project.}

\nolinenumbers 

\full{\hideLIPIcs}{} 

\EventEditors{Robert Hirschfeld and Tobias Pape}
\EventNoEds{2}
\EventLongTitle{34th European Conference on Object-Oriented Programming (ECOOP 2020)}
\EventShortTitle{ECOOP 2020}
\EventAcronym{ECOOP}
\EventYear{2020}
\EventDate{July 13--17, 2020}
\EventLocation{Berlin, Germany}
\EventLogo{}
\SeriesVolume{166}
\ArticleNo{4}

\newcommand{\lstCodeSize}{\scriptsize}
\newcommand{\lstPrimitiveStyle}{\color{blue}\bfseries}

\newcommand{\lstNumberStyle}{\tiny\sffamily\color{gray}}

\lstdefinestyle{nonumber}{%
  numbers=none,
  xleftmargin=1em,
  framexleftmargin=1em,
}

\newcommand{\ifempty}[3]{
  \ifthenelse{\isempty{#1}}{#2}{#3}
}

\theoremstyle{definition}

\newcommand{\godeltwo}{Godel2}

\newcommand{\subs}[2]{\left\{\nicefrac{#1}{#2}\right\}}

\newcommand{\domain}[1]{\mathsf{dom}(#1)}

\newcommand{\ov}[1]{\overline{#1}}
\newcommand{\tra}[1]{\xrightarrow{#1}}

\newcommand{\hb}[3][]{\ifthenelse{\isempty{#1}}{}{#1\triangleright}#2\mapsto #3}
\newcommand{\nothb}[3][]{\neg(\hb[#1]{#2}{#3})}

\newcommand{\with}{\mathbin{\binampersand}}

\newcommand{\G}{\Gamma}
\newcommand{\Gi}{\G '}
\newcommand{\Gii}{\G ''}

\newcommand{\vval}[1][]{\ifthenelse{\isempty{#1}}{v}{v_{#1}}}
\newcommand{\vvali}[1][]{\ifthenelse{\isempty{#1}}{v'}{v'_{#1}}}
\newcommand{\expr}[1][]{\ifthenelse{\isempty{#1}}{e}{e_{#1}}}
\newcommand{\expri}[1][]{\ifthenelse{\isempty{#1}}{e'}{e'_{#1}}}
\newcommand{\elli}{\ell '}

\newcommand{\psend}[2]{#1 ! \langle #2 \rangle }
\newcommand{\ENCan}[1]{\langle #1 \rangle}
\newcommand{\precv}[2]{#1 ? (#2) }

\newcommand{\pin}[2]{#1 \, \m{in} \, #2}

\newcommand{\sel}[2]{\m{select}\{#1\}_{#2} }
\newcommand{\newch}[1]{\m{newchan}(#1)}

\newcommand{\close}[1]{\m{close}\,#1}

\newcommand{\defcall}[2]{#1\langle #2 \rangle}

\newcommand{\buff}[3]{{#2}\ENCan{{#3}}{::}#1}
\newcommand{\closedbuff}[3]{{#2}^\star\ENCan{#3}{::}{#1}}

\newcommand{\buflen}[1]{|#1|}

\newcommand{\dataAct}[2]{#1,#2}

\newcommand{\ruleChanSendReq}{snd}
\newcommand{\ruleChanReceiveReq}{rvc}
\newcommand{\ruleChanCloseReq}{end}
\newcommand{\ruleChanCloseAck}{buf}
\newcommand{\ruleChanSendAck}{push}
\newcommand{\ruleChanReceiveAck}{pop}
\newcommand{\ruleChanClosedReceiveAck}{cpop}
\newcommand{\ruleChanClosedDefaultValue}{cld}
\newcommand{\ruleChanClose}{close}
\newcommand{\ruleChanSyncCom}{scom}
\newcommand{\ruleChanAsyncSend}{out}
\newcommand{\ruleChanAsyncReceive}{in}

\newcommand{\ruleMutLockReq}{lck}
\newcommand{\ruleMutUnlockReq}{ulck}

\newcommand{\ruleRMutRLockReq}{rlck}
\newcommand{\ruleRMutRUnlockReq}{rulck}
\newcommand{\ruleMutLockAck}{m-lck}
\newcommand{\ruleMutUnlockAck}{m-ulck}
\newcommand{\ruleRMutLockAck}{rw-lck}
\newcommand{\ruleRMutUnlockAck}{rw-ulck}
\newcommand{\ruleRMutRLockAck}{rw-rlck}
\newcommand{\ruleRMutRUnlockAck}{rw-rulck}
\newcommand{\ruleRMutLockStageAck}{rw-wait}
\newcommand{\ruleRMutStagedRUnlockAck}{rw-wulck}
\newcommand{\ruleMutLock}{c-lck}
\newcommand{\ruleMutUnlock}{c-ulck}
\newcommand{\ruleRMutLock}{c-lck}
\newcommand{\ruleRMutUnlock}{c-ulck}
\newcommand{\ruleRMutRLock}{c-rlck}
\newcommand{\ruleRMutRUnlock}{c-rulck}
\newcommand{\ruleRMutLockStage}{c-wait}

\newcommand{\ruleNewChan}{newc}
\newcommand{\ruleNewMut}{newm}
\newcommand{\ruleNewRMut}{newrwm}
\newcommand{\ruleNewMem}{newv}
\newcommand{\ruleSilentAction}{tau}
\newcommand{\ruleParLeft}{par-l}
\newcommand{\ruleParRight}{par-r}
\newcommand{\rulePar}{par}
\newcommand{\ruleAlphaEquiv}{alpha}

\newcommand{\ruleRestrict}{res}
\newcommand{\ruleRestrictFree}{res1}
\newcommand{\ruleRestrictBind}{res2}
\newcommand{\ruleSelectConstruct}{bra}
\newcommand{\ruleITEConstruct}{sel}
\newcommand{\ruleITETrue}{ift}
\newcommand{\ruleITEFalse}{iff}
\newcommand{\ruleDefinition}{def}
\newcommand{\ruleContinue}{con}
\newcommand{\ruleTransitivity}{tra}
\newcommand{\ruleUnlockLock}{u-l}
\newcommand{\ruleUnlockRLock}{u-rl}
\newcommand{\ruleRUnlockLock}{ru-l}
\newcommand{\ruleLockRLock}{l-rl}
\newcommand{\ruleSendReceiveAsync}{buf}
\newcommand{\ruleSendReceiveSync}{scom}
\newcommand{\ruleCloseBeforeDefault}{end}
\newcommand{\ruleReceiveFromClosed}{cl-rcv}
\newcommand{\ruleSelectSendReceiveAsync}{buf-snd}
\newcommand{\ruleSendSelectReceiveAsync}{buf-rcv}
\newcommand{\ruleSelectSendReceiveSync}{scom-snd}
\newcommand{\ruleSendSelectReceiveSync}{scom-rcv}
\newcommand{\ruleReduction}{red}
\newcommand{\ruleRestrictChan}{resc}
\newcommand{\ruleRestrictMem}{resv}
\newcommand{\ruleRestrictMut}{resm}

\newcommand{\ruleTypingPar}{parr}
\newcommand{\ruleTypingMut}{mut}
\newcommand{\ruleTypingMutLocked}{l-m}
\newcommand{\ruleTypingRMut}{rmut}
\newcommand{\ruleTypingRMutLocked}{l-rw}
\newcommand{\ruleTypingRMutStaged}{w-rw}
\newcommand{\ruleTypingChan}{buf}
\newcommand{\ruleTypingChanClosed}{c-buf}
\newcommand{\ruleTypingMem}{heap}
\newcommand{\ruleTypingCall}{var}
\newcommand{\ruleTypingZero}{zero}

\newcommand{\migoP}[1][]{\ifthenelse{\isempty{#1}}{P}{P_{#1}}}
\newcommand{\migoPi}[1][]{\ifthenelse{\isempty{#1}}{P'}{P'_{#1}}}
\newcommand{\migoPii}[1][]{\ifthenelse{\isempty{#1}}{P''}{P''_{#1}}}
\newcommand{\migoQ}[1][]{\ifthenelse{\isempty{#1}}{Q}{Q_{#1}}}
\newcommand{\migoQi}[1][]{\ifthenelse{\isempty{#1}}{Q'}{Q'_{#1}}}
\newcommand{\migoR}[1][]{\ifthenelse{\isempty{#1}}{R}{R_{#1}}}
\newcommand{\migoRi}[1][]{\ifthenelse{\isempty{#1}}{R'}{R'_{#1}}}

\newcommand{\migoX}[1][]{\ifthenelse{\isempty{#1}}{X}{X_{#1}}}
\newcommand{\migoXi}[1][]{\ifthenelse{\isempty{#1}}{X'}{X'_{#1}}}

\newcommand{\migoDef}[1][]{\ifthenelse{\isempty{#1}}{D}{D_{#1}}}
\newcommand{\migoDefi}[1][]{\ifthenelse{\isempty{#1}}{D'}{D'_{#1}}}
\newcommand{\migoProg}[1][]{\ifthenelse{\isempty{#1}}{\prog}{\prog_{#1}}}
\newcommand{\migoProgi}[1][]{\ifthenelse{\isempty{#1}}{\prog'}{\prog'_{#1}}}

\newcommand{\migoCont}{;}
\newcommand{\migoRes}[2]{(\nub #1)#2}

\newcommand{\migoPref}[1][]{\ifthenelse{\isempty{#1}}{\mu}{\mu_{#1}}}
\newcommand{\migoChanPref}[1][]{\ifthenelse{\isempty{#1}}{\pi}{\pi_{#1}}}
\newcommand{\migoMutPref}[1][]{\ifthenelse{\isempty{#1}}{\ell}{\ell_{#1}}}

\newcommand{\migoSel}[3]{\sel{#1\migoCont #2}{#3}}

\newcommand{\migoNewch}[3]{\newch{#1{:}#2,#3}}

\newcommand{\name}{u}

\newcommand{\varname}{x}

\newcommand{\varnamey}{y}

\newcommand{\nameVec}{\tilde{\name}}

\newcommand{\varyVec}[1][]{\ifthenelse{\isempty{#1}}{\tilde{\varnamey}}{\tilde{\varnamey}_{#1}}}
\newcommand{\channame}[1][]{\ifthenelse{\isempty{#1}}{c}{c_{#1}}}
\newcommand{\channamei}[1][]{\ifthenelse{\isempty{#1}}{c'}{c'_{#1}}}
\newcommand{\channamea}[1][]{\ifthenelse{\isempty{#1}}{a}{a_{#1}}}
\newcommand{\channameb}[1][]{\ifthenelse{\isempty{#1}}{b}{b_{#1}}}
\newcommand{\chanVec}{\tilde{\channame}}
\newcommand{\migotmp}[1][]{\ifthenelse{\isempty{#1}}{t}{t_{#1}}}
\newcommand{\actname}[1][]{\ifthenelse{\isempty{#1}}{o}{o_{#1}}}
\newcommand{\actnamei}[1][]{\ifthenelse{\isempty{#1}}{o'}{o'_{#1}}}
\newcommand{\actnameii}[1][]{\ifthenelse{\isempty{#1}}{o''}{o''_{#1}}}
\newcommand{\actnameVec}[1][]{\tilde{\actname[#1]}}
\newcommand{\actgen}[1][]{\ifthenelse{\isempty{#1}}{\alpha}{\alpha_{#1}}}
\newcommand{\actgeni}[1][]{\ifthenelse{\isempty{#1}}{\alpha'}{\alpha'_{#1}}}
\newcommand{\actgenb}[1][]{\ifthenelse{\isempty{#1}}{\beta}{\beta_{#1}}}
\newcommand{\acttau}{\tau}
\newcommand{\occurnone}{\ast}
\newcommand{\occurgen}[1][]{\ifthenelse{\isempty{#1}}{\iota}{\iota_{#1}}}
\newcommand{\occurgeni}[1][]{\ifthenelse{\isempty{#1}}{\iota'}{\iota'_{#1}}}
\newcommand{\occurleft}[1]{\ifthenelse{\isempty{#1}}{1.\occurgen}{1.{#1}}}
\newcommand{\occurright}[1]{\ifthenelse{\isempty{#1}}{2.\occurgen}{2.{#1}}}
\newcommand{\actleft}[1]{\m{left}\left(#1\right)}
\newcommand{\actright}[1]{\m{right}\left(#1\right)}

\newcommand{\tyact}[1][]{\ifthenelse{\isempty{#1}}{\ell}{\ell_{#1}}}
\newcommand{\tyacti}[1][]{\ifthenelse{\isempty{#1}}{\elli}{\elli_{#1}}}

\newcommand{\tysend}[1]{\overline{#1}}
\newcommand{\tyrecv}[1]{#1}

\newcommand{\tsend}[2]{\tysend{#1}} 
\newcommand{\trecv}[2]{\tyrecv{#1}} 

\newcommand{\ctxComp}{,}

\newcommand{\tch}[3]{#1 \oplus \{#2\}_{#3}}
\newcommand{\tbr}[3]{#1 \with \{#2\}_{#3}}

\newcommand{\tdefcall}[2]{\typevarfont{#1}\langle #2 \rangle}
\newcommand{\End}{\mathsf{end}}

\newcommand{\ts}{\blacktriangleright}

\newcommand{\fgo}{{\sf MiGo}\xspace}

\newcommand{\asyncfgo}{{\sf AMiGo}\xspace}
\newcommand{\fgom}{{\sf MiGo}${}^+$\xspace}

\newcommand{\llab}{\m{L}}
\newcommand{\rlab}{\m{R}}
\newcommand{\iflab}[2]{#1 \!\cdot\! #2}

\newcommand{\traceset}{\mathbb{T}}

\newcommand{\clbarb}[1]{\labclose{#1}}
\newcommand{\cclbarb}[1]{\labclsnd{#1}}
\newcommand{\stbarb}[2]{(\labstore{#1},#2)}
\newcommand{\ldbarb}[2]{\labload{#1}}

\newcommand{\heapbarb}[1]{\labheap{#1}}
\newcommand{\lckbarb}[1]{\lablock{#1}}
\newcommand{\ulckbarb}[1]{\labunlock{#1}}
\newcommand{\rlckbarb}[1]{\labrlock{#1}}
\newcommand{\rulckbarb}[1]{\labrunlock{#1}}
\newcommand{\mutbarb}[1]{\labmut{#1}}
\newcommand{\lckmutbarb}[1]{\lablckmut{#1}}
\newcommand{\rmutbarb}[2]{\labrmut{#1}{}}

\newcommand{\waitrmutbarb}[2]{\labwaitrmut{#1}{}}
\newcommand{\waitmutbarb}[1]{\labwaitmut{#1}}

\newcommand{\ac}{\m{AC}}
\newcommand{\ic}{\m{Inf}}
\newcommand{\convm}{\m{May}{\mbox{\scriptsize $\Downarrow$}}}
\newcommand{\infcond}{\m{InfCond}}
\newcommand{\term}{\m{Terminate}}

\newcommand{\ltsT}{\mathcal{T}}
\newcommand{\ltsStates}{\mathcal{S}}
\newcommand{\ltsTransitions}{\mathcal{A}}
\newcommand{\ltsFlags}[1]{\mathcal{F}(#1)}

\newcommand{\muCol}[1]{{\color{blue}#1}}
\newcommand{\muFmt}[1]{\muCol{\mathsf{#1}}}

\newcommand{\muVar}[1][]{\muCol{\ifempty{#1}{\muFmt{Z}}{\muFmt{Z}_{#1}}}}

\newcommand{\muAct}[1][]{\muCol{\ifempty{#1}{\alpha}{\alpha_{#1}}}}
\newcommand{\muSync}{\muCol{\mathbb{S}}}

\newcommand{\muPred}[1][]{\muCol{\ifempty{#1}{\phi}{\phi_{#1}}}}

\newcommand{\muNeg}[1]{\muCol{\neg{#1}}}
\newcommand{\muAnd}{\mathbin{\muCol{\land}}}
\newcommand{\muConj}[2]{\muCol{\bigwedge_{#1}{#2}}}

\newcommand{\muSome}[2]{\muCol{\left\langle {#1} \right\rangle {#2}}}
\newcommand{\muAll}[2]{\muCol{\left[{#1}\right]{#2}}}
\newcommand{\muTrue}{\muCol{\top}}
\newcommand{\muGFP}[2]{\muCol{\nu{#1}\mathbin{\!.\!}{#2}}}

\newcommand{\muLFP}[2]{\muCol{\mu{#1}\mathbin{\!.\!}{#2}}}
\newcommand{\muOr}{\mathbin{\muCol{\lor}}}

\newcommand{\muFalse}{\muCol{\bot}}
\newcommand{\muImplies}{\mathbin{\muCol{\Rightarrow}}}

\newcommand{\muSomeNeg}[2]{\muCol{\langle -{#1} \rangle {#2}}}
\newcommand{\muAllNeg}[2]{\muCol{[-{#1}]{#2}}}

\newcommand{\muWordEmpty}[1][]{\muCol{\epsilon}}

\newcommand{\EqTypes}[1][]{\ifthenelse{\isempty{#1}}{\mathbfsf{T}}{\mathbfsf{T}_{#1}}}
\newcommand{\typrefix}[1][]{\ifthenelse{\isempty{#1}}{\kappa}{\kappa_{#1}}}
\newcommand{\tymemprefix}[1][]{\ifthenelse{\isempty{#1}}{\vartheta}{\vartheta_{#1}}}
\newcommand{\tymutprefix}[1][]{\ifthenelse{\isempty{#1}}{\xi}{\xi_{#1}}}
\newcommand{\labtysnd}[1]{\overline{#1}}
\newcommand{\labtyrcv}[1]{{#1}}
\newcommand{\labtystore}[2]{(\labstore{#1},#2)}
\newcommand{\labtyload}[1]{\labload{#1}}
\newcommand{\labtyheap}[1]{\labheap{#1}}
\newcommand{\labtylock}[1]{\lablock{#1}}
\newcommand{\labtyunlock}[1]{\labunlock{#1}}
\newcommand{\labtyrlock}[1]{\labrlock{#1}}
\newcommand{\labtyrunlock}[1]{\labrunlock{#1}}

\newcommand{\tyvary}{y}
\newcommand{\tyvarz}{z}

\newcommand{\tyvaryVec}[1][]{\ifthenelse{\isempty{#1}}{\tilde{\tyvary}}{\tilde{\tyvary}_{#1}}}

\newcommand{\typevarfont}[1]{\mathbfsf{#1}}
\newcommand{\typevara}{\typevarfont{t}}

\newcommand{\tyclosedbuffer}[1]{#1^\star}

\newcommand{\tyheap}[1]{\labheap{#1}}

\newcommand{\labsync}[1]{\tau_{#1}}
\newcommand{\labclose}[1]{\End[#1]}
\newcommand{\labclosedual}[1]{\overline{\End}[#1]}
\newcommand{\labclsnd}[1]{\tyclosedbuffer{#1}}
\newcommand{\labload}[1]{\mathsf{r}\langle #1\rangle}
\newcommand{\labstore}[1]{\mathsf{w}\langle #1\rangle}
\newcommand{\labloaddual}[1]{\overline{\labload{#1}}}
\newcommand{\labstoredual}[1]{\overline{\labstore{#1}}}
\newcommand{\labheap}[1]{#1^{\mbox{\tiny${\blacksquare}$}}}
\newcommand{\lablock}[1]{\mathsf{l}\langle #1 \rangle}
\newcommand{\labunlock}[1]{\mathsf{ul}\langle #1 \rangle}
\newcommand{\labrlock}[1]{\mathsf{rl}\langle #1 \rangle}
\newcommand{\labrunlock}[1]{\mathsf{rul}\langle #1 \rangle}
\newcommand{\labmut}[1]{\ulcorner #1 \urcorner}
\newcommand{\lablckmut}[1]{\ulcorner #1 \urcorner^{\star}}
\newcommand{\labrmut}[2]{\llcorner #1 \lrcorner_{#2}}

\newcommand{\labwaitrmut}[2]{\llcorner #1 \lrcorner^{\blacktriangledown}_{#2}}
\newcommand{\labwaitmut}[1]{\labrmut{#1}{}^{\blacktriangle}}
\newcommand{\labspawn}[1]{\mathit{sp}}
\newcommand{\semty}[1]{\xrightarrow{#1}}

\newcommand{\closedchans}{{N}} 
\newcommand{\typesem}[2]{#2}
\newcommand{\typesemF}[1]{\typesem{\closedchans}{#1}}

\newcommand{\typesemFdash}[1]{\typesem{\closedchans'}{#1}}

\newcommand{\tycontxt}{\mathbb{C}}

\newcommand{\conv}{\downarrow}
\newcommand{\MAPAST}[1]{{#1}^\ast}

\newcommand{\astyOpenbuf}[3]{\lfloor #1 \rfloor_{#2}^{#3}}
\newcommand{\astyBufRcv}[1]{#1^{\bullet}} 
\newcommand{\astyBufSend}[1]{^{\bullet} #1} 
\newcommand{\astynew}[2]{(\nu\,#1^{#2})}

\newcommand{\trulename}[1]{\text{\scriptsize$\langle$\sc #1$\rangle$}\xspace}
\newcommand{\rulename}[1]{\text{\scriptsize[\sc #1]}\xspace}
\newcommand{\ltsrulename}[1]{\text{\scriptsize$|$\sc #1$|$}\xspace}

\newcommand{\hbrulename}[1]{\text{\scriptsize(\sc #1)}\xspace}

\newcommand{\yes}{{\ensuremath{\checkmark}}}
\newcommand{\noo}{{\color{red} \ensuremath{\times}}}

\begin{document}

\maketitle

\begin{abstract}
Go is a popular concurrent programming language thanks to its
ability to efficiently combine concurrency and systems programming.
In Go programs, a number of concurrency bugs can be caused by a 
mixture of data races and communication problems. In this
paper, we develop a theory based on behavioural types to
statically detect data races and deadlocks in Go programs. We first
specify lock safety/liveness and data race properties over a Go
program model, using the happens-before relation defined in the
Go memory model. 
We represent these properties
of programs in a $\mu$-calculus model of types,
and validate them using type-level model-checking.  
We then extend the framework to account for Go's 
channels, and implement a static verification tool
which can detect concurrency errors. 
This is, to the best of our knowledge, the first static
verification framework of this kind for the Go language, uniformly
analysing concurrency errors caused by a mix of shared memory
accesses and asynchronous message-passing communications.

\end{abstract}

\section{Introduction}
\label{sec:introduction}
Go is a concurrent programming language designed by Google for
\textit{programming at scale} \cite{pike:go}.
Over the last few years, it
has seen rapid growth and adoption: for instance in 2018,
major developer surveys~\cite{web:gosurvey18} 
show that StackOverflow placed Go in the
top 5 most loved and the top 5 most wanted languages; and
Github has reported in~\cite{web:octoverse}
that Go was the
7th fastest growing language.

One of the core pillars of Go is concurrent programming 
features, including the locking of shared memory for thread
synchronisation, and 
the use of explicit message
passing through channels,
inspired by 
process calculi 
concurrency models \cite{book:csp,book:CCS}. 
In practice, shared
accesses to memory using locking mechanisms are unavoidable, 
and could be accidental. 
It is also of note that both shared memory and 
message passing operations 
provide a substantial part of the 
concurrency features of Go, and are the ones that are more 
prone to misuse-induced bugs.
These unsafe memory accesses may lead to \emph{data races},
where programs silently enter an inconsistent execution state leading
to hard-to-debug failures.
\full{\vspace{8mm}}{}
\begin{wrapfigure}{l}{0.6\linewidth}
\vspace{-4mm}
{
\begin{lstlisting}[language=Go,style=golang,firstline=9]
func main() {
	var x int
	m := new(sync.RWMutex) /*<\label{line:example-newrwm-r}>*/
	go /*<\label{line:ex-rwm-spawn-r}\hskip -0.6em>*/ f(m, &x)
	m.RLock()   // acquire the lock for reading
	x += 10     // write not protected by the lock
	m.RUnlock() // release the read-lock
	m.Lock()    // acquire the lock for writing
	fmt.Println("x is", x)
	m.Unlock()  // release the write-lock
}

func f(m *sync.RWMutex, ptr *int) {
	m.RLock()
	*ptr += 20  // write not protected by the lock
	m.RUnlock()
}
\end{lstlisting}
}
\vspace{-3mm}
\caption{Go program with RWMutex (unsafe)}
\label{fig:running-ex-rwmut-race}
\vspace{-5mm}
\end{wrapfigure}

\indent\figurename~\ref{fig:running-ex-rwmut-race} illustrates a Go program,
which makes use of lock \lstgo{m} to synchronise the \lstgo{main} and \lstgo{f}
functions updating the content of variable \lstgo{x}. 
On line~\ref{line:example-newrwm-r}, the statement 
\lstgo{m := new(sync.RWMutex)} creates
a new read-write lock \lstgo{m}, called \lstgo{RWMutex} in Go, used to 
guard memory accesses based on their status as readers or writers.
The \lstgo{RWMutex} 
object can then be passed around directly as on line~\ref{line:ex-rwm-spawn-r}, 
circumventing the issue that could arise 
if we copied the mutex structure instead. 
It can be locked for writing by calling its \lstgo{Lock()} 
method, unlocked from writing handle with its \lstgo{Unlock()} 
method, and locked and unlocked for reading with the \lstgo{RLock()} 
and \lstgo{RUnlock()} methods. Readers and writers are mutually exclusive, and 
writers are mutually exclusive to each other too (hence the name \lstgo{Mutex}, 
for {\em mutual exclusion} lock), but an arbitrary number of 
readers can hold the lock at the same time.
The \lstgo{go} keyword in front of a function call on
line~\ref{line:ex-rwm-spawn-r} spawns a lightweight thread
(called a \textit{goroutine}) to execute the body of function \lstgo{f}.
The two parameters of function \lstgo{f} -- a rwmutex \lstgo{m}, and an
\lstgo{int} pointer \lstgo{ptr} -- are shared between the caller and callee
goroutines, \lstgo{main} and \lstgo{f}. Since concurrent access to the shared
pointer \lstgo{ptr} may introduce a data race, the developer tries to ensure 
serialised, mutually exclusive access to \lstgo{ptr} in
\lstgo{f} and \lstgo{x} in \lstgo{main} by using read-locks. Using read-locks 
is unsafe in this case, allowing
simultaneous write requests to \lstgo{x} on lines~\ref{line:ex-rwm-main-lock} and
\ref{line:ex-rwm-f-lock}, the program could then output 
``\texttt{x is 20}'' with a bad scheduling, dropping the increase of 10
in the same thread as the print statement.

\begin{wrapfigure}{r}{0.57\linewidth}
\vspace{-7mm}
{
\begin{lstlisting}[language=Go,style=golang,firstline=9]
func main() {
	var x int
	m := new(sync.RWMutex) /*<\label{line:example-newrwm-s}>*/
	go /*<\label{line:ex-rwm-spawn-s}\hskip -0.6em>*/ f(m, &x)
	m.Lock()    // acquire the lock for writing
	x += 10     // protected by the lock/*<\label{line:ex-rwm-main-lock}>*/
	m.Unlock()  // release the write-lock
	m.RLock()   // acquire the lock for reading
	fmt.Println("x is", x)
	m.RUnlock() // release the read-lock
}

func f(m *sync.RWMutex, ptr *int) {
	m.Lock()
	*ptr += 20  // protected by the lock/*<\label{line:ex-rwm-f-lock}>*/
	m.Unlock()
}
\end{lstlisting}
}
\vspace{-3mm}
\caption{Go program with RWMutex (safe)}
\label{fig:running-ex-rwmut-safe}
\vspace{-5mm}
\end{wrapfigure}

\figurename~\ref{fig:running-ex-rwmut-safe} illustrates the same Go program, 
using the \lstgo{RWMutex} feature correctly by putting writer sections of the 
code under writer locks. This alone prevents the data race seen in the first version 
of the program.

Go provides an optional \emph{runtime} data race
detector~\cite{web:go-race-detector-blog,web:go-race-detector} as a part of
the Go compiler toolchain. The race detector is based on LLVM's
ThreadSanitizer~\cite{tsan,web:tsan,SPIV11:tsan} library, which detects
races that manifest during execution.
It can be enabled by building a program using the ``\texttt{-race}'' flag.
During the program execution, the race detector creates up to four
shadow words for 
every memory object to store historical accesses of the object. It
compares every new access with the stored shadow word values to detect
possible races. 
These runtime operations cause high overheads of 
the runtime detector (5--10 times overhead in memory usage and 2--20 times in execution
time on average~\cite{web:go-race-detector}), hence 
it is unrealistic to run it with race detection turned on in production
code; and because of that, race detection
relies on extensive testing or fuzzing techniques~\cite{gh:gofuzz,gh:syzkaller}. 
Moreover,  
as reported in \cite{Tu2019},
the detector fails to find many non-blocking bugs as 
it cannot keep a sufficiently long enough history; and 
its semantics does not capture Go specific non-blocking bugs. 

The Go memory model~\cite{web:go-mem-model} defines the behaviour of memory
access in Go as a \emph{happens-before} relation 
by a \emph{combination} of shared memory and
channel communications. 
It is also reported in~\cite{Tu2019} that
the most 
difficult bugs to detect are caused when 
synchronisation mechanisms are used together with 
message passing operations. 
For instance, 
Go can use message passing for sharing memory (channel-as-lock) or
passing pointers 
through channels (pointer-through-channel), which might lead to
a serious non-blocking bug, i.e. the program may continue to 
execute in unwanted and incorrect states or 
corrupt data in its computations~\cite{Tu2019}, due to subtle interplays 
with buffered asynchronous communications.

These motivate us to \emph{uniformly} model, statically
analyse and detect concurrent
non-blocking/blocking shared memory/channel-communications bugs in Go,
using a formal model based on a process calculus \cite{book:csp,book:CCS}.
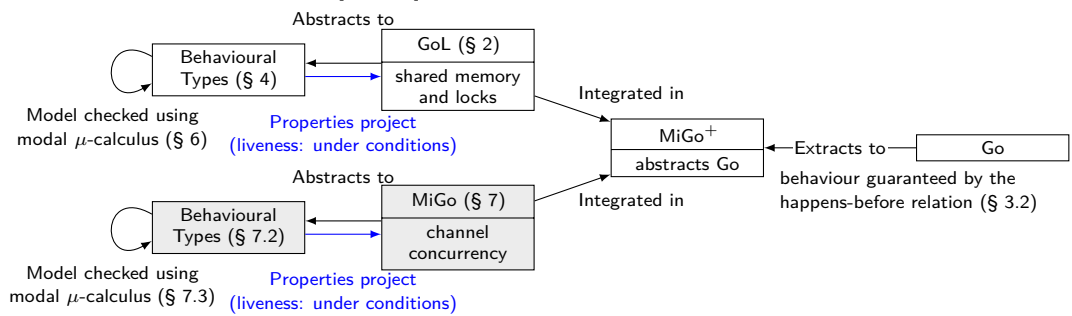
\begin{figure}[h]
\vspace*{-3mm}
\begin{tikzpicture}[
    mtool/.style={draw,text width=7em,text centered,fill=gray!15},
    tool/.style={draw,text width=7em,text centered},
    arrowlabel/.style={inner sep=.5,fill=pagecolor},
    arrowlabelabove/.style={above=.5,inner sep=0},
    arrowlabelunder/.style={below=.5,inner sep=0},
    arrowlabelunderp/.style={below right=.3 and -.8,inner sep=0},
    arrowlabelunderpp/.style={below right=.6 and -.8,inner sep=0},
    arrowlabelunderr/.style={below=.8,inner sep=0},
    arrowlabelon/.style={above right=.1,inner sep=0},
    arrowlabeloff/.style={below right=.1,inner sep=0},
    comment/.style={text width=\linewidth-12em,anchor=north west,
    fill=gray!10 
    }
  ]
  \scriptsize\sffamily
  \node (migop) [rectangle split,rectangle split parts=2,tool] {\fgom \nodepart{two} abstracts Go};
  \node (migo) [rectangle split,rectangle split parts=2,below left=.1 and 1 of migop,mtool] 
        {\fgo (\S~\ref{sect:chan-extension})\nodepart{two} channel concurrency};
  \node (gol) [rectangle split,rectangle split parts=2,above left=.1 and 1 of migop,tool]
        {\gol (\S~\ref{sect:core})\nodepart{two} shared memory and locks};
  \node (behavmigo) [left=1 of migo,mtool] {Behavioural Types (\S~\ref{subsect:chan-live-safe})};
  \node (behavgol) [left=1 of gol,tool] {Behavioural Types (\S~\ref{sect:core_typ})};
  \node (go) [right=2 of migop,tool] {Go};
  \path [draw,->,>=latex] (migo) -- node [arrowlabeloff] {Integrated in} (migop);
  \path [draw,->,>=latex] (gol)  -- node [arrowlabelon] {Integrated in} (migop);
  \path [draw,->,>=latex] (migo.175) -- node [arrowlabelabove] {Abstracts to} (behavmigo.005);
  \path [draw,->,>=latex,blue] (behavmigo.355) -- node [arrowlabelunder] {Properties project} node [arrowlabelunderr] {(liveness: under conditions)} (migo.185);
  \path [->,>=latex] (behavmigo.170) edge[bend right=120,looseness=5,min distance=8mm] 
        node [arrowlabelunder] {Model checked using} node [arrowlabelunderr] {modal $\mu$-calculus (\S~\ref{subsect:mod-mu-chan})} (behavmigo.190);
  \path [draw,->,>=latex] (gol.175) -- node [arrowlabelabove] {Abstracts to} (behavgol.005);
  \path [draw,->,>=latex,blue] (behavgol.355) -- node [arrowlabelunder] {Properties project} node [arrowlabelunderr] {(liveness: under conditions)} (gol.185);
  \path [->,>=latex] (behavgol.170) edge[bend right=120,looseness=5,min distance=8mm] 
        node [arrowlabelunder] {Model checked using} node [arrowlabelunderr] {modal $\mu$-calculus (\S~\ref{sect:form-hb-modmu})} (behavgol.190);
  \path [draw,->,>=latex] (go) -- node [arrowlabel] {Extracts to} node [arrowlabelunderp] {behaviour guaranteed by the} 
        node [arrowlabelunderpp] {happens-before relation (\S~\ref{subsec:happenbefore})} (migop);
\end{tikzpicture}
  \caption{Overview of this paper.\label{fig:results}}
\vspace{-6mm}
\end{figure}
\\[1.5mm]
\noindent{\bf Contributions and Outline.}
\figurename~\ref{fig:results} outlines the relationship between the results presented in this paper. 
This work proposes a {\em uniform} model which handles first 
shared memory concurrency (\S~\ref{sect:core}), and then
message-passing concurrency (\S~\ref{sect:chan-extension}) 
based on 
\emph{concurrent behavioural types}, and  
presents the theory, design and 
implementation of a concurrent bug detector for Go. 
We formalise a happens-before relation and several key safety and liveness 
properties in the process calculus following the Go memory 
model~\cite{web:go-mem-model} (\S~\ref{sub:live-safe-process}).
More specifically, in this work, we present the \goldy language (\gol for short), 
used as a subset of processes of the Go language, 
and the behavioural types used to model 
mutual-exclusion locks and shared memory primitives.
We then use this calculus and its types to tackle 
lock liveness and safety, as well as another form of safety: 
\emph{data race detection}. Our further extension to channels 
(\S~\ref{sect:chan-extension}) enables us to detect the errors caused 
by a mixture of shared memory and message passing concurrency.  
The formulation of a happens-before relation and
classification of a data race with respect to the Go memory model 
along with static analysis of this kind is, to the best of our knowledge, 
the first of its kind, at least for Go and its mixed memory management features. 

Through type soundness and progress theorems of our behavioural 
typing system (\S~\ref{sect:core_typ},\S~\ref{sec:types}), 
we are able to represent properties of processes by those of types 
in the modal $\mu$-calculus (\S~\ref{sect:form-hb-modmu}). 
In this paper, we explore in particular the formal relationship 
between type-level properties given by 
the modal $\mu$-calculus and process properties: 
we prove which subsets of \gol satisfy 
the properties of the types characterised by the modal $\mu$-calculus
(Theorem \ref{the:mod-mu-proc-check}). 

We also present a static analysis tool based on the theory. 
The tool infers from Go programs~\cite{web:migoinfer}
the memory accesses, locks 
and message-passing primitives
as behavioural types, 
and generates 
a $\mu$-calculus model from these types~\cite{web:godel2}. 
We then apply the mCRL2 model checker~\cite{Cranen2013} 
to detect blocking and non-blocking concurrency
errors (\S~\ref{sec:implementation}).  
We conclude the paper with an overview of related works (\S~\ref{sec:related}).

Detailed proofs and additional material can be found in%
\full{{\bf the Appendix.}}{the full version of the paper \cite{JY20:full}.}%
The tool and benchmark are available from~\cite{web:godel2,web:benchmark,web:migoinfer}.

\section{ \gol: a Memory-Aware Core Language for Go}\label{sect:core}
This section introduces a core language that models
shared memory concurrency, dubbed \goldy (simple subset of Go 
with shared memory primitives and locks only), shortened as \gol. 
\gol supports two key features for shared memory concurrency: 
(1) \emph{shared variables}, created by 
a shared variable creation primitive, whose values 
can be read from and written to  
by multiple threads; and (2) \emph{locks} and \emph{read-write locks} (rwlocks) are modelled 
by creating a lock store, and recording how it is accessed 
by (read-)lock and (read-)unlock calls. 

\subsection{Syntax of \gol}\label{sect:core_syn}
\begin{figure}
\vspace*{-8mm}
\begin{center}
\framebox{
\begin{tabular}{l}
\begin{tabular}{c|c}
$\!\!
\begin{array}{rcl}
\procP,\procQ,\procR & \coloneqq & \prefMem\pCont\procP \ \mid \ \pPPar{\procP}{\procQ}
    \ \mid \ \zero \ \mid \ \pRes{\nGeneric}{\procP}\\
    &  |  & \pITECont{\nExpr}{\procP}{\procQ}{\procR}\\
    &  |  & \pCall{\procX}{\vExpr}{\vGeneric} \ \mid \ \pNewMem{\nVar}{\sigma}\pCont\procP\\
    &  |  & \pNewLock{\nLock}\pCont\procP \ \mid \ \pNewRWLock{\nLock}\pCont\procP\\
    &  |  & \stMem{\nVar}{\sigma}{\nVal} \ \mid \ \stLock{\nLock} \ \mid \ \stLockL{\nLock}\\
    &  |  & \stRWLock{\nLock}{i} \ \mid \ \stRWLockL{\nLock}{i} \ \mid \ \stRWLockW{\nLock}{i}
\end{array}
$
& 
$\begin{array}{lcl}
  \procD   & \coloneqq & \pDef{\procX}{\vVar}{\procP}\\
\procProg & \coloneqq & \pIn{\{ \procD[i] \}_{i\in I}}{\procP}\\
\prefMem & \coloneqq & \pSilent \ \mid \ \pLoad{\nVar}{\nVary}\\
    & \mid & \pStore{\nVar}{\nExpr} \ \mid \ \prefLock \\
\prefLock & \coloneqq & \pLock{\nLock} \ \ \mid \ \ \pUnlock{\nLock}\\
      & \mid & \pRLock{\nLock} \ \mid \ \pRUnlock{\nLock}\\
\nVal & \coloneqq  & \num \mid \true \mid \false \mid \nVar \\
\nExpr & \coloneqq & \nVal \mid \opNot{\nExpr} \mid \opSucc{\nExpr}
\end{array}\!\!
$
\end{tabular}
\\
\hline
\\[-0.7em]
$
\begin{array}{c}
\pPar{\procP}{\procQ} \equiv \pPar{\procQ}{\procP}
\qquad
\pPar{\procP}{\pPPar{\procQ}{\procR}} \equiv \pPar{\pPPar{\procP}{\procQ}}{\procR}
\qquad 
\pPar{\procP}{\zero} \equiv \procP
\qquad
\pRes{\nVar}{\stMem{\nVar}{\sigma}{\nVal}} \equiv \zero
\\[1mm]
\pRes{\nLock}{\stLock{\nLock}} \equiv \zero
\qquad
\pRes{\nLock}{\stLockL{\nLock}} \equiv \zero
\qquad
\pRes{\nLock}{\stRWLock{\nLock}{i}} \equiv \zero
\qquad
\pRes{\nLock}{\stRWLockL{\nLock}{i}} \equiv \zero
\qquad
\pRes{\nLock}{\stRWLockW{\nLock}{i}} \equiv \zero
\\[1mm]
\pRes{\nGeneric}{\pRes{\nGenerici}{\procP}} \equiv 
  \pRes{\nGenerici}{\pRes{\nGeneric}{\procP}}
\qquad
\pPar{\procP}{\pRes{\nGeneric}{\procQ}} \equiv 
  \pRes{\nGeneric}{\pPPar{\procP}{\procQ}}
\ \mbox{\scriptsize $(\nGeneric\not\in \fn{\procP})$}
\end{array}
$
\end{tabular}
}
\end{center}
\vspace{-3mm}
\caption{Syntax of the Process language (top) and Structural Congruence for Stores (bottom).}\label{fig:syntax}
\vspace{-0.5cm}
\end{figure}

The syntax of the calculus, together with the
standard structural congruence $\procP\equiv \procPi$ (which includes
$\equiv_\alpha$), is given in \figurename~\ref{fig:syntax}, where   
$\nExpr,\nExpri$ range over {\em expressions}, 
$\nVar,\nVary$ over {\em variables}, $\nLock,\nLocki$ over {\em locks},
  $\nGeneric,\nGenerici$ over {\em identifiers} (either shared variables
  or locks) and $\nVal$ over {\em values} 
  (either local variables, natural numbers or booleans). 
We write $\vExpr$, $\vVal$, $\vVar$ and $\vGeneric$ for a
list of expressions, values, variables and names respectively, 
and use $\cdot$ as the concatenation operator.

Process syntax ($\procP,\procQ,\procR,...$) is given as follows.  
The \emph{prefix} $\prefMem\pCont\procP$ contains 
either (1) a \emph{silent action} $\pSilent$; (2) a 
{\em store action} of $\nExpr$ in $\vVar$, $\pStore{\nVar}{\nExpr}$; 
(3) a {\em load action} of $\nVar$, bound to $\nVary$ in the continuation, 
$\pLoad{\nVar}{\nVary}$; and (4) 
actions ($\prefLock$) for lock/unlock and read-lock/unlock on program locks 
(denoted by $\nLock$). 
  
There are three constructs for ``$\m{new}$'':
a {\em new variable} process
$\pNewMem{\nVar}{\sigma}\pCont\procP$ creates a new shared variable in the heap
with payload type $\sigma$, binding it to $\nVar$ in the continuation $\procP$; 
a {\em new lock} process
$\pNewLock{\nLock}\pCont\procP$ creates a new program lock and $\pNewRWLock{\nLock}\pCont\procP$ creates 
a new program read-write lock, 
binding them to $\nLock$ in the continuation. 
The syntax includes the {\em conditional} $\pITE{\nExpr}{\procP}{\procQ}$, 
{\em parallel} process $\pPar{\procP}{\procQ}$, and 
the {\em inactive} process $\zero$ (often omitted). 

A Go program is modelled as a program $\procProg$ in \gol, written
$\pIn{\{ \procD[i] \}_{i\in I}}{\procP}$, which consists of a set of mutually recursive
process definitions which encode the goroutines and functions
used in the program, together with a process $\procP$ that encodes the
program entry point (\verb=main=).
The entry point is usually modelled as 
$\pEntryPt$, a call to a defined process $\pEntryProc$. 
The entry point is the main process in a collection of
mutually recursive process definitions (ranged over by $\procD$),
parametrised by a list of (expressions and locks) variables.

Process variable $\procX$ is bound by {\em definition} $\procD$ of the form of 
$\pDef{\procX}{\vVar}{\procP}$ where $\fn{\procD} = \emptyset$. 
This is used by {\em process call} 
$\pCall{\procX}{\vExpr}{\vGeneric}$ which denotes an instance of the process definition
bound to $\procX$, with formal parameters instantiated to $\vExpr$
and $\vGeneric$. Note that the entry point could take parameters, if the programmer 
wants the program to depend on user input data for example, but our examples 
never make use of that capability.

The part of the syntax denoted by the stores is 
\emph{runtime} constructs which are generated during the execution 
(i.e.~not written by the programmer and appearing as 
standalone parallel terms): 
a shared variable store 
$\stMem{\nVar}{\sigma}{\nVal}$ contains message $\nVal$ of type $\sigma$; and  
we represent five internal states of lock stores, situated on the 
last line of the left column, where the index $i$ 
is used for rwlocks and the superscripts $\star$ and
$\blacktriangledown$ respectively denote locked and waiting locks.
{\em Restriction} $\pRes{\nGeneric}{\procP}$
denotes the runtime handle $\nGeneric$ for a lock or shared variable bound in
$\procP$, and thus hidden from external processes.

\full{Finally, the notation $\fn{\procP}$ denotes the sets of free names 
(locks, shared variables, local variables), ie. ones that have not been bound 
by a restriction operator $\pRes{\nGeneric}{}$, a definition $\procD$, 
a ``$\m{new}$'' construct, or a load action,
cf. \figurename~\ref{fig:names} in Appendix~\ref{app:freename}.}
{Finally, the notation $\fn{\procP}$ denotes the sets of free names 
(locks, shared variables, local variables), ie. ones that have not been bound 
by a restriction operator $\pRes{\nGeneric}{}$, a definition $\procD$, 
a ``$\m{new}$'' construct, or a load action.}

\begin{example}[Processes from \figurename~\ref{fig:running-ex-rwmut-race} 
and \figurename~\ref{fig:running-ex-rwmut-safe}]
\label{ex:syntax}
The following process represents the code in \figurename~\ref{fig:running-ex-rwmut-race}. We first separate the \lstgo{main} function in two parts: the part that 
instantiates the variable and lock, and spawns the side process in parallel to the 
continuation, that we call $\pEntryProc$; and the rest that processes in 
parallel to the second goroutine that we put in a separate process $\procP$. 
Process $\procQ$ is the representation of function \lstgo{f}, that is run in 
the second goroutine.
\[
\procProg[\m{race}] \coloneqq
    \pIn
	{\left\{\begin{array}{@{}l@{\ }l@{\ }l}
	  \pEntryProc & = & \pNewMem{\nVar}{\m{int}}\pCont
	   \pNewRWLock{\nLock}\pCont
	   \pPPar{\pCall{\procP}{\nVar,\nLock}{}}
	         {\pCall{\procQ}{\nVar,\nLock}{}}\\
	  \procP(\nVary,\nVarz) & = & \pRLock{\nVarz}\pCont
	   \pLoad{\nVary}{\nTmp[1]}\pCont
	   \pStore{\nVary}{\nTmp[1]+10}\pCont
	   \pRUnlock{\nVarz}\pCont\\
	  & & \pRLock{\nVarz}\pCont
	   \pLoad{\nVary}{\nTmp[2]}\pCont
	   \pSilent\pCont
	   \pRUnlock{\nVarz}\pCont
	   \zero\\
          \procQ(\nVary,\nVarz) & = & \pRLock{\nVarz}\pCont
	   \pLoad{\nVary}{\nTmp[0]}\pCont
	   \pStore{\nVary}{\nTmp[0]+20}\pCont
	   \pRUnlock{\nVarz}\pCont
	\zero
	\end{array}\right\}}
     {\pEntryPt}
\]

The next process represents the code in \figurename~\ref{fig:running-ex-rwmut-safe} 
in the same fashion as above. 
\[
\procProg[\m{safe}] \coloneqq
    \pIn
	{\left\{\begin{array}{@{}l@{\ }l@{\ }l}
	  \pEntryProc & = & \pNewMem{\nVar}{\m{int}}\pCont
	   \pNewRWLock{\nLock}\pCont
	   \pPPar{\pCall{\procP}{\nVar,\nLock}{}}
	         {\pCall{\procQ}{\nVar,\nLock}{}}\\
	  \procP(\nVary,\nVarz) & = & \pLock{\nVarz}\pCont
	   \pLoad{\nVary}{\nTmp[1]}\pCont
	   \pStore{\nVary}{\nTmp[1]+10}\pCont
	   \pUnlock{\nVarz}\pCont\\
	  & & \pRLock{\nVarz}\pCont
	   \pLoad{\nVary}{\nTmp[2]}\pCont
	   \pSilent\pCont
	   \pRUnlock{\nVarz}\pCont
	   \zero\\
          \procQ(\nVary,\nVarz) & = & \pLock{\nVarz}\pCont
	   \pLoad{\nVary}{\nTmp[0]}\pCont
	   \pStore{\nVary}{\nTmp[0]+20}\pCont
	   \pUnlock{\nVarz}\pCont
	\zero
	\end{array}\right\}}
     {\pEntryPt}
\]
\end{example}

\subsection{Operational Semantics}\label{sect:core_sem}
The semantics of \gol is given by the labelled transition system (LTS)
shown in \figurename~\ref{fig:redsem}. The LTS
system enables us to give a simple and uniform 
definition of barbs in Definition~\ref{def:barb}
and a formal correspondence with the modal
$\mu$-calculus described in \S~\ref{sect:form-hb-modmu}. 
The LTS rules are 
written $\procP \tra{\actgen} \procPi$,
where $\actgen$ is a label of the form:
\begin{center}
\framebox{$
\begin{array}{@{}r@{\ }c@{\ }l@{}}
\actgen & \coloneqq & \actname[\m{l}] \, \mid \,
\dataAct{\actname[\m{m}]}{\nExpr}
\hfill
\occurgen \coloneqq  \occurnone
\, \mid \, \occurleft{\occurgen}
\, \mid \, \occurright{\occurgen}
\hfill
\actname[\m{m}] \ \coloneqq \labload{\nVar}
\, \mid \, \labloaddual{\nVar}
\, \mid \, \labtystore{\nVar}{\occurgen}
\, \mid \, \labstoredual{\nVar}
\\[1mm]
\actname[\m{l}] & \coloneqq & \lablock{\nLock}
\, \mid \, \labunlock{\nLock}
\, \mid \, \labrlock{\nLock}
\, \mid \, \labrunlock{\nLock}
\, \mid \, \labmut{\nLock}
\, \mid \, \lablckmut{\nLock}
\, \mid \, \labrmut{\nLock}{}
\, \mid \, \labwaitrmut{\nLock}{}
\, \mid \, \labwaitmut{\nLock}{}
\, \mid \, \labsync{\nGeneric}
\, \mid \, \acttau
\qquad
\actname\ \coloneqq \actname[\m{m}] \, \mid \, \actname[\m{l}] 
\end{array}
$}
\end{center}
They can be either a \emph{data-dependent action} $\actname[\m{m}]$ 
along with its data $\nExpr$, used for synchronisation purposes on actions
that transmit data, 
or a \emph{data-independent action} $\actname[\m{l}]$
alone, used for synchronisation on actions that do not transmit meaningful data, 
and for the synchronisations $\labsync{\nGeneric}$ and silent action $\acttau$. 

The actions in $\actname[\m{m}]$ define $\ldbarb{\nVar}{}$ (read), 
$\stbarb{\nVar}{\occurgen}$ (write), 
$\labloaddual{\nVar}$ and $\labstoredual{\nVar}$ (dual actions)
of a shared variable $\nVar$, where $\occurgen$ 
denotes an \emph{occurrence} (a position in the 
parallel composition) that is a string of 1s, 2s and $\occurnone$. 
The actions in $\actname[\m{l}]$ define (1) 
$\lckbarb{\nLock}$ (lock), 
$\ulckbarb{\nLock}$ (unlock), 
$\rlckbarb{\nLock}$ (read-lock) and 
$\rulckbarb{\nLock}$ (read-unlock); 
(2) lock store actions, $\mutbarb{\nLock}$, 
$\lckmutbarb{\nLock}$, 
$\rmutbarb{\nLock}{i}$, $\waitrmutbarb{\nLock}{i}$ and $\labwaitmut{\nLock}$ 
(whose purpose is to interact with each action in (1) to produce the lock 
synchronisation $\labsync{\nLock}$); as well as 
(3) synchronisations $\labsync{\nGeneric}$ and silent actions. 

\begin{remark}\rm
{\bf (1)} The write action $\stbarb{\nVar}{\occurgen}$ 
uses occurrence $\occurgen$ 
to denote the position of the thread 
which contains that action. 
By using occurrences, we can differentiate two writes 
on the same variable happening at the same time, and 
thereby formally define the notion 
of data race (see~Definition~\ref{def:hb-drace}); 
and {\bf (2)} one lock store can produce several different 
actions which then produce lock synchronisation $\labsync{\nLock}$
with different lock primitives. 
This allows us to implement the properties 
with mCRL2 straightforwardly, cf.~\S~\ref{sec:implementation}. 
\end{remark}

We also define the general label $\actname$ for actions,  
which only contains action markers and no data, 
and will be of use for data-independent marking later on, such as barbs.
\emph{Occurences} are ranged over by $\occurgen$, $\occurgen'$,..., 
where $\occurnone$ denotes the empty occurrence, 
while $\occurleft{\occurgen}$ (resp. $\occurright{\occurgen}$) 
denotes the left (resp. right) shift of of $\occurgen$. 
The left and right shifting operators on action $\actgen$, 
$\actleft{\actgen}$ and $\actright{\actgen}$, are defined as:  
\[ 
\actleft{\dataAct{\labtystore{\nVar}{\occurgen}}{\nExpr}} = 
\dataAct{\labtystore{\nVar}{\occurleft{\occurgen}}}{\nExpr}\quad 
\mbox{and} \quad 
\actright{\dataAct{\labtystore{\nVar}{\occurgen}}{\nExpr}} = 
\dataAct{\labtystore{\nVar}{\occurright{\occurgen}}}{\nExpr}
\]
with $\actleft{\actgen} = \actright{\actgen} = \actgen$  
if $\actgen\neq\dataAct{\labtystore{\nVar}{\occurgen}}{\nExpr}$. 
Example \ref{ex:occurrence} will explain the use of these operators with
the LTS rules. 

\begin{figure}[h]
\vspace{-3mm}
\begin{center}
\framebox{
\begin{tabular}{l}
\begin{tabular}{l|l}
$
\!\!\!\!\!\!\!\!\!\!\!
\begin{array}{l}
  \framebox{\text{Lock and Memory actions}}
  \\
  \begin{array}{@{}l@{\ }l@{}}
  \rulename{\ruleLockLockReq}&
   {\pLock{\nLock}\pCont\procP\tra{\lablock{\nLock}}\procP}
  \\
  \rulename{\ruleLockUnlockReq}&
   {\pUnlock{\nLock}\pCont\procP\tra{\labunlock{\nLock}}\procP}
  \\
   \rulename{\ruleRWLockRLockReq}&
    {\pRLock{\nLock}\pCont\procP\tra{\labrlock{\nLock}}\procP}
  \\
   \rulename{\ruleRWLockRUnlockReq}&
    {\pRUnlock{\nLock}\pCont\procP\tra{\labrunlock{\nLock}}\procP}
  \\
  \rulename{\ruleMemLoadReq}&
   {\pLoad{\nVar}{\nVary}\pCont\procP \tra{\dataAct{\labload{\nVar}}{\nVal}} 
       \procP\subs{\nVal}{\nVary}}
  \\
  \rulename{\ruleMemStoreReq}&
   {\pStore{\nVar}{\nExpr}\pCont\procP \tra{\dataAct{\labtystore{\nVar}{\occurnone}}{\nExpr}} \procP}
  \\
  \hline
   \rulename{\ruleLockLockAck}&
    {\stLock{\nLock}\tra{\labmut{\nLock}}\stLockL{\nLock}}
  \\
   \rulename{\ruleLockUnlockAck}&
    {\stLockL{\nLock}\tra{\lablckmut{\nLock}}\stLock{\nLock}}
  \\
   \rulename{\ruleRWLockLockAck}&
    {\stRWLockW{\nLock}{0}\tra{\labmut{\nLock}}\stRWLockL{\nLock}{0}}
  \\
   \rulename{\ruleRWLockUnlockAck}&
    {\stRWLockL{\nLock}{0}\tra{\lablckmut{\nLock}}\stRWLock{\nLock}{0}}
  \\
   \rulename{\ruleRWLockRLockAck}&
    {\stRWLock{\nLock}{i}\tra{\labrmut{\nLock}{}}\stRWLock{\nLock}{i+1}}
  \\
  \rulename{\ruleRWLockRUnlockAck}& 
    {\stRWLock{\nLock}{i+1}\tra{\labwaitrmut{\nLock}{}}\stRWLock{\nLock}{i}}
  \\
   \rulename{\ruleRWLockLockStageAck}&
    {\stRWLock{\nLock}{i}\tra{\labwaitmut{\nLock}}\stRWLockW{\nLock}{i}}
  \\
  \rulename{\ruleRWLockStagedRUnlockAck}& 
           {\stRWLockW{\nLock}{i+1}\tra{\labwaitrmut{\nLock}{}}\stRWLockW{\nLock}{i}}
  \\
  \rulename{\ruleMemLoadAck}&
   {\stMem{\nVar}{\sigma}{\nVal}\tra{\dataAct{\labloaddual{\nVar}}{\nVal}}\stMem{\nVar}{\sigma}{\nVal}}
  \\
  \rulename{\ruleMemStoreAck}&
   {\stMem{\nVar}{\sigma}{\nVal}\tra{\dataAct{\labstoredual{\nVar}}{\nVali}}\stMem{\nVar}{\sigma}{\nVali}}
  \\[.2em]
  \end{array}
  \end{array}
$
&\
$
\begin{array}{@{}l@{}}
  \framebox{\text{Synchronisation rules}}
  \\[2mm]
  \begin{array}{c}
  \rulename{\ruleMemLoad}
   \inferrule
    {\procP \tra{\dataAct{\labload{\nVar}}{\vVal}} \procPi \quad 
     \procQ \tra{\dataAct{\labloaddual{\nVar}}{\nVal}} \procQi}
    {\pPar{\procP}{\procQ} \tra{\labsync{\nVar}} \pPar{\procPi}{\procQ}}
  \\
  \rulename{\ruleMemStore}
   \inferrule
    {\procP \tra{\dataAct{\labtystore{\nVar}{\occurgen}}{\nExpr}} \procPi \quad
     \procQ \tra{\dataAct{\labstoredual{\nVar}}{\nVal} }\procQi \quad
     \nExpr\downarrow\nVal}
    {\pPar{\procP}{\procQ} \tra{\labsync{\nVar}} \pPar{\procPi}{\procQi}}
  \\
  \rulename{\ruleLockLock}
   \inferrule
    {\procP \tra{\lablock{\nLock}} \procPi \quad 
     \procQ \tra{\labmut{\nLock}} \procQi}
    {\pPar{\procP}{\procQ} \tra{\labsync{\nLock}} \pPar{\procPi}{\procQi}}
  \\
  \rulename{\ruleLockUnlock}
   \inferrule
    {\procP \tra{\labunlock{\nLock}} \procPi \quad 
     \procQ \tra{\lablckmut{\nLock}} \procQi}
    {\pPar{\procP}{\procQ} \tra{\labsync{\nLock}} \pPar{\procPi}{\procQi}}
  \\
  \rulename{\ruleRWLockRLock}
   \inferrule
    {\procP \tra{\labrlock{\nLock}} \procPi \quad 
     \procQ \tra{\labrmut{\nLock}{}} \procQi}
    {\pPar{\procP}{\procQ} \tra{\labsync{\nLock}} \pPar{\procPi}{\procQi}}
  \\
  \rulename{\ruleRWLockRUnlock}
   \inferrule
    {\procP \tra{\labrunlock{\nLock}} \procPi \quad 
     \procQ \tra{\labwaitrmut{\nLock}{}} \procQi}
    {\pPar{\procP}{\procQ} \tra{\labsync{\nLock}} \pPar{\procPi}{\procQi}}
  \\
  \rulename{\ruleRWLockLockStage}
   \inferrule
    {\procP \tra{\lablock{\nLock}} \procPi \quad 
     \procQ \tra{\labwaitmut{\nLock}} \procQi}
    {\pPar{\procP}{\procQ} \tra{\acttau} \pPar{\procP}{\procQi}}
  \\[1em]
  \rulename{\ruleSilentAction}\
   \pSilent\pCont\procP \tra{\acttau} \procP
  \end{array}
  \end{array}\!\!\!
$
\end{tabular}\hspace{-1em}
 \\
 \hline
 \\[-0.9em]
 $\!\!\!\!\!
 \begin{array}{l}
 \!\!\!\begin{array}{@{}lr@{\ }l@{}}
 \framebox{Runtime structures creation}
 &
 \rulename{\ruleNewMem}
 &
    {\pNewMem{\nVary}{\sigma}\pCont\procP \tra{\acttau}
      \pRes{\nVary}{\pPPar{\procP}{\stMem{\nVary}{\sigma}{\bot}}}}
 \\[0.4em]
 \rulename{\ruleNewMut}\ 
    {\pNewLock{\nLock}\pCont\procP \tra{\acttau}
     \pRes{\nLock}{\pPPar{\procP}{\stLock{\nLock}}}}
 \hspace{4.5em}~
 &
 \rulename{\ruleNewRMut}
 &
    {\pNewRWLock{\nLock}\pCont\procP \tra{\acttau}
     \pRes{\nLock}{\pPPar{\procP}{\stRWLock{\nLock}{0}}}}
 \\[0.4em]
 \end{array}
 \\[0.6em]
 \hline
 \\[-0.7em]
 \rulename{\ruleParLeft}
  \inferrule
   {\procP \tra{\actgen} \procPi}
   {\pPar{\procP}{\procQ} 
     \tra{\actleft{\actgen}} \pPar{\procPi}{\procQ}}
 \hfill
 \rulename{\ruleParRight}
  \inferrule
   {\procQ \tra{\actgen} \procQi}
   {\pPar{\procP}{\procQ} 
     \tra{\actright{\actgen}} \pPar{\procP}{\procQi}}
 \hfill
 \rulename{\ruleRestrictFree}
  \inferrule
   {\procP \tra{\actgen} \procPi \quad \nGeneric\notin \fn{\actgen}}
   {\pRes{\nGeneric}{\procP} \tra{\actgen}
    \pRes{\nGeneric}{\procPi}}
 \\
 \rulename{\ruleITETrue}
  \inferrule
   {\nExpr \conv \true}
   {\pITECont{\nExpr}{\procP}{\procQ}{\procR}\tra{\acttau}
    \procP}
 \hfill
 \rulename{\ruleITEFalse}
  \inferrule
   {\nExpr \conv \false}
   {\pITECont{\nExpr}{\procP}{\procQ}{\procR}\tra{\acttau}
    \procQ}
 \hfill
 \rulename{\ruleRestrictBind}
  \inferrule
   {\procP \tra{\labsync{\nGeneric}} \procPi}
   {\pRes{\nGeneric}{\procP} \tra{\acttau}
    \pRes{\nGeneric}{\procPi}}
 \\
 \rulename{\ruleDefinition}
  \inferrule
   {\pPar{\procP\subs{\vVal,\vGeneric}{\vVar}}{\procQ} \tra{\actgen}
    \procR \quad \nExpr[i] \conv \nVal[i]
    \quad \procX(\vVar) = \procP \in \{ \procD[i] \}_{i\in I}}
   {\pPar{\pCall{\procX}{\vExpr}{\vGeneric}}{\procQ}
    \tra{\actgen} \procR}
 \hfill
 \rulename{\ruleAlphaEquiv}
  \inferrule
   {\procP \equiv_\alpha \procPi \quad \procPi\tra{\actgenb}\procPii}
   {\procP \tra{\actgenb} \procPii}
 \end{array}\hspace*{-2.5mm}
$
\end{tabular}
}
\end{center}
\vspace{-3mm}
\caption{LTS Reduction Semantics for the Processes.}\label{fig:redsem}
\vspace{-3mm}
\end{figure}

This LTS defines the semantics of shared variables, 
locks, and read-write locks which closely follow the specifications 
in \cite{web:golang}. 
We first highlight the operational semantics of locks 
from \cite{web:golangmutex} and rwlocks from 
\cite{web:golangrwmutex}. 
A \textbf{\emph{lock}} is a mutual exclusion lock. 
It must not be copied after its first use: 
a lock $\nLock$ is created by \rulename{\ruleNewMut}, which is guaranteed 
fresh by the ``$\pRes{\nLock}{}$'' operation. It is then 
locked by \rulename{\ruleMutLock} and unlocked by \rulename{\ruleMutUnlock}. 
A \textbf{\emph{read-write lock}} (rwlock)
is a reader/writer mutual exclusion lock. 
The lock can be held by an arbitrary number of readers or a single writer. 
The zero value for a rwlock is an unlocked state. 
If a goroutine holds a rwlock for reading and another goroutine 
calls \lstgo{Lock}, no goroutine should expect to be able to acquire a read-lock 
until both the initial read-lock and the staged \lstgo{Lock} call are released. 
This is to ensure that the lock eventually becomes 
available to writers; a blocked \lstgo{Lock} call excludes 
new readers from acquiring the 
lock. To model this situation, we annotate a freshly created 
rwlock by the counter $i$ (instanciated at 0 by \rulename{\ruleNewRMut}); 
this counter is incremented by any fired read-lock 
(by \rulename{\ruleRMutRLock}), and 
blocked from increasing if a \lstgo{Lock} action gets staged (by 
\rulename{\ruleRMutLockStage}, 
{\bf note how the \lstgo{Lock} action is \emph{not} consumed by this rule}); 
then it is unlocked by read-unlock calls (by \rulename{\ruleRMutRUnlock}) until 
the pending number 
of read-locks becomes $0$, and finally write-locked (by \rulename{\ruleRMutLock}) 
and further 
unlocked by the corresponding unlock (by \rulename{\ruleRMutUnlock}), 
if a \lstgo{Lock} was previously staged by \rulename{\ruleRMutLockStage}. 

A \textbf{\emph{shared variable}} 
is implemented at runtime by a 
named area in the store, which stores a value of its 
payload data type, and that can be written to or read by any 
process within its scope. 
It is created by \rulename{\ruleNewMem} with 
an initial value for declared type $\sigma$ (0 for $\m{int}$, 
$\false$ for $\m{bool}$, etc.), 
accessed for reading by \rulename{\ruleMemLoad} and 
for writing by \rulename{\ruleMemStore}. 

The \rulename{\rulePar-$\ast$} rules are explained in Example~\ref{ex:occurrence} below.

\begin{example}[Occurrences]\label{ex:occurrence}
Let $\procP = \pStore{\nVar}{\nExpr}\pCont\procPi$, 
$\procQ = \pStore{\nVar}{\nExpri}\pCont\procQi$ and 
$\procR = \pLoad{\nVar}{\nVarz}\pCont\procRi$.
It follows $\procP\tra{\dataAct{\labtystore{\nVar}{\occurnone}}{\nExpr}}\procPi$, 
$\procQ\tra{\dataAct{\labtystore{\nVar}{\occurnone}}{\nExpri}}\procQi$ and 
$\procR\tra{\dataAct{\labtyload{\nVar}}{\nVal}}\procRi\subs{\nVal}{\nVarz}$.
\\
If we compose $\procP$ and $\procQ$, we use \rulename{\ruleParLeft} and \rulename{\ruleParRight} to determine the new reductions:
\full{
}{
  \vspace{-.5mm}
}
\[
\begin{array}{r@{\ }c@{\ }l@{\hspace{3cm}}l@{\ }l}
  \pPar{\procP}{\procQ} & 
    \tra{\dataAct{\labtystore{\nVar}{\occurleft{\occurnone}}}{\nExpr}} & 
    \pPar{\procPi}{\procQ} &
    \actleft{\dataAct{\labtystore{\nVar}{\occurnone}}{\nExpr}} & = 
      \dataAct{\labtystore{\nVar}{\occurleft{\occurnone}}}{\nExpr}\\
  \pPar{\procP}{\procQ} & 
    \tra{\dataAct{\labtystore{\nVar}{\occurright{\occurnone}}}{\nExpri}} & 
    \pPar{\procP}{\procQi} &
    \actright{\dataAct{\labtystore{\nVar}{\occurnone}}{\nExpri}} & = 
      \dataAct{\labtystore{\nVar}{\occurright{\occurnone}}}{\nExpri}
\end{array}
\]
Composing again, with $\procR$:
\[
\begin{array}{r@{\ }c@{\ }l@{\hspace{1.15cm}}l@{\ }l}
  \pPar{\pPPar{\procP}{\procQ}}{\procR} & 
    \tra{\dataAct{\labtystore{\nVar}{\occurleft{\occurleft{\occurnone}}}}{\nExpr}} & 
    \pPar{\pPPar{\procPi}{\procQ}}{\procR} &
    \actleft{\dataAct{\labtystore{\nVar}{\occurleft{\occurnone}}}{\nExpr}} & = 
      \dataAct{\labtystore{\nVar}{\occurleft{\occurleft{\occurnone}}}}{\nExpr}\\
  \pPar{\pPPar{\procP}{\procQ}}{\procR} & 
    \tra{\dataAct{\labtystore{\nVar}{\occurleft{\occurright{\occurnone}}}}{\nExpri}} & 
    \pPar{\pPPar{\procP}{\procQi}}{\procR} &
    \actleft{\dataAct{\labtystore{\nVar}{\occurright{\occurnone}}}{\nExpri}} & = 
      \dataAct{\labtystore{\nVar}{\occurleft{\occurright{\occurnone}}}}{\nExpri}\\
  \pPar{\pPPar{\procP}{\procQ}}{\procR} & 
    \tra{\dataAct{\labtyload{\nVar}}{\nVal}} & 
    \pPar{\pPPar{\procP}{\procQ}}{\procRi\subs{\nVal}{\nVarz}} &
    \actright{\dataAct{\labtyload{\nVar}}{\nVal}} & = 
      \dataAct{\labtyload{\nVar}}{\nVal}
\end{array}
\]
\end{example}

For process definitions, 
we implicitly assume the existence of an ambient set of definitions 
$\{\procD[i] \}_{i\in I}$.
Rule \rulename{\ruleDefinition} replaces $\procX$ by the corresponding process 
definition (according to the underlying definition environment), 
instantiating the parameters accordingly. The remaining rules are 
standard from process calculus literature \cite{SangiorgiD:picatomp}. 
We define $\tra{}$ as $\equiv\tra{\acttau}\equiv \cup
\equiv\tra{\labsync{\name}}\equiv$. 

We define a 
{\em normal form} for terms, which is used later in \S~\ref{sect:form-hb-modmu}: 
\begin{definition}[Normal Form]\label{def:migo-normal-form}\rm
A term $\procP$ is \emph{in normal form} if 
$\procP = \pRes{\nameVec}{\procPi}$ and 
$\procPi\not\equiv\pRes{\nGeneric}{\procPii}$.
\end{definition}
We note that, with structural congruence, every well-formed term can be 
transformed to normal form, and we can then study reduction up to normal form, in order 
to witness synchronisation actions on channels, memory and mutex.

\section{Defining Safety and Liveness: Data Race and Happens-Before}\label{sub:live-safe-process}
We define the properties of data race freedom 
and lock safety/liveness through \emph{barbs} 
(\S~\ref{subsec:barb}). 
A {\bf\emph{data race}} happens 
when two writers (or a reader and a writer) can concurrently access the same 
shared variable at the same time. 
{\bf\emph{Unsafe lock access}} happens if (1) unlock 
happens before lock happens or before waiting read-unlocks release 
the lock; or (2) read-unlock happens before read-lock happens or
after a lock call accesses the process lock. 
{\bf\emph{Lock liveness}} identifies the
ability of (read-)lock requests to always eventually fire.     
Our first main result is a formalisation of the \emph{happens-before}
relation and other properties specified in the Go memory model 
\cite{web:go-mem-model}
and a correspondence between a data race characterisation 
through the happens before relation and 
another characterisation of a data race through barbs. 

\subsection{Safety and Liveness Properties through Barbs}
\label{subsec:barb}
We first define barbed process predicates \cite{MilnerR:barbis}
introducing predicates for locks and shared variable accesses. 
The predicate $\procP\barb{\actname}$ means that $\procP$ 
immediately offers a visible action $\actname$.
\begin{definition}[Process barbs]\label{def:barb}\rm
The barbs are defined as follows:
\begin{description}
\item[Prefix Actions:]

  $
\begin{array}{lll}
  \pStore{\nVar}{\nExpr}\barb{\stbarb{\nVar}{\occurnone}};
  &
  \pLoad{\nVar}{\nVary}\barb{\ldbarb{\nVar}{\occurnone}};
&
\pLock{\nLock}\barb{\lckbarb{\nLock}};
\\ [1mm]
\pUnlock{\nLock}\barb{\ulckbarb{\nLock}};
&
\pRLock{\nLock}\barb{\rlckbarb{\nLock}};
&
\pRUnlock{\nLock}\barb{\rulckbarb{\nLock}}
\end{array}
$
 \item[Programs:]
 if $\procP\tra{\actname,\nExpr}\procPi$ where $\actname=\actname[\m{m}]$ is an action over
  a shared variable, or $\procP\tra{\actname}\procPi$ where $\actname=\actname[\m{l}]$ is
  $\labsync{\nGeneric}$ or a lock action, then $\procP\barb{\actname}$.
\end{description}
\end{definition}

Actions in this case are the same ones as defined before in the operational 
semantics of \gol, expect for silent action $\acttau$. We write 
${\procP}\wbarb{\actname}$ if ${\procP}\tra{}^* {\procPi}$ and 
${\procPi}\barb{\actname}$.  

We first define a safety property for locks in Definition~\ref{def:safe}.

\begin{definition}[Safety]~\label{def:safe}\rm
Program $\migoProg$ 
is {\em safe} if for all $\procP$ such that 
$\migoProg\tra{}^*\pRes{\vGeneric}{\procP}$, 
(a) if $\procP\barb{\ulckbarb{\nLock}}$
then $\procP\barb{\lckmutbarb{\nLock}}$; and 
(b) if $\procP\barb{\rulckbarb{\nLock}}$
then $\procP\barb{\waitrmutbarb{\nLock}{i}}$.
\end{definition}

Safety states that 
in all reachable program
states, 
the unlock action will happen only if the process lock 
is already locked by the lock action; 
and the read-unlock will happen only if the process lock is locked 
by the read-lock action. 

Next we define the liveness property:
all (read-)lock requests will always eventually fire 
(i.e.~perform a synchronisation).  

\begin{definition}[Liveness]~\label{def:live}\rm
Program $\procProg$ is \emph{live} if for all $\procP$ such that
$\procProg \tra{}^* \pRes{\vGeneric}{\procP}$,
if ${\procP}\barb{\lckbarb{\nLock}}$ or 
${\procP}\barb{\rlckbarb{\nLock}}$ 
then 
$\procP\wbarb{\labsync{\nLock}}$.
\end{definition}

\subsection{Happens Before and Data Race}
\label{subsec:happenbefore}
We now define the happens-before relation,
closely following  
\cite{web:go-mem-model}, and investigate its relationship with 
data races. The {\em happens-before} relation between actions
$\actname$ and $\actnamei$, denoted by  
$\hb[\procP]{\actname}{\actnamei}$, 
is defined in \figurename~\ref{fig:hbrel}. 
\begin{figure}
\full{
  \vspace{-5mm}
}{
  \vspace{-7mm}
}
\begin{center}
\framebox{$
\hspace*{0.4em}
\begin{array}{c}
\hbrulename{\ruleContinue}
\inferrule
{\prefMem\barb{\actname}\quad \procP\barb{\actnamei}}
{\hb[\prefMem\pCont\procP]{\actname}{\actnamei}}
\quad\ 
\hbrulename{\ruleTransitivity}
\inferrule
{\hb[\migoP]{\actname}{\actnamei}\quad \hb[\migoP]{\actnamei}{\actnameii}}
{\hb[\migoP]{\actname}{\actnameii}}
\quad\ 
\hbrulename{\ruleReduction}
\inferrule
{\procP \tra{}^{\ast} \procPi \quad \hb[\procPi]{\actname}{\actnamei}}
{\hb[\procP]{\actname}{\actnamei}}
\\[1em]
\!\begin{array}{cc}
\hbrulename{\ruleParLeft}
\inferrule
{\hb[\migoP]{\actname}{\actnamei}}
{\hb[\migoP\parr \migoQ]{\actleft{\actname}}{\actleft{\actnamei}}}
&
\hbrulename{\ruleParRight}
\inferrule
{\hb[\migoQ]{\actname}{\actnamei}}
{\hb[\migoP\parr \migoQ]{\actright{\actname}}{\actright{\actnamei}}}
\\[1.2em]
\hbrulename{\ruleUnlockLock}
\inferrule
{\migoP\barb{\lckbarb{\nLock}} \quad \migoP\barb{\ulckbarb{\nLock}}}
{\hb[\migoP]{\ulckbarb{\nLock}}{\lckbarb{\nLock}}}
&
\hbrulename{\ruleRUnlockLock}
\inferrule
{\migoP\barb{\lckbarb{\nLock}} \quad \migoP\barb{\rulckbarb{\nLock}}}
{\hb[\migoP]{\rulckbarb{\nLock}}{\lckbarb{\nLock}}}
\\[1.2em]
\hbrulename{\ruleUnlockRLock}
\inferrule
{\migoP\barb{\rlckbarb{\nLock}} \quad \migoP\barb{\ulckbarb{\nLock}}}
{\hb[\migoP]{\ulckbarb{\nLock}}{\rlckbarb{\nLock}}}
&
\hbrulename{\ruleLockRLock}
\inferrule
{\migoP\barb{\rlckbarb{\nLock}} \quad \migoP\barb{\lckbarb{\nLock}}}
{\hb[\migoP]{\lckbarb{\nLock}}{\rlckbarb{\nLock}}}
\\[1.2em]
\hbrulename{\ruleRestrict}
\inferrule
{\hb[\migoP]{\actname}{\actnamei}\quad \name \not\in \fn{\actname}\cup\fn{\actnamei}}
{\hb[\migoRes{\name}{\migoP}]{\actname}{\actnamei}}
&
\hbrulename{\ruleAlphaEquiv}
\inferrule
{\hb[\migoP]{\actname}{\actnamei}\quad \migoP\equiv_{\alpha} \migoQ}
{\hb[\migoQ]{\actname}{\actnamei}}
\end{array}
\end{array}
$}
\\[1mm]
We omit the symmetric rules for most rules ending in a parallel process
$\migoP\parr \migoQ$.
\end{center}
\full{
  \vspace{-3mm}
}{
  \vspace{-4.5mm}
}
\caption{Happens-Before Relation}\label{fig:hbrel}
\full{
  \vspace{-5mm}
}{
  \vspace{-5.56mm}
}
\end{figure}

It is a binary relation which is transitive, non-reflexive 
and non-symmetric, where  
$\actname,\actnamei \in \{
\stbarb{\nVar}{\occurgen},
\labload{\nVar},
\lablock{\nLock},
\labunlock{\nLock},
\labrlock{\nLock},
\labrunlock{\nLock}\}$. The operation $\actleft{\actname}$ denotes that  
occurrence $\occurgen$ in $\actname$ changes to $\occurleft{\occurgen}$, defined 
as before by 
$\actleft{\stbarb{\nVar}{\occurgen}}=\stbarb{\nVar}{\occurleft{}}$; 
otherwise $\actleft{\actname}=\actname$. The rules 
follow the specification in \cite{web:go-mem-model}.  

Rule \hbrulename{\ruleContinue} specifies that  
within a single goroutine, the happens-before order is 
the order expressed by the program. Rule \hbrulename{\ruleReduction} 
gives a form of inheritance: if $\procP$ reduces to $\procPi$ and $\procPi$ 
has an order between two actions, then $\procP$ accepts this order as valid as well, 
as it is a possible future. However, if $\hb[\procP]{\actname}{\actnamei}$, it does 
not necessarily hold for all of $\procP$'s reductions.

Rule \hbrulename{\ruleParLeft} replaces $\stbarb{\nVar}{\occurgen}$ with 
$\stbarb{\nVar}{\occurleft{}}$ if $\actname$ or $\actnamei$ is a write action. 
Rule \hbrulename{\ruleParRight} is symmetric. 
Rules \hbrulename{\ruleUnlockLock}, \hbrulename{\ruleRUnlockLock}, 
\hbrulename{\ruleUnlockRLock} and \hbrulename{\ruleLockRLock} specify the ordering 
between (read)locks and (read)unlocks, following the reduction 
semantics. 

The following definition states that 
if a write action happens concurrently with another 
write action or a read action to the same variable, 
the program has a data-race. 

\begin{definition}[Data Race]~\label{def:hb-drace}\rm
Program $\procProg$ has a {\em data race} if there exist 
two distinct actions $\actname[1]\neq \actname[2]$, two distinct occurrences
$\occurgen\neq \occurgeni$, 
and $\procProg\tra{}^{\ast}\pRes{\vGeneric}{\procP}$,
with $\actname[1] = \stbarb{\nVar}{\occurgen}$
and $\actname[2] \in \{\stbarb{\nVar}{\occurgeni},\ldbarb{\nVar}{\occurgeni}\}$, such that
$\procP\wbarb{\actname[1]}$, $\procP\wbarb{\actname[2]}$, 
$\nothb[\procP]{\actname[1]}{\actname[2]}$ and $\nothb[\procP]{\actname[2]}{\actname[1]}$. 
Program $\procProg$ is {\em data race free} if it has no data race.
\end{definition}

\noindent The following theorem states that the data race defined with 
the happens-before relation coincides with the characterisation 
given by barbs.%
\full{See Appendix~\ref{app:proofs} for the proof.}{The proof is by induction, see~\cite{JY20:full}.}

\begin{theorem}[Characterisation of Data Race]~\label{the:drace}\rm
$\procProg$ has a data race if and only if 
there exists $\procP$ such that $\procProg\tra{}^{\ast}\pRes{\vGeneric}{\procP}$ with 
$\procP\barb{\actname[1]}$, $\procP\barb{\actname[2]}$,  
$\actname[1] = \stbarb{\nVar}{\occurgen}$, 
$\actname[2] \in \{\stbarb{\nVar}{\occurgeni},\ldbarb{\nVar}{\occurgeni}\}$ and 
$\occurgen\neq \occurgeni$.
\end{theorem}

\begin{example}[Processes from \figurename~\ref{fig:running-ex-rwmut-race}]
\label{ex:reduction}
We show a possible reduction of $\migoProg[\mbox{\scriptsize race}]$ in 
Example~\ref{ex:syntax} that causes the (bad) race. 
\[
\begin{array}{@{}r@{\ }c@{\ }l@{}}
\small
\procProg[\m{race}] & = &
    \pNewMem{\nVar}{\m{int}}\pCont
      \pNewRWLock{\nLock}\pCont
      \left(\begin{array}{@{}l@{\ }l}
        \pPar{
	  & \pRLock{\nLock}\pCont
	   \pLoad{\nVar}{\nTmp[1]}\pCont
	   \pStore{\nVar}{\nTmp[1]+10}\pCont
	   \pRUnlock{\nLock}\pCont\\
	  & \ \pRLock{\nLock}\pCont
	   \pLoad{\nVar}{\nTmp[2]}\pCont
	   \pSilent\pCont
	   \pRUnlock{\nLock}\pCont
	   \zero\\}{
          & \pRLock{\nLock}\pCont
	   \pLoad{\nVar}{\nTmp[0]}\pCont
	   \pStore{\nVar}{\nTmp[0]+20}\pCont
	   \pRUnlock{\nLock}\pCont
	\zero}
	\end{array}
      \right)\\
  & \tra{}^{2} & \pRes{\nVar\nLock}{
      \left(\begin{array}{@{}l@{\ }l}
        \pPar{\pPar{\pPar{
	  & \pRLock{\nLock}\pCont
	   \pLoad{\nVar}{\nTmp[1]}\pCont
	   \pStore{\nVar}{\nTmp[1]+10}\pCont
	   \pRUnlock{\nLock}\pCont
\\
	  & \ 
\pRLock{\nLock}\pCont
	   \pLoad{\nVar}{\nTmp[2]}\pCont
	   \pSilent\pCont
	   \pRUnlock{\nLock}\pCont
	   \zero\\}{
          & \pRLock{\nLock}\pCont
	   \pLoad{\nVar}{\nTmp[0]}\pCont
	   \pStore{\nVar}{\nTmp[0]+20}\pCont
	   \pRUnlock{\nLock}\pCont
	\zero}}{\stMem{\nVar}{\m{int}}{0}}}{\stRWLock{\nLock}{0}}
	\end{array}
      \right)}\\
  & \tra{}^{6} & \pRes{\nVar\nLock}{
      \left(\begin{array}{@{}l@{\ }l}
        \pPar{\pPar{\pPar{
	  & \pStore{\nVar}{10}\pCont
	   \pRUnlock{\nLock}\pCont
	   \pRLock{\nLock}\pCont
	   \pLoad{\nVar}{\nTmp[2]}\pCont
	   \pSilent\pCont
	   \pRUnlock{\nLock}\pCont
	   \zero\\}{
          & \pStore{\nVar}{20}\pCont
	   \pRUnlock{\nLock}\pCont
	\zero}}
	  {\stMem{\nVar}{\m{int}}{0}}}
	 {\stRWLock{\nLock}{2}}
	\end{array}
      \right)} = \pRes{\nVar\nLock}{\procPi}
\end{array}
\]
\label{runnig-ex-red}%
Note that the first line is obtained by rewriting using the process definition 
structure and the \rulename{\ruleDefinition} rule, that tells us the 
rewritten program and the program with calls share the same reductions.
Then we have $\procPi\barb{\stbarb{\varname}{\occurleft{\occurleft{\occurleft{\occurnone}}}}}$ and
$\procPi\barb{\stbarb{\varname}{\occurleft{\occurleft{\occurright{\occurnone}}}}}$, 
hence $\procProg[\m{race}]$ has a data race.

On the other hand, $\procProg[\m{safe}]$ is data race free, which is ensured by 
checking every reduction chain of the process for the absence of data race.
\end{example}

\section{A Behavioural Typing System for \gol}\label{sect:core_typ}
Our typing system introduces
types for locks and shared memory, 
representing the status of runtime processes 
accessing to shared variables. 
It serves as a behavioural abstraction of a valid \gol program,
where types take the form of CCS processes with name creation.

\subsection{Behavioural Types with Shared Variables and Mutexes}
\label{subsec:typesyntax}
The syntax of types ($\tyT,\tyS,...$) and the structural congruence for the types are 
given in \figurename~\ref{fig:types-syntax}.
\begin{figure}
\begin{center}
\vspace{-5mm}
\framebox{
\begin{tabular}{l}
\begin{tabular}{c|c}
$\!\!
\begin{array}{rcl}
\tyT,\tyS & \coloneqq & \prefTyMem\tyCont\tyT \ \mid \ \tyPPar{\tyT}{\tyS} 
	\ \mid \ \zero \ \mid \ \tyRes{\nGeneric}{\tyT}\\
    &  |  & \tyChoice{\tyT[i]}{i\in I} \ \mid \ \tyCall{\procX}{\vGeneric} \ \mid \ \tyNewMem{\nVar}\tyT\\
    &  |  & \tyNewLock{\nLock}{\tyT} \ \mid \ \tyNewRWLock{\nLock}\tyT\\
    &  |  & \tyMem{\nVar} \ \mid \ \tyLockU{\nLock} \ \mid \ \tyLockL{\nLock}\\
    &  |  & \tyRWLock{\nLock}{i} \ \mid \ \tyRWLockL{\nLock}{i} \ \mid \ \tyRWLockW{\nLock}{i}
\end{array}
$
& 
$\begin{array}{lcl}
\tyEqn & \coloneqq & \tyIn{\{ \tyCallee[i](\vVary[i]) = \tyT[i] \}_{i\in I}}{\tyT}\\
\prefTyMem & \coloneqq & \tySilent \ \mid \ \tyLoad{\nVar}\\
    & \mid & \tyStore{\nVar} \ \mid \ \prefTyLock \\
\prefTyLock & \coloneqq & \tyLock{\nLock} \ \ \mid \ \ \tyUnlock{\nLock}\\
    & \mid & \tyRLock{\nLock} \ \mid \ \tyRUnlock{\nLock}
\end{array}\!\!
$
\end{tabular}
\\
\hline
\\[-0.7em]
$
\begin{array}{c}
\tyPar{\tyT}{\tyS} \equiv \tyPar{\tyS}{\tyT}
\qquad
\tyPar{\tyT}{\tyPPar{\tyS}{\tySi}} \equiv \tyPar{\tyPPar{\tyT}{\tyS}}{\tySi}
\qquad 
\tyPar{\tyT}{\zero} \equiv \tyT
\qquad
\tyRes{\nVar}{\tyMem{\nVar}} \equiv \zero
\\[1mm]
\tyRes{\nLock}{\tyLockU{\nLock}} \equiv \zero
\qquad
\tyRes{\nLock}{\tyLockL{\nLock}} \equiv \zero
\qquad
\tyRes{\nLock}{\tyRWLock{\nLock}{i}} \equiv \zero
\qquad
\tyRes{\nLock}{\tyRWLockL{\nLock}{i}} \equiv \zero
\qquad
\tyRes{\nLock}{\tyRWLockW{\nLock}{i}} \equiv \zero
\\[1mm]
\tyRes{\nGeneric}{\tyRes{\nGenerici}{\tyT}} \equiv 
  \tyRes{\nGenerici}{\tyRes{\nGeneric}{\tyT}}
\qquad
\tyPar{\tyT}{\tyRes{\nGeneric}{\tyS}} \equiv 
  \tyRes{\nGeneric}{\tyPPar{\tyT}{\tyS}}
\ \mbox{\scriptsize $(\nGeneric\not\in \fn{\tyT})$}
\end{array}
$
\end{tabular}
}
\end{center}
\vspace{-3mm}
\caption{Syntax of the types.}\label{fig:types-syntax}
\vspace{-5mm}
\end{figure}

The type $\prefTyMem\tyCont\tyT$ denotes a store 
$\tyStore{\nGeneric}$, load $\tyLoad{\nGeneric}$ of shared
variable $\nGeneric$, lock $\tyLock{\nLock}$, unlock $\tyUnlock{\nLock}$,
rlock $\tyRLock{\nLock}$, runlock $\tyRUnlock{\nLock}$ of a (rw)lock $\nLock$, 
followed by the behaviour denoted by type $\tyT$.
It also includes 
an explicit silent action $\tySilent$ followed by the behaviour $\tyProc$.

The type constructs $\tyMem{\nVar}$, $\tyLockU{\nLock}$, $\tyLockL{\nLock}$,
$\tyRWLock{\nLock}{i}$, $\tyRWLockL{\nLock}{i}$ and $\tyRWLockW{\nLock}{i}$
denote the type representations of runtime 
shared variable, unlocked and locked locks,
unlocked (or read-locked), locked and lock-waiting rwlocks,
respectively.
Types for variables and locks include 
shared variable and (rw)lock
creation $\tyNewMem{\nVar}{\tyT}$, $\tyNewLock{\nLock}{\tyT}$ and 
$\tyNewRWLock{\nLock}{\tyT}$ which respectively 
bind $\nVar$ and $\nLock$ in $\tyT$. 
$\fn{\tyT}$ denotes the set of free names of type $\tyT$.

\subsection{Typing System with Shared Variables and Mutexes}

Our typing system is defined in \figurename~\ref{fig:goldy-typing}.
\begin{figure}[ht]
\vspace{-5mm}
\begin{center}
\framebox{
\begin{tabular}{c}
$
\!\!\!\begin{array}{c}
\begin{array}{ll}
\framebox{$\G \vdash \procP \ts \tyT$}
\hfill
\trulename{\ruleTypingZero}
\inferrule
{ }
{\G \vdash \zero \ts \zero}
&
\trulename{\ruleNewMem}
\inferrule
{\G \ctxComp \typing{\nVar}{\typeMem{\sigma}} \vdash \procP \ts \tyT}
{\G \vdash {\pNewMem{\nVar}{\sigma}}\pCont\procP \ts \tyNewMem{\nVar}\tyT}
\\[1.5em]
\trulename{\ruleNewMut}
\inferrule
{\G \ctxComp \typing{\nLock}{\typeLock} \vdash \procP \ts \tyT}
{\G \vdash \pNewLock{\nLock}\pCont\procP \ts \tyNewLock{\nLock}\tyT}
&
\trulename{\ruleNewRMut}
\inferrule
{\G \ctxComp \typing{\nLock}{\typeLock} \vdash \procP \ts \tyT}
{\G \vdash \pNewRWLock{\nLock}\pCont\procP \ts \tyNewRWLock{\nLock}\tyT}
\\[1.5em]
\trulename{\ruleMutLockReq}
\inferrule
{\G \vdash \typing{\nLock}{\typeLock} \quad \G \vdash \procP \ts \tyT}
{\G \vdash \pLock{\nLock}\pCont\procP \ts \tyLock{\nLock}\tyCont\tyT}
&
\trulename{\ruleMemStoreReq}
\inferrule
{\G \vdash \typing{\nVar}{\typeMem{\sigma}} \quad \G \vdash \typing{\nExpr}{\sigma} \quad \G \vdash \procP \ts \tyT} 
{\G \vdash \pStore{\nVar}{\nExpr}\pCont\procP \ts \tyStore{\nVar}\tyCont\tyT}
\\[1.5em]
\trulename{\ruleMutUnlockReq}
\inferrule
{\G \vdash \typing{\nLock}{\typeLock} \quad \G \vdash \procP \ts \tyT}
{\G \vdash \pUnlock{\nLock}\pCont\procP \ts \tyUnlock{\nLock}\tyCont\tyT}
&
\trulename{\ruleMemLoadReq}
\inferrule
{\G \vdash \typing{\nVar}{\typeMem{\sigma}} \quad \G \ctxComp \typing{\nVary}{\typeMem{\sigma}} \vdash \procP \ts \tyT}
{\G \vdash \pLoad{\nVar}{\nVary}\pCont\procP \ts \tyLoad{\nVar}\tyCont\tyT}
\\[1.5em]
\trulename{\ruleRMutRLockReq}
\inferrule
{\G \vdash \typing{\nLock}{\typeLock} \quad \G \vdash \procP \ts \tyT}
{\G \vdash \pRLock{\nLock}\pCont\procP \ts \tyRLock{\nLock}\tyCont\tyT}
&
\trulename{\ruleRMutRUnlockReq}
\inferrule
{\G \vdash \typing{\nLock}{\typeLock} \quad \G \vdash \procP \ts \tyT}
{\G \vdash \pRUnlock{\nLock}\pCont\procP \ts \tyRUnlock{\nLock}\tyCont\tyT}
\\[1.5em]
\trulename{\ruleSilentAction}
\inferrule
{\G \vdash \procP \ts \tyT}
{\G \vdash \pSilent\pCont\procP \ts \tySilent\tyCont\tyT}
&
\trulename{\ruleTypingCall}
\inferrule
{\G \vdash \typing{\vExpr}{\tilde{\sigma}} \quad \G \vdash \typing{\vGeneric}{\vTypeGen}}
{\G, \procX(\tilde{\sigma},\vTypeGen) \vdash \pCall{\procX}{\vExpr}{\vGeneric} \ts \tyCall{\procX}{\vGeneric}}
\\[1.5em]
\trulename{\rulePar}
\inferrule
{\G \vdash \procP \ts \tyT \quad \G \vdash \procQ \ts \tyS}
{\G \vdash \pPar{\procP}{\procQ} \ts \tyPPar{\tyT}{\tyS}}
&
\trulename{\ruleITEConstruct}
\inferrule
{\G \vdash \typing{\nExpr}{\typeBool} \quad \G \vdash \procP \ts \tyT \quad \G \vdash \procQ \ts \tyS}
{\G \vdash \pITE{\nExpr}{\procP}{\procQ} \ts \tyITE{\tyT}{\tyS}}
\\[1.5em]
\end{array}
\\[1.5em]
\hline
\\[-0.3em]
\framebox{$\G \vdash_B \procP \ts \tyT$}
\hspace{1em}
\trulename{\ruleTypingMut}
\inferrule
{\G \vdash \typing{\nLock}{\typeLock}}
{\G \vdash_{\{\nLock\}} \stLock{\nLock} \ts \tyLockU{\nLock}}
\hspace{1em}
\trulename{\ruleTypingMutLocked}
\inferrule
{\G \vdash \typing{\nLock}{\typeLock}}
{\G \vdash_{\{\nLock\}} \stLockL{\nLock} \ts \tyLockL{\nLock}}
\hspace{1em}
\trulename{\ruleTypingRMut}
\inferrule
{\G \vdash \typing{\nLock}{\typeLock}}
{\G \vdash_{\{\nLock\}} \stRWLock{\nLock}{i} \ts \tyRWLock{\nLock}{i}}
\\[1.5em]
\trulename{\ruleTypingRMutLocked}
\inferrule
{\G \vdash \typing{\nLock}{\typeLock}}
{\G \vdash_{\{\nLock\}} \stRWLockL{\nLock}{i} \ts \tyRWLockL{\nLock}{i}}
\hfill
\trulename{\ruleTypingRMutStaged}
\inferrule
{\G \vdash \typing{\nLock}{\typeLock}}
{\G \vdash_{\{\nLock\}} \stRWLockW{\nLock}{i} \ts \tyRWLockW{\nLock}{i}}
\hfill
\trulename{\ruleTypingMem}
\inferrule
{\G \vdash \typing{\nVar}{\typeMem{\sigma}}}
{\G \vdash_{\{\nVar\}}\stMem{\nVar}{\sigma}{\nVal} \ts \tyMem{\nVar}}
\\[1.5em]
\trulename{\ruleRestrict}
\inferrule
{\G \ctxComp \typing{\nGeneric}{\nTypeGen} \vdash_{B} \procP \ts \tyT}
{\G \vdash_{B\backslash\nGeneric} \pRes{\nGeneric}{\procP} \ts \tyRes{\nGeneric}{\tyT}}
\hfill
\trulename{\ruleTypingPar}
\inferrule
{\G \vdash_{B_1} \procP \ts \tyT \quad \G \vdash_{B_2} \procQ \ts \tyS \quad \mbox{\tiny $B_1 \cap B_2 = \emptyset$}}
{\G \vdash_{B_1 \cup B_2} \tyPar{\procP}{\procQ} \ts \tyPPar{\tyT}{\tyS}}
\\[1.5em]
\hline
\\[-0.3em]
\framebox{$\G \vdash \procProg \ts \tyEqn$}\hfill~\\[0.7em]
\hfill
\trulename{\ruleDefinition}
\inferrule
{\forall i\in I : \G \ctxComp {\procX[i](\tilde{\sigma}_i,\vTypeGen[i])}
    \ctxComp \typing{\vVar[i]}{\tilde{\sigma}_i}
    \ctxComp \typing{\vVary[i]}{\vTypeGen[i]}
 \vdash \procP[i] \ts \tyT[i]
 \\
 \G \ctxComp \procX[1](\tilde{\sigma}_1,\vTypeGen[1])
    \ctxComp \ldots
    \ctxComp \procX[n](\tilde{\sigma}_n,\vTypeGen[n])
 \vdash \procQ \ts \tyS
}
{\G \vdash \pIn{\{\pDef{\procX[i]}{\vVar[i],\vVary[i]}{\procP[i]} \}_{i\in I}}{\procQ}
    \ts \tyIn{\{\tyDef{\tyCallee[{\procX[i]}]}{\vVary[i]}{\tyT[i]}\}_{i \in I}}{\tyS}
}
\end{array}\!\!\!
$
\end{tabular}
}
\end{center}
\vspace{-3mm}
\caption{Typing Rules for Shared Variables and Mutexes.\label{fig:goldy-typing}}
\vspace{-5mm}
\end{figure}

The judgement ($\G \vdash \procP \ts \tyT$),
where $\G$ is a typing environment that maintains information about locks and 
shared variables, and types the part of a term explicitly written by the developer. 
We write $\G \vdash \mathcal{J}$ for $\mathcal{J} \in \G$ and
$\G \vdash \typing{\nExpr}{\sigma}$ to state that the expression $\nExpr$ is
well-typed according to the types of variables in $\G$. We write 
$\typing{\nGeneric}{\nTypeGen}$ for the typing of a name in generality, 
which can be (1)
$\typing{\nVar}{\typeMem{\sigma}}$ to denote a shared variable $\nVar$ 
with stored value type $\sigma$ and (2) 
$\typing{\nLock}{\typeLock}$ to state that $\nLock$ is a (rw)lock.
We omit the rules of expressions $\nExpr$. 
We write $\domain{\G}$ to denote 
the set of locks and shared variable bindings in $\G$.

The rules are as follows. 
Rules \trulename{\ruleMemLoadReq} and \trulename{\ruleMemStoreReq} type 
load and store types for shared variable $\nVar$ where 
the type of the stored value 
matches the payload type $\sigma$ of value $\nVar$, and the
continuation $\procP$ has type $\tyT$.
Rules \trulename{\ruleMutLockReq} and \trulename{\ruleMutUnlockReq} (and \trulename{\ruleRMutRLockReq} and
\trulename{\ruleRMutRUnlockReq}) type the lock actions in processes 
by corresponding types. There
is no payload type to check, only that the lock name is associated to
a lock or read-write lock.
Rules {\trulename{\ruleNewMem}} and {\trulename{\ruleNewMut} (resp. \trulename{\ruleNewRMut}}) allocate a
fresh shared variable name with payload type
$\sigma$ or a lock (resp. rwlock). Other context rules 
are standard. 

The judgement ($\G \vdash_B P \ts \tyT$) 
types process created during execution of
a program and provides the invariants to prove the type safety. 
$B$ is a set of shared variables and locks with associated
runtime buffers to ensure their uniqueness. 
A shared variable heap is typed with rule {\trulename{\ruleTypingMem}},
and all five states of locks are typed by corresponding lock types. 
Restriction is typed here, as it takes the relevant type out of the typing context 
and removes the corresponding name from $B$.

The judgement ($\G \vdash \procProg \ts \tyEqn$) types a program, 
that consists of a process and a set of runtime stores, 
accordingly to their respective types. 

We use the structural congruence on types 
to define {\em normal forms} of types in the same way as done for 
\gol terms in Definition~\ref{def:migo-normal-form}, and study further 
properties on types up to normal form.
Examples of typing of processes can be found in Example~\ref{ex:type}.

\begin{example}\label{ex:type}
The unsafe program of \figurename~\ref{fig:running-ex-rwmut-race}, modelled by process $\procProg[\m{race}]$ in
Example~\ref{ex:syntax}, has the following type:
\[
\tyEqn[\m{race}] \coloneqq \tyIn
  {
    \left\{\begin{array}{@{}l@{\ }l@{\ }l}
      \tyEntryProc & = & \tyNewMem{\nVar}
        \tyNewRWLock{\nLock}
        \tyPPar{\tyCall{\procP}{\nVar,\nLock}}
               {\tyCall{\procQ}{\nVar,\nLock}}
      \\
      \tyCallee[\procP](\nVary,\nVarz) & = & \tyRLock{\nVarz}\tyCont
        \tyLoad{\nVary}\tyCont
        \tyStore{\nVary}\tyCont
        \tyRUnlock{\nVarz}\tyCont
        \tyRLock{\nVarz}\tyCont
        \tyLoad{\nVary}\tyCont
        \tySilent\tyCont
        \tyRUnlock{\nVarz}\tyCont
        \zero
      \\
      \tyCallee[\procQ](\nVary,\nVarz) & = & \tyRLock{\nVarz}\tyCont
        \tyLoad{\nVary}\tyCont
        \tyStore{\nVary}\tyCont
        \tyRUnlock{\nVarz}\tyCont
        \zero
    \end{array}\right\}
  }
  {\tyEntryPt}
\]

The safe version in \figurename~\ref{fig:running-ex-rwmut-safe}, modelled
by 
process $\procProg[\m{safe}]$ 
in Example~\ref{ex:syntax}, has type:
\[
\tyEqn[\m{safe}] \coloneqq \tyIn
  {
    \left\{\begin{array}{@{}l@{\ }l@{\ }l}
      \tyEntryProc & = & \tyNewMem{\nVar}
        \tyNewRWLock{\nLock}
        \tyPPar{\tyCall{\procP}{\nVar,\nLock}}
               {\tyCall{\procQ}{\nVar,\nLock}}
      \\
      \tyCallee[\procP](\nVary,\nVarz) & = & \tyLock{\nVarz}\tyCont
        \tyLoad{\nVary}\tyCont
        \tyStore{\nVary}\tyCont
        \tyUnlock{\nVarz}\tyCont
        \tyRLock{\nVarz}\tyCont
        \tyLoad{\nVary}\tyCont
        \tySilent\tyCont
        \tyRUnlock{\nVarz}\tyCont
        \zero
      \\
      \tyCallee[\procQ](\nVary,\nVarz) & = & \tyLock{\nVarz}\tyCont
        \tyLoad{\nVary}\tyCont
        \tyStore{\nVary}\tyCont
        \tyUnlock{\nVarz}\tyCont
        \zero
    \end{array}\right\}
  }
  {\tyEntryPt}
\]
\end{example}

\subsection{Operational Semantics of the Behavioural Types}
\label{sub:sync-ty-sem}

This section defines the semantics of our types.
The labels, ranged over by $\actname, \actnamei$, have the form:
\begin{center}\framebox{$
\begin{array}{@{\,}r@{\,}l@{\,}}
\actname \coloneqq & 
\labtyload{\nVar} \ | \ \labtystore{\nVar}{\occurgen} \ | \ 
\labtylock{\nLock} \ | \ \labtyunlock{\nLock} \ | \
\labtyrlock{\nLock} \ | \ \labtyrunlock{\nLock} 
\ | \ \labheap{\nVar}
\ | \ \mutbarb{\nLock} \ | \ \lckmutbarb{\nLock}
\
| \ 
\rmutbarb{\nLock}{i}\ | \ \waitrmutbarb{\nLock}{i}\ | \ \waitmutbarb{\nLock} \ | \ 
\acttau \ | \ \labsync{\nGeneric} 
\end{array}
$}\end{center}
The labels denote the actions introduced in
this paper: 
load and store actions, lock, unlock, rlock and runlock actions, 
shared heap manipulation, and the five kinds of 
(rw)lock state transitions. The end of the line is for 
silent transition and
synchronisation over a name. 

The semantics of our types is given by the
labelled transition system (LTS) (modulo $\alpha$-conversion), 
extending that of CCS, which is shown 
in \figurename~\ref{fig:types-sem}. 
\begin{figure}[h]
\vspace{-3mm}
\begin{center}
\framebox{
\begin{tabular}{l}
\begin{tabular}{l|l}
$
\!\!\!\!\!\!\!\!\!\!\!
\begin{array}{l}
  \framebox{\text{Lock and Memory actions}}
  \\
  \begin{array}{ll}
   \ltsrulename{\ruleMutLockReq}&
    {\tyLock{\nLock}\tyCont\tyT}\tra{\labtylock{\nLock}}{\tyT}
  \\
   \ltsrulename{\ruleMutUnlockReq}&
    {\tyUnlock{\nLock}\tyCont\tyT}\tra{\labtyunlock{\nLock}}{\tyT}
  \\
   \ltsrulename{\ruleRMutRLockReq}&
    {\tyRLock{\nLock}\tyCont\tyT}\tra{\labtyrlock{\nLock}}{\tyT}
  \\
   \ltsrulename{\ruleRMutRUnlockReq}&
    {\tyRUnlock{\nLock}\tyCont\tyT}\tra{\labtyrunlock{\nLock}}{\tyT}
  \\
   \ltsrulename{\ruleMemLoadReq}&
    {\tyLoad{\nVar}\tyCont\tyT}\tra{\labtyload{\nVar}}{\tyT}
  \\
   \ltsrulename{\ruleMemStoreReq}&
    {\tyStore{\nVar}\tyCont\tyT}\tra{\labtystore{\nVar}{\occurnone}}{\tyT}
  \\
  \hline
   \ltsrulename{\ruleMutLockAck}& 
    {\tyLockU{\nLock}}\tra{\mutbarb{\nLock}}
    {\tyLockL{\nLock}}
  \\
   \ltsrulename{\ruleMutUnlockAck}& 
    {\tyLockL{\nLock}}\tra{\lckmutbarb{\nLock}}
    {\tyLockU{\nLock}}
  \\
    \ltsrulename{\ruleRMutLockAck}&
    {\tyRWLockW{\nLock}{0}}\tra{\mutbarb{\nLock}}
    {\tyRWLockL{\nLock}{0}}
  \\
   \ltsrulename{\ruleRMutUnlockAck}&
    {\tyRWLockL{\nLock}{0}}\tra{\lckmutbarb{\nLock}}
    {\tyRWLock{\nLock}{0}}
  \\
   \ltsrulename{\ruleRMutRLockAck}&
    {\tyRWLock{\nLock}{i}}\tra{\rmutbarb{\nLock}{i}}
    {\tyRWLock{\nLock}{i+1}}
  \\
   \ltsrulename{\ruleRMutRUnlockAck}&
    {\tyRWLock{\nLock}{i+1}\tra{\waitrmutbarb{\nLock}{i}}\tyRWLock{\nLock}{i}}
  \\
   \ltsrulename{\ruleRMutLockStageAck}&
    {\tyRWLock{\nLock}{i}}\tra{\waitmutbarb{\nLock}}
    {\tyRWLockW{\nLock}{i}}
  \\
   \ltsrulename{\ruleRMutStagedRUnlockAck}&
    {\tyRWLockW{\nLock}{i+1}\tra{\waitrmutbarb{\nLock}{i}}\tyRWLockW{\nLock}{i}}
  \\[0.2em]
   \ltsrulename{\ruleMemActionAck}& 
    {\tyMem{\nVar}}\tra{\labheap{\nVar}}{\tyMem{\nVar}}
  \\[.2em]
  \end{array}
  \end{array}
$
&\
$
\begin{array}{l}
  \framebox{\text{Synchronisation rules}}
  \\[1em]
  \begin{array}{c}
  \ltsrulename{\ruleMemAction}\
    \inferrule
    {\tyT \semty{\actname} \tyTi
      \quad
      \tyS \semty{\labtyheap{\nVar}} \tySi
      \quad
      \mbox{\scriptsize $\actname = {\labtystore{\nVar}{\occurgen}}, {\labtyload{\nVar}}$}
    }
    {\tyPar{\tyT}{\tyS}\tra{\labsync{\nVar}}
     \tyPar{\tyTi}{\tySi}}
  \\[1em]
  \ltsrulename{\ruleMutLock}\
    \inferrule
    {\tyT\tra{\lckbarb{\nLock}}\tyTi \quad \tyS\tra{\mutbarb{\nLock}}\tySi}
    {\tyPar{\tyT}{\tyS} \tra{\labsync{\nLock}} \tyPar{\tyTi}{\tySi}}
  \\[1em]
  \ltsrulename{\ruleMutUnlock}\
    \inferrule
    {\tyT\tra{\ulckbarb{\nLock}}\tyTi \quad \tyS\tra{\lckmutbarb{\nLock}}\tySi}
    {\tyPar{\tyT}{\tyS}\tra{\labsync{\nLock}} \tyPar{\tyTi}{\tySi}}
  \\[1em]
    \ltsrulename{\ruleRMutRLock}
    \inferrule
    {\tyT\tra{\rlckbarb{\nLock}}\tyTi \quad \tyS\tra{\rmutbarb{\nLock}{i}}\tySi}
    {\tyPar{\tyT}{\tyS} \tra{\labsync{\nLock}} \tyPar{\tyTi}{\tySi}}
  \\[1em]
    \ltsrulename{\ruleRMutRUnlock}
    \inferrule
    {\tyT\tra{\rulckbarb{\nLock}}\tyTi \quad \tyS\tra{\waitrmutbarb{\nLock}{i}}\tySi}
    {\tyPar{\tyT}{\tyS} \tra{\labsync{\nLock}} \tyPar{\tyTi}{\tySi}}
  \\[1em]
    \ltsrulename{\ruleRMutLockStage}
    \inferrule
    {\tyT\tra{\lckbarb{\nLock}}\tyTi \quad \tyS\tra{\waitmutbarb{\nLock}}\tySi}
    {\tyPar{\tyT}{\tyS} \tra{\labsync{\nLock}} \tyPar{\tyT}{\tySi}}
  \\[1.2em]
  \ltsrulename{\ruleSilentAction}\
   \tySilent\tyCont\tyT\tra{\acttau}\tyT
  \end{array}\hspace*{-5mm}
  \end{array}\hspace*{-5mm}
$
\end{tabular}\hspace*{-5mm}
 \\
 \hline
 \\[-0.9em]
 $\!\!\!\!\!
 \begin{array}{l}
 \!\!\!\begin{array}{lr@{\ }l}
 \framebox{Runtime structures creation}
 &
 \ltsrulename{\ruleNewMem}
 &
    {\tyNewMem{\nVar}{\tyT}\tra{\acttau}
      \tyRes{\nVar}{\tyPPar{\tyT}{\tyMem{\nVar}}}}
  \\[.4em]
 \ltsrulename{\ruleNewMut}\ 
    {\tyNewLock{\nLock}{\tyT}\tra{\acttau}
     \tyRes{\nLock}{\tyPPar{\tyT}{\tyLockU{\nLock}}}}
  \hspace{8em}~
 &
 \ltsrulename{\ruleNewRMut}
 &
    {\tyNewRWLock{\nLock}{\tyT}\tra{\acttau}
     \tyRes{\nLock}{\tyPPar{\tyT}{\tyRWLock{\nLock}{0}}}}
 \\[.4em]
 \end{array}
 \\[1em]
 \hline
 \\[-0.7em]
 \framebox{Context rules}
 \hfill
 \ltsrulename{\ruleAlphaEquiv}
  \inferrule
   {\tyT \equiv \tyTi \quad \tyTi\tra{\actname}\tyTii}
   {\tyT \tra{\actname} \tyTii}
 \hfill
 \ltsrulename{\ruleITEConstruct}
  \inferrule
   {j \in I}
   {\tyChoice{\tyT[i]}{i\in I}\tra{\acttau}\tyT[j]}
 \\[1em]
 \ltsrulename{\ruleRestrictFree}
  \inferrule
   {\tyT \tra{\actname} \tyTi \quad \nGeneric\notin \fn{\actname}}
   {\tyRes{\nGeneric}{\tyT}\tra{\actname}
    \tyRes{\nGeneric}{\tyTi}}
 \hfill
 \ltsrulename{\ruleRestrictBind}
  \inferrule
   {\tyT \tra{\labsync{\nGeneric}} \tyTi}
   {\tyRes{\nGeneric}{\tyT}\tra{\acttau}
    \tyRes{\nGeneric}{\tyTi}}
 \hfill
 \ltsrulename{\ruleParLeft}
  \inferrule
   {\tyT \tra{\actname} \tyTi}
   {\tyPar{\tyT}{\tyS} 
     \tra{\actleft{\actname}} \tyPar{\tyTi}{\tyS}}
 \\[1em]
 \ltsrulename{\ruleParRight}
  \inferrule
   {\tyS \tra{\actname} \tySi}
   {\tyPar{\tyT}{\tyS} 
     \tra{\actright{\actname}} \tyPar{\tyT}{\tySi}}
 \hfill
 \ltsrulename{\ruleDefinition}
  \inferrule
   {\tyPar{\tyT\subs{\vGeneric}{\vVar}}{\tyS} \tra{\actname}
    \tyTi \quad \tyCallee[\procX](\vVar) = \tyT}
   {\tyPar{\tyCall{\procX}{\vGeneric}}{\tyS}
    \tra{\actname} \tyTi}
 \end{array}\hspace*{-2.5mm}
$
\end{tabular}
}
\end{center}
\vspace{-3mm}
\caption{LTS Reduction Semantics for the Types.}\label{fig:types-sem}
\vspace{-3mm}
\end{figure}

\\
Rules \ltsrulename{\ruleMemStoreReq} and \ltsrulename{\ruleMemLoadReq} 
allow a type to
emit a store and load action on a shared variable $\nVar$.
Rule \ltsrulename{\ruleMutLockReq} (resp.\ \ltsrulename{\ruleMutUnlockReq}) emits a lock (resp.\ unlock) 
action on a shared lock $\nLock$.
Rules \ltsrulename{\ruleNewMem} and \ltsrulename{\ruleNewMut} 
(resp. \ltsrulename{\ruleNewRMut}) create a 
a new shared heap $\nVar$ or unlocked lock (resp. rwlock) store $\nLock$.
Rule \ltsrulename{\ruleMemActionAck} models the ability of a shared heap to be read or
updated at any time, and rule \ltsrulename{\ruleMemAction} allows a load or
store action to synchronise with its associated heap.

Rule \ltsrulename{\ruleMutLockAck} 
makes a lock to be closed, and rule \ltsrulename{\ruleMutUnlockAck} 
unlocks a claimed lock. Rules \ltsrulename{\ruleMutLock} 
and \ltsrulename{\ruleMutUnlock}
make the corresponding actions to synchronise with their associated lock store.
Equivalent rules for rwlocks act the same as in the processes. Pay attention 
to the same quirk as in processes: \ltsrulename{\ruleRMutLockStage} does not 
consume the lock action in $\tyT$, as this rules serves to forbid further 
read-lock calls from being executed if a lock call is staged.

Rule  \ltsrulename{\ruleITEConstruct} represents 
the internal choice behaviour of the conditional processes. 

In \figurename~\ref{fig:types-sem}, we omit the symmetric rules for
parallel composed processes (such as \ltsrulename{\ruleMemAction}). 
We write $\tra{}$ for $\equiv\tra{\acttau}\equiv \cup 
\equiv\tra{\labsync{\nGeneric}}\equiv$ and 
$\tyT \tra{}^\ast \tra{\actname}$ if there exist $\tyTi$ and $\tyTii$
such that $\tyT\tra{}^\ast \tyTi\tra{\actname}\tyTii$.

\begin{example}\label{ex:typered}
The unsafe version of \figurename~\ref{fig:running-ex-rwmut-race}, modelled by process $\procProg[\m{race}]$ in
Example~\ref{ex:syntax} and typed by $\tyEqn[\m{race}]$ in Example~\ref{ex:type}, 
has the following possible reduction (following the same reduction order as Example~\ref{ex:reduction}):
\[
\begin{array}{r@{\ }c@{\ }r@{}l}
\tyEqn[\m{race}] & = & \tyNewMem{\nVar}
      \tyNewRWLock{\nLock}
      &\left(\begin{array}{@{}l@{\ }l}
        \tyPar{
	  & \tyRLock{\nLock}\tyCont
	   \tyLoad{\nVar}\tyCont
	   \tyStore{\nVar}\tyCont
	   \tyRUnlock{\nLock}\tyCont
	   \tyRLock{\nLock}\tyCont
	   \tyLoad{\nVar}\tyCont
	   \tySilent\tyCont
	   \tyRUnlock{\nLock}\tyCont
	   \zero\\}{
          & \tyRLock{\nLock}\tyCont
	   \tyLoad{\nVar}\tyCont
	   \tyStore{\nVar}\tyCont
	   \tyRUnlock{\nLock}\tyCont
	\zero}
	\end{array}
      \right)\\[1mm] 
  & \ \tra{}^{2} & \tyRes{\nVar\nLock}{
      &\left(\begin{array}{@{}l@{\ }l}
        \tyPar{\tyPar{\tyPar{
	  & \tyRLock{\nLock}\tyCont
	   \tyLoad{\nVar}\tyCont
	   \tyStore{\nVar}\tyCont
	   \tyRUnlock{\nLock}\tyCont
	  \tyRLock{\nLock}\tyCont
	   \tyLoad{\nVar}\tyCont
	   \tySilent\tyCont
	   \tyRUnlock{\nLock}\tyCont
	   \zero\\}{
          & \tyRLock{\nLock}\tyCont
	   \tyLoad{\nVar}\tyCont
	   \tyStore{\nVar}\tyCont
	   \tyRUnlock{\nLock}\tyCont
	\zero}}
	  {\tyMem{\nVar}}}
	 {\tyRWLock{\nLock}{0}}
	\end{array}
      \right)}
\\[1mm] & \tra{}^{6} & \tyRes{\nVar\nLock}{
      &\left(\begin{array}{@{}l@{\ }l}
        \tyPar{\tyPar{\tyPar{
	   \tyStore{\nVar}\tyCont
	   \tyRUnlock{\nLock}\tyCont
	  \tyRLock{\nLock}\tyCont
	   \tyLoad{\nVar}\tyCont
	   \tySilent\tyCont
	   \tyRUnlock{\nLock}\tyCont
	   \zero\\}{
	   \tyStore{\nVar}\tyCont
	   \tyRUnlock{\nLock}\tyCont
	\zero}}
	  {\tyMem{\nVar}}}
	 {\tyRWLock{\nLock}{2}}
	\end{array}
      \right)} = \tyRes{\nVar\nLock}{\tyTi}
\end{array}
\]
We note that $\tyTi$ is a type of $\procPi$ which has a data race in 
Example~\ref{runnig-ex-red}.
\end{example}

\section{Properties of \gol Processes and Types}\label{sec:types}
This section proves 
two main results, the subject reduction and progress properties 
with respect to behavioural types. 
Our goal is to 
classify subsets of \gol programs for which liveness, data race freedom 
and safety coincide with liveness, data race freedom and safety of 
their types.%
\full{}{Detailled proofs for this section are available in~\cite{JY20:full}.}

\subsection{Type soundness of \gol processes}\label{sect:gen-types}
A basic property for types is to be preserved under structural congruence
and to be able to reduce the same as the process. 

\begin{proposition}[Subject Congruence]\label{prop:cong}\rm
If $\G \vdash_B \procP \ts \tyT$ and
$\procP \equiv \procPi$, then $\exists \tyTi \equiv \tyT$ such that 
$\G \vdash_B \procPi \ts \tyTi$.
\end{proposition}
\full{See Appendix~\ref{app:proofs} for the proof.}{}%
The following 
type soundness theorem shows that behaviours of processes can be 
simulated by behaviours of types. 
\begin{theorem}[Subject Reduction]\label{the:red}\rm
If $\G \vdash_{B} \procP \ts \tyT$ and
$\procP \tra{} \procPi$, then $\exists \tyTi$ such that 
$\G \vdash_{B} \procPi \ts \tyTi$ and
$\tyT \tra{} \tyTi$.
\end{theorem}
\full{See Appendix~\ref{app:proofs} for the proof.}{}%
The following progress theorem says that 
the action availability on types infers that on processes. 

We first need to define 
\emph{barbs} to represent capabilities of a type at a given time in reduction, 
akin to how process barbs are defined in Definition~\ref{def:barb}. 

\begin{definition}[Type Barbs]\label{def:barbs:types}\rm
The barbs on types are defined as follows:
\begin{description}
\item[Prefix Actions:] 
  $
\begin{array}{lll}
  \tyStore{\nVar}\barb{\stbarb{\nVar}{\occurnone}};
&
\tyLoad{\nVar}\barb{\ldbarb{\nVar}{\occurnone}};
&
\tyLock{\nLock}\barb{\lckbarb{\nLock}};
\\[1mm]
\tyUnlock{\nLock}\barb{\ulckbarb{\nLock}};
&
\tyRLock{\nLock}\barb{\rlckbarb{\nLock}};
&
\tyRUnlock{\nLock}\barb{\rulckbarb{\nLock}}
\end{array}
$

 \item[Types:]
 if $\tyT\tra{\actname}\tyTi$ where $\actname$ is a communication action over
  a shared variable or
  $\labsync{\nGeneric}$ or a lock action, then $\tyT\barb{\actname}$.
\end{description}
\end{definition}

\begin{theorem}[Progress]\label{the:prog}\rm
Suppose $\G \vdash \procP \ts \tyT$. Then 
if $\tyT\tra{\actname}\tyT[0]$ for $\actname\in \{ \labsync{\nGeneric}, \acttau\}$ 
for some heap or lock $\nGeneric$, then there exists $\procPi,\tyTi$ such that 
$\procP\tra{}\procPi$, 
$\tyT\tra{\actname}\tyTi$, 
and $\G \vdash \procPi \ts \tyTi$. 
\end{theorem}
To prove this theorem, we use a lemma which shows a correspondence
of barbs between processes and types (defined similarly with barbs of 
processes, cf Definition~\ref{def:barbs:types}).%
\full{The proof can be found in Appendix \ref{app:proofs}.}{}%
Note that in Theorem~\ref{the:prog}, $\tyTi$ and $\tyT[0]$ might be different. 
This is because a selection type (i.e. the internal choice) 
can reduce non-deterministically but the corresponding conditional process
usually is deterministic.

\subsection{Safety and Liveness for Types}\label{sect:types-prop}
In this subsection, we define 
safety and liveness for types, which correspond to  
Definitions \ref{def:safe}, \ref{def:live} and \ref{def:hb-drace}, respectively.

\begin{definition}[Safety]~\label{def:type-safe}\rm
Type $\tyEqn$ 
is {\em safe} if for all $\tyT$ such that 
$\tyEqn\tra{}^*\tyRes{\vGeneric}{\tyT}$, 
(a) if $\tyT\barb{\ulckbarb{\nLock}}$
then $\tyT\barb{\lckmutbarb{\nLock}}$; and 
(b) if $\tyT\barb{\rulckbarb{\nLock}}$
then $\tyT\barb{\waitrmutbarb{\nLock}{i}}$.
\end{definition}

\begin{definition}[Liveness]~\label{def:type-live}\rm
Type $\tyEqn$ is \emph{live} if for all $\tyT$ such that
$\tyEqn \tra{}^* \tyRes{\vGeneric}{\tyT}$,
if ${\tyT}\barb{\lckbarb{\nLock}}$ or 
${\tyT}\barb{\rlckbarb{\nLock}}$ 
then 
$\tyT\wbarb{\labsync{\nLock}}$.
\end{definition}

\begin{definition}[Data Race]~\label{def:type-drace}\rm
$\tyEqn$ has a data race if and only if 
there exists $\tyT$ such that 
$\tyEqn\tra{}^{\ast}\tyRes{\vGeneric}{\tyT}$ with 
$\tyT\barb{\actname[1]}$, $\tyT\barb{\actname[2]}$,  
$\actname[1] = \stbarb{\nVar}{\occurgen}$, 
$\actname[2] \in \{\stbarb{\nVar}{\occurgeni},\ldbarb{\nVar}{\occurgeni}\}$ and 
$\occurgen\neq \occurgeni$.
\end{definition}

We say that $\tyEqn$ is data race free if it has no data race.

\subsection{Liveness and Safety for Typed \gol}\label{sect:safe-live}

In this section, we state several 
propositions and theorems adapted from \cite{LNTY2017} 
to our new process and types primitives and their LTSs. Our goal is to 
classify subsets of \gol programs for which liveness, data race freedom 
and safety coincide with liveness, data race freedom and safety of 
their types.

First, we prove that safety and data race freedom (which is a form of safety) 
have no restriction, and that 
proving that a type is safe always entails the associated program is safe.

\begin{theorem}[Process Safety and Data Race Freedom]\label{the:proc-safe-drf}
Suppose $\G \vdash \procProg \ts \tyEqn$ and $\tyEqn$ is safe (resp. data 
race free). 
Then $\procProg$ is safe (resp. data race free).
\end{theorem}

We then prove that liveness of types is equivalent to liveness of 
programs for a subset of the \gol programs, in three steps: 
(1) programs that always have a terminating path, 
(2) finite branching programs, and (3) programs that simulate 
non-deterministic branching in infinitely recurring conditionals.

We first study the case of programs that always have a path to termination:

\begin{definition}[May Converging Program]\label{def:may-conv}\rm
Let $\G \vdash \procProg \ts \tyEqn$. We write $\procProg\in\convm$ 
if for all $\procProg\tra{}^{*}\procPi$, $\procPi\tra{}^{*}\zero$.
\end{definition}

\label{ex:may-conv}
An example of May Converging program is the following program, where 
process $\procP$ loops and alternates $\nVar$ to values 1 and 0 until 
the $end$ flag is set, and $\procQ$ loops reading $\nVar$ until it reads 
a value 0, in which case it sets the $end$ flag and returns:
\[
\procProg[\m{mc}] \coloneqq \pIn
  {
    \left\{\begin{array}{@{}l@{\ }r@{\ }l}
      \pEntryProc & = & \pNewMem{\nVar}{\m{int}}\pCont
        \pNewMem{end}{\m{bool}}\pCont
        \pNewRWLock{\nLock}\pCont\\
        & & \pPPar{\pCall{\procP}{\nVar,end,\nLock}{}}
              {\pCall{\procQ}{\nVar,end,\nLock}{}}
      \\
      \procP(\nVar,end,\nLock) & = & \pLock{\nLock}\pCont
        \pLoad{\nVar}{\nVary}\pCont
        \pStore{\nVar}{1-\nVary}\pCont
        \pLoad{end}{\nVarz}\pCont\\
        & & \pUnlock{\nLock}\pCont
        \pITE{\nVarz}{\zero\ }{\pCall{\procP}{\nVar,end,\nLock}{}}
      \\
      \procQ(\nVar,end,\nLock) & = & \pLock{\nLock}\pCont
        \pLoad{\nVar}{\nVary}\pCont
        \pUnlock{\nLock}\pCont\\
        & & \pITE{\nVary = 0}{\pLock{\nLock}\pCont
                           \pStore{end}{\true}\pCont
                           \pUnlock{\nLock}\pCont
                           \zero\\ & &}
                         {\pCall{\procQ}{\nVar,end,\nLock}{}}
    \end{array}\right\}
  }
  {\pEntryPt}
\]

The next proposition states that on these programs, proving liveness of 
their types is enough to ensure liveness of the associated program.

\begin{proposition}\label{prop:may-conv-live}
Assume $\G \vdash \procProg \ts \tyEqn$ and $\tyEqn$ is live. 
{\rm (1)} Suppose there exists $\procPi$ such that $\procProg \tra{}^{*} \procPi
\not \tra{}$. Then $\procPi\equiv\zero$; and {\rm (2)} If $\procProg\in\convm$, 
then $\procProg$ is live.
\end{proposition}

We now need to define a subset of May Converging programs, that is the set 
of always terminating programs. This is needed because our implementation, 
that we describe in \S~\ref{sec:implementation}, 
only allows to check and ensure liveness for 
terminating programs, ie. the result of our tool for liveness is assured to 
coincide with actual program liveness only on terminating programs. 

Note that the tool is able to model check non-terminating programs (under the 
assumption they don't spawn an unbounded amount of new threads), but may 
in rare instances lead to a false positive, due to the approximations the 
model checker has to make in this case.

\begin{definition}[Terminating Program]\label{def:term-prog}\rm
We write $\procProg\in\term$ if there exists some non-negative number $n$ 
such that, for all $\procP$ such that $\procProg\tra{}^{n}\procP$, 
$\procP\equiv\zero$.
\end{definition}

The following proposition states that this subset of programs is included in 
the set of May Converging programs. We note that this inclusion is strict: 
a program that may loop forever on a select construct, with a timeout branch 
that terminates the program, is May Converging but not terminating in the 
sense of the above definition, as we may always find a reduction path that 
continues longer than any finite bound.

\begin{proposition}\label{prop:term-is-may-conv}
$\procProg\in\term$ implies $\procProg\in\convm$.
\end{proposition}

\begin{proof}
By definition of the May Converging set of programs, all programs that 
always converge are May Converging.
\end{proof}

\begin{example}\label{ex:terminating}
Note that the running examples we defined in \figurename~\ref{fig:running-ex-rwmut-race} 
and~\ref{fig:running-ex-rwmut-safe} are both terminating, and so are their modelling 
processes given in Example~\ref{ex:syntax}.
\end{example}

The next set of programs we highlight is finite branching programs. 
We first define a series of items, including deterministic 
marking of conditionals and the set of infinitely branching programs, in order 
to grab everything not infinitely branching ({\em ie.} outside of the 
defined set).

\textbf{Marked Programs.}\label{def:marking}
\ Given a program $\procProg$ we define its {\em marking}, written $\m{mark}(\procProg)$, 
as the program obtained by deterministically labelling every occurrence of a conditional 
of the form $\pITE{\nExpr}{\procP}{\procQ}$ in $\procProg$, as 
$\pMPITE{\nExpr}{\procP}{\procQ}{n}$, such that $n$ is distinct natural number 
for all conditionals in $\procProg$.

\textbf{Marked Reduction Semantics.}\label{def:marked-red}
We modify the marked reduction semantics, written 
$\procP\tra{l}\procPi$, stating that 
program $\procP$ reduces to $\procPi$ in a single step, performing action $l$. The 
grammar of action labels is defined as: 
$l\coloneqq \actgen \mid \iflab{n}{\llab} \mid \iflab{n}{\rlab}$ 
where $\actgen$ denotes a non-conditional action, taking into account all 
existing actions and all rules expect \rulename{\ruleITETrue} and \rulename{\ruleITEFalse}, 
$\iflab{n}{\llab}$ denotes a conditional branch 
marked with the natural number $n$ in which the $\m{then}$ branch is chosen, and 
$\iflab{n}{\rlab}$ denotes a conditional branch in which the $\m{else}$ 
branch is chosen. Because of the changes in notations, conditional branches 
are not considered a standard reduction step in $\tra{}$ any more.
The marked reduction semantics replace rules \rulename{\ruleITETrue} and \rulename{\ruleITEFalse}.

\textbf{Trace.}\label{def:trace}
We define an execution trace of a program $\procP$ as 
the potentially infinite sequence of action labels $\vec{l}$ such that 
$\procP\tra{l_1}\procP[1]\tra{l_2}\ldots$, with $\vec{l}=\{l_1,l_2\ldots\}$. We write 
$\traceset_{\procP}$ for the set of all possible traces of a process $\procP$.

\textbf{Reduction Contexts }\label{def:red-context}
are given by:
$\tycontxt_{r} \ \coloneqq \ [] \, \mid \, (\procP \parr \tycontxt_{r})
\, \mid \, (\tycontxt_{r} \parr \procP) \, \mid \, \pRes{\nGeneric}{\tycontxt_{r}}$.

\textbf{Infinite Conditional.}\label{def:inf-cond}
We say that $\procProg$ has infinite conditionals, written as $\procProg\in\ic$, iff 
$\m{mark}(\procProg)\tra{}^{*}\tycontxt_{r} 
[\pMPITE{\nExpr}{\procP}{\procQ}{n}]=\procR$, for some $n$, and $\procR$ 
has an infinite trace where $\iflab{n}{\llab}$ or $\iflab{n}{\rlab}$ appears 
infinitely often. 
We say that such an $n$ is an {\em infinite conditional mark} and write 
$\infcond(\procProg)$ for the set of all such marks.

We state in the next proposition that finite branching programs can be 
ensured live by checking for liveness of their types.

\begin{proposition}[Liveness for Finite Branching]\label{prop:finite-branch-live}
Suppose $\G \vdash \procProg \ts \tyEqn$ and $\tyEqn$ is live and 
$\procProg \not\in \ic$. Then $\procProg$ is live.
\end{proposition}

An example of finite branching program is the Dining Philosophers problem:
\[
\procProg[\m{dinephil}] \coloneqq \pIn
  {
    \left\{\begin{array}{@{}l@{\ }r@{\ }l@{}}
      \pEntryProc & = & \pNewMem{f_1}{\m{int}}\pCont
        \pNewMem{f_2}{\m{int}}\pCont
        \pNewMem{f_3}{\m{int}}\pCont\\
        & & \pNewLock{l_1}\pCont
        \pNewLock{l_2}\pCont
        \pNewLock{l_3}\pCont\\
        & &\left(\begin{array}{@{}l@{\ }l@{}}
               \pPar{&\pPar{\pCall{\procP}{f_1,f_2,l_1,l_2}{1}}
                          {\pCall{\procP}{f_2,f_3,l_2,l_3}{2}}\\}
                    {&\pCall{\procP}{f_1,f_3,l_1,l_3}{3}}\end{array}\right)
      \\
      \procP(f_l,f_r,l_l,l_r,id) & = & \pLock{l_l}\pCont
        \pLoad{f_l}{\nVary}\pCont
        \pSilent\pCont
        \pStore{f_l}{id}\pCont\\
        & & \pLock{l_r}\pCont
        \pLoad{f_r}{\nVarz}\pCont
        \pSilent\pCont
        \pStore{f_r}{id+2}\pCont\\
        & & \pUnlock{l_r}\pCont
        \pUnlock{l_l}\pCont
        \pCall{\procP}{f_l,f_r,l_l,l_r}{id}
    \end{array}\right\}
  }
  {\pEntryPt}
\]

Here, $\procP$ defines the behaviour of a philosopher, 
trying to get a hold of both forks assigned to him, and them release them. 
Other implementations of this problem's algorithm (including ones using channel
communications) can be found in the%
\full{Appendix.}{full version~\cite{JY20:full}.}

Next we define in the infinite branching programs a subset containing 
only programs that simulate non-deterministic branching.

\textbf{Conditional Mapping.}\label{def:ast-mapping}
The mapping $\MAPAST{(\procProg)}$ replaces all occurrences of 
marked conditionals 
$\pMPITE{\nExpr}{\procP}{\procQ}{n}$, such that $n\in\infcond(\procProg)$, 
with $\pITE{\ast}{\procP}{\procQ}$. Its reduction semantics 
follow the nondeterministic semantics of selection in types, reducing with 
a $\acttau$ label. This mapping is applicable to processes $\procP$.

\textbf{Alternating Conditionals.}\label{def:alt-cond}
We say that $\procProg$ has {\em alternating conditional branches}, written 
$\procProg\in\ac$, iff $\procProg\in\ic$ and if 
$\procProg\tra{}^{\ast}\pRes{\vGeneric}{\procP}$ then 
$\MAPAST{\procP}\wbarb{\actname}$ implies $\procP\wbarb{\actname}$.

The concurrent version of the Prime Sieve~\cite{LNTY2017,NY2016} is an example of program that has 
alternating conditionals. Our implementation of it in Go can be found in the%
\full{Appendix,}{full version~\cite{JY20:full},}%
and is not detailed here as it uses channels, which we will introduce in an 
extension to this work in \S~\ref{sect:chan-extension}. 
An other simple example of alternating 
conditionals is as follows:
\[
\procProg[\m{ac}] \coloneqq \pIn
  {
    \left\{\begin{array}{@{}l@{\,}l@{\,}l@{}}
      \pEntryProc & = & \pNewMem{\nVar}{\m{bool}}\pCont
        \pNewMem{\nVary}{\m{int}}\pCont
        \pCall{\procP}{\nVar,\nVary}{}
      \\
      \procP(b,i) & = & \pLoad{b}{\nVarz}\pCont
        \pITE{\nVarz}
          {\pLoad{i}{\nTmp}\pCont \\
           & & 
           \pStore{i}{\nTmp+1}\pCont
           \pStore{b}{\opNot{\nVarz}}\pCont
           \pCall{\procP}{b,i}{}}
          {\pStore{b}{\opNot{\nVarz}}\pCont
           \pCall{\procP}{b,i}{}}
    \end{array}\right\}
  }
  {\pEntryPt}
\]

We finally state that programs in the alternating conditionals set 
can be ensured live by ensuring that their types are live.

\begin{theorem}[Liveness]\label{the:alt-cond-live}
Suppose $\G\vdash\procProg\ts\tyEqn$ and $\tyEqn$ is live and 
$\procProg\in\ac$. Then $\procProg$ is live.
\end{theorem}

To summarise this section, we identified three classes of \gol programs for which 
we can prove liveness by proving type liveness: (1) programs that always have 
access to a terminating path (Definition~\ref{def:may-conv} and 
Proposition~\ref{prop:may-conv-live}), including the strict subset of 
programs that always terminate within a finite number of reduction steps, 
similar to our running examples; (2) programs that do not exhibit an 
infinite branch containing an infinitely occurring conditional 
(Proposition~\ref{prop:finite-branch-live}), such as the Dining Philosophers 
problem (used in our benchmarks, 
see \S~\ref{sec:implementation} for more details); and (3) programs with 
infinite branches that contain infinitely 
occurring conditionals, with the condition that these infinitely occurring 
conditionals simulate a non-deterministic choice 
(Theorem~\ref{the:alt-cond-live}), like our Prime Sieve implementation~\cite{LNTY2017,NY2016} and the example presented above.

\section{Verifying Program Properties: the Modal $\mu$-Calculus}
\label{sect:form-hb-modmu}
In this section, we introduce the modal $\mu$-calculus and express
various properties 
over the types.
We then explain how the type-level properties are transposed to
process-level properties, as proved in \S~\ref{sect:safe-live}. 

\subsection{The Modal $\mu$-Calculus}\label{sect:modal-mu}

We first define a \emph{pointed LTS} for the types, 
to denote the capabilities available at this point in the simulation.

\begin{definition}[Pointed LTS of types]\rm
  We define the pointed LTS of a program's types as:
\begin{description}
\item[A set of states] $\ltsStates$, labelled by the (restriction-less) types
accessible by reducing from the entrypoint $\tyEntryProc$ with $\tra{\acttau}$ and $\tra{\labsync{\nGeneric}}$; 
this entrypoint is defined as the type of 
the entrypoint $\pEntryProc$ of the program: 
$\ltsStates \coloneqq \{\tyT: \tyEntryProc \tra{}^{\ast}
  \tyRes{\vGeneric}{\tyT}
  \mbox{ and } \tyT \not\equiv \tyRes{\vGenerici}{\tyTi}\}
$.
\item[A set of labelled transitions] $\ltsTransitions$, 
in $\ltsStates\times\ltsStates\times\{\acttau,\labsync{\nGeneric}\}$: 
$\ltsTransitions \coloneqq \{(\tyT,\tyTi,\actname): 
\tyT,\tyTi\in\ltsStates \mbox{ and } \tyT \tra{\actname} \tyTi \}$.
\item[A set of barbs] attached to each state, describing the actions its 
labelled type can take according to the set of barbs of this type. These take the form 
of the barbs as they were defined above:
$\forall\tyT\in\ltsStates, \ltsFlags{\tyT} \coloneqq
\{\barb{\actname}\ :\ \tyT\barb{\actname}\}$. 
\end{description}
\end{definition}

The modal $\mu$-calculus is a calculus that allows to express temporal properties 
on such pointed LTS, like the fact that there exists an accessible state where some 
property is true, or the fact that some property is true in {\em all reachable states}. 
The syntax of these formulae is given below, where $\muAct$
is a set of barbs over the types available to the LTS of types, or transition actions $\acttau$ 
or $\labsync{\nGeneric}$ available as transitions to the LTS of types, as defined above:
\begin{center}
\framebox{$\begin{array}{lcl}
\muPred & \coloneqq & \muTrue
	\mid \muFalse
	\mid \muNeg{\muPred}
	\mid \muPred \muAnd \muPred
	\mid \muPred \muOr \muPred
	\mid \muPred \muImplies \muPred
	 \mid  \muAll{\muAct}{\muPred}
	\mid \muSome{\muAct}{\muPred}
	\mid \muGFP{\muVar}{\muPred}
	\mid \muLFP{\muVar}{\muPred}
	\mid \muVar\\[1mm]
\muAct  & \coloneqq & \muCol{\muAct + \muAct}
	\mid\ \muCol{\barb{\actname}} 
        \mid\ \muCol{\barb{\tilde{\actname}}}
	\mid \muCol{\acttau} \mid \muCol{\labsync{\nGeneric}}
        \mid \muSync\quad\quad  
\muSync  \coloneqq  \{\muCol{\labsync{\nGeneric}} : \nGeneric\in\fn{\tyT}\}\cup\{\muCol{\acttau}\}
\end{array}
$}
\end{center}
The formulae contain the true and false constants, negation, implication, conjunction and 
disjunction (both of which can be generalised over a set of actions, where this set can be 
restrained by some condition). 

The {\em diamond modality}, 
$\muSome{\muAct}{\muPred}$, is true when at least one of the actions in $\muAct$ is available 
from the current state and, if it is a barb then $\muPred$ must be true in the current state, 
and if it is a transition action then $\muPred$ must be true in the resulting state. 
If no action in $\muAct$ is available, then this formula is false.
For example, $\muSome{\barb{\stbarb{\tyvarz}{\occurnone}}}{\muTrue}$ holds on every state where 
a store action on $\tyvarz$ is available as the main action, but not when the only store 
action available is labelled otherwise, e.g. $\occurleft{\occurnone}$. 

The {\em box modality}, $\muAll{\muAct}{\muPred}$, is valid when, for every state reachable by 
following an action in $\muAct$ from the current state, $\muPred$ is true. This set of states can be 
empty, in case no action in $\muAct$ is available, in which case this formula is vacuously true.
For example, $\muAll{\acttau}{\muFalse}$ is true only when no $\acttau$ transition is available to 
the current state of the pointed LTS of the type.

The {\em lowest fixed point} $\muLFP{\muVar}{\muPred}$ and {\em greatest fixed point} 
$\muGFP{\muVar}{\muPred}$ are the standard recursive constructs, 
where the least fixed point is the intersection of prefixed points, 
and the greatest fixed point is the union of 
postfixed points. That implies the following properties, given for understanding:

\begin{enumerate}
\item\label{prop:lfp-empty} $\muLFP{\muVar}{\muVar} = \muFalse$: the lowest fixed point 
defaults to false; 
\item\label{prop:gfp-full} $\muGFP{\muVar}{\muVar} = \muTrue$: the greatest fixed point 
defaults to true; 
\item\label{prop:lfp-prefix} if $\muPred[] [\muVar \coloneqq \muCol{\psi}] \muImplies \muCol{\psi}$ 
then $\muLFP{\muVar}{\muPred} \muImplies \muCol{\psi}$: the lowest fixed point can be expanded on the left of a logical implication; 
\item\label{prop:gfp-postfix} if $\muCol{\psi} \muImplies \muPred[] [\muVar \coloneqq \muCol{\psi}]$ 
then $\muCol{\psi} \muImplies \muGFP{\muVar}{\muPred}$: the greatest fixed point can be expanded on the right of a logical implication.
\end{enumerate}

To express that some modal $\mu$-calculus formula $\muPred$ is true on a state labelled with 
type $\tyT$ in the LTS $\ltsT$, we say that $\tyT$ satisfies $\muPred$ in the LTS $\ltsT$, 
written $\sat[\ltsT]{\tyT}{\muPred}$.

Two key properties that can be expressed are:
$\muPred$ is {\em always true}, which means that every state $\tyT$ in $\ltsT$ satisfies 
that formula; and $\muPred$ is {\em eventually true} which means that there exists a reachable 
state that satisfies this formula. These are expressed with the fixed-point modalities 
explained above:
\begin{center}
\framebox{
Always $\muPred$\label{item:mu-always}:\quad
$\muCol{\Psi(\muPred)=}\muGFP{\muVar}{\muPred\muAnd\muAllNeg{}{\muVar}}$\qquad
Eventually $\muPred$\label{item:mu-eventually}:\quad
$\muCol{\Phi(\muPred)=}\muLFP{\muVar}{\muPred\muOr\muSomeNeg{}{\muVar}}$
}
\end{center}

\subsection{Properties of the Behavioural Types}

\figurename~\ref{table:mu-prop} defines the {\em local} properties we check on the states of the 
behavioural types LTS, which means they are defined for 
{\em one state only}. The global 
properties can be checked on the entrypoint of the LTS by checking for 
$\muCol{\Psi(\muPred)}$, ie. ``always $\muPred$''.

\begin{wrapfigure}{l}{0.55\linewidth}
\vspace{-2mm}
\scalebox{0.84}{
\framebox{
\begin{minipage}{1.1\linewidth}
\begin{enumerate}
\item\label{item:mu-mut-safe1}
Mutex safety (a):\\
$\muCol{\psi_{s_a}=}\muConj{\nLock}{\muSome{\barb{\ulckbarb{\nLock}}}{\muTrue} \muImplies \muSome{\barb{\lckmutbarb{\nLock}}}{\muTrue}}$
\\[-0.7em]
\item\label{item:mu-mut-safe2}
Mutex safety (b):\\
$\muCol{\psi_{s_b}=}\muConj{\nLock}{\muSome{\barb{\rulckbarb{\nLock}}}{\muTrue} \muImplies \muSome{\barb{\waitrmutbarb{\nLock}{i}}}{\muTrue}}$
\\[-0.7em]
\item\label{item:mu-mut-live}
Mutex liveness:\\
$\muCol{\psi_{l}=}\muConj{\nLock}{\muSome{\barb{\lckbarb{\nLock}}+\barb{\rlckbarb{\nLock}}}{\muTrue} \muImplies \muCol{\Phi\left(\muSome{\labsync{\nLock}}{\muTrue}\right)}}$
\\[-0.7em]
\item\label{item:mu-dr-free}
Data race freedom:\\
$\muCol{\psi_{d}=}\muConj{\nVar,\occurgen}{\muSome{\barb{\stbarb{\nVar}{\occurgen}}}{\muTrue} \muImplies \muAll{\sum_{\occurgeni\neq\occurgen}\barb{\stbarb{\nVar}{\occurgeni}}+\barb{\ldbarb{\nVar}{\occurgeni}}}{\muFalse}}$
\end{enumerate}
\end{minipage}
}}
\caption{Modal $\mu$-calculus properties of types}\label{table:mu-prop}
\vspace*{-5mm}
\end{wrapfigure}

Property $\muCol{\psi_{s_a}}$ checks for the first half of lock safety, that is a lock can 
only be unlocked if it is currently in locked state, and property $\muCol{\psi_{s_b}}$ checks the 
second half of lock safety, that is a read/write-lock can only be read-unlocked one level if 
it is in a read-locked state currently. 

Property $\muCol{\psi_l}$ states lock liveness, that is if a lock or read-lock action is staged, 
the same lock will eventually synchronise (and as such, when applied on a global level 
$\muCol{\Psi(\psi_l)}$, the lock or read-lock in question will eventually fire, since it becomes 
false if at any point there is a lock or read-lock staged but no future synchronisation 
on the lock). 
Remember that in our model, 
liveness of the types only entails liveness of the program 
if the program is in one of the subsets defined previously, in particular 
if the program terminates or only has alternating conditionals.

Finally, property $\muCol{\psi_d}$ checks local data race freedom, that is if a write action 
is available on some variable $\nVar$, then no other read or write action is available on 
the same variable in the current state. $\muCol{\Psi(\psi_d)}$ checks for data race freedom 
on the whole of accessible states, so checking that on the entrypoint $\typevara_{0}$ of 
a type LTS $\ltsT$ ensures the type of the associated program is data race free, and thus that 
said program is data race free.

\begin{example}\label{ex:checking}
We can check that the type $\tyTi$ from
Example~\ref{ex:typered} does not verify $\muCol{\psi_d}$:
{\footnotesize
\[
\muCol{\psi_d=} \muCol{\biggl(\muSome{\barb{\stbarb{\tyvarz}{\occurleft{\occurleft{\occurleft{\occurnone}}}}}}{\!\muTrue}
\muImplies \muAll{\barb{\stbarb{\tyvarz}{\occurleft{\occurleft{\occurright{\occurnone}}}}}+\barb{\ldbarb{\tyvarz}{\occurgeni}}}{\!\muFalse}\biggr)}
\muCol{\muAnd\biggl(\muSome{\barb{\stbarb{\tyvarz}{\occurleft{\occurleft{\occurright{\occurnone}}}}}}{\!\muTrue}
\muImplies \muAll{\barb{\stbarb{\tyvarz}{\occurleft{\occurleft{\occurleft{\occurnone}}}}}+\barb{\ldbarb{\tyvarz}{\occurgeni}}}{\!\muFalse}\biggr)}
\]}
which is false for $\tyTi$, hence
$\nsat[\ltsT_{\m{race}}]{\tyTi}{\muCol{\psi_d}}$ : locally, $\tyTi$ has a datarace.
Then $\nsat[\ltsT_{\m{race}}]{\tyEntryProc}{\muCol{\Psi(\psi_d)}}$,
meaning $\EqTypes_{\m{race}}$ has a data race, since its associated entrypoint in its 
LTS $\ltsT_{\m{race}}$ does not satisfy data race freedom property $\muCol{\Psi(\psi_d)}$.

On the other hand, the type $\EqTypes[\m{safe}]$ from
Example~\ref{ex:typered}, modelling the safe version of our running example, 
verifies the data race freedom property, as well as safety and liveness: 
\[\sat[\ltsT_{\m{safe}}]{\EqTypes[\m{safe}]}{\muCol{\Psi(\psi_d)\wedge
\Psi(\psi_l)\wedge
\Psi(\psi_{s_a}\wedge\psi_{s_b})}}\]

The types corresponding to the other examples in \S~\ref{sect:safe-live} 
($\procProg[\m{mc}],\procProg[\m{dinephil}]$ and $\procProg[\m{ac}]$) are also 
safe, live and data race free.
\end{example}

The following theorem states that type-level model-checking can justify 
process properties under the conditions given in \S~\ref{sect:safe-live}. 
We define the pointed LTS of processes $\ltsT_{\procProg}$ and the satisfaction 
property $\sat[\ltsT_{\procProg}]{\procProg}{\muPred}$ in the same way as they 
are defined for types in this section.

\begin{theorem}[Model Checking of \gol processes]
\label{the:mod-mu-proc-check}
Suppose $\G \vdash \procProg \ts \tyEqn$.
\begin{enumerate}
\item If $\sat[\ltsT_{\tyEqn}]{\tyEqn}{\muCol{\Psi(\muPred)}}$ for 
  $\muPred\in\{\muCol{\psi_{s_a}},\muCol{\psi_{s_b}},\muCol{\psi_{d}}\}$, then 
  $\sat[\ltsT_{\procProg}]{\procProg}{\muCol{\Psi(\muPred)}}$.
\item If $\sat[\ltsT_{\tyEqn}]{\tyEqn}{\muCol{\Psi(\psi_{l})}}$ and
either (a)
  $\procProg\in\convm$ or (b) $\procProg\not\in\ic$ or 
(c) $\procProg\in\ac$, then 
$\sat[\ltsT_{\procProg}]{\procProg}{\muCol{\Psi(\psi_{l})}}$. 
\end{enumerate}
\end{theorem}

\begin{proof}
By Theorems~\ref{the:proc-safe-drf} and \ref{the:alt-cond-live}, and 
Propositions~\ref{prop:may-conv-live}
and~\ref{prop:finite-branch-live}.  
\end{proof}

\section{Extending the framework for Go with channels}
\label{sect:chan-extension}
\begin{wrapfigure}[12]{l}{0.57\linewidth}
\vspace{-8mm}
{
\begin{lstlisting}[language=Go,style=golang,firstline=9]
func main() {
	var x int
	ch := make(chan int, /*<\tikz[baseline]{\node[anchor=base,inner sep=1,fill=red!30]{1}}>*/) /*<\hspace{3.8em}$\Rightarrow$ \tikz[baseline]{\node[anchor=base,inner sep=1,fill=red!30]{2}}>*/ /*<\label{line:example-makechan}>*/ 
	go /*<\label{line:example-spawn}\hskip -0.6em>*/ f(ch, &x)
	ch <- Lock // send to ch 
	x += 10    // protected by ch/*<\quad$\Rightarrow$>*/ race/*<\label{line:example-main-lock}>*/
	<-ch       // receive from ch
	ch <- Lock
	fmt.Println("x is", x)
	<-ch
}

func f(ch chan int, ptr *int) {
	ch <- Lock
	*ptr += 20 // protected by ch/*<\quad$\Rightarrow$>*/ race/*<\label{line:example-f-lock}>*/
	<-ch
}
\end{lstlisting}
}
\vspace{-4mm}
\caption{Go programs: safe (size $1$) $\Rightarrow$ race (size $2$)}
\label{fig:chan:running-ex}
\end{wrapfigure}

One of the core features of the Go language is the use of channels for 
communication in concurrent programming. 
In Go programs, a number of concurrency bugs can be caused by a 
\emph{mixture} of data races and communication problems. In this
section, we develop a theory 
which can uniformly
analyse concurrency errors caused by a mix of shared memory
accesses and asynchronous message-passing communications, 
integrating coherently our framework in 
\cite{LNTY2017,LNTY:Icse18}. 
We include channel communications as a synchronisation 
primitive in our model for data race checking, following the official Go specification. 

\indent\figurename~\ref{fig:chan:running-ex} illustrates a Go program,
which makes use of a channel \lstgo{ch} to synchronise the \lstgo{main} and \lstgo{f}
functions updating the content of the shared variable \lstgo{x}. 
On line~\ref{line:example-makechan}, the statement 
\lstgo{ch := make(chan int, $num$)} creates
a new shared channel \lstgo{ch} with a buffer size of $num$ for passing
\lstgo{int} values. Channels can be sent to or received from using the
\lstgo{<-} operator, where \lstgo{ch <- $value$} and \lstgo{<-ch} depict
sending $value$ to the channel and receiving from the channel respectively.
At runtime, sending to a full channel (i.e.~number of items in channel $\geq num$),
or receiving from an empty channel (i.e.~number of items in channel $= 0$)
blocks.
The \lstgo{go} keyword in front of a function call on
line~\ref{line:example-spawn} spawns a lightweight thread
(called a \textit{goroutine}) to execute the body of function \lstgo{f}.
The two parameters of function \lstgo{f} -- a channel \lstgo{ch}, and an
\lstgo{int} pointer \lstgo{ptr} -- are shared between the caller and callee
goroutines, \lstgo{main} and \lstgo{f}. Since concurrent access to the shared
pointer \lstgo{ptr} may introduce a data race, a pair of channel send and
receive are used to ensure serialised, mutually exclusive access to \lstgo{ptr} in
\lstgo{f} and \lstgo{x} in \lstgo{main}. If 
the buffer 
size of the shared channel is set to \lstgo{2} by mistake (as denoted
by $\Rightarrow$ in line~\ref{line:example-makechan}), allowing
simultaneous write requests to \lstgo{x} on lines~\ref{line:example-main-lock} and
\ref{line:example-f-lock}, the program could output 
``\texttt{x is 20}'' with a bad scheduling, dropping the increase of 10
in the same thread as the print statement. 
We use this program as our running example in this section.

\subsection{Channels in Processes}
\label{subsec:chan:syntax}
We add to the processes the following constructs to account for
\emph{channel actions} (defined as $\migoChanPref \coloneqq \psend{\channame}{\nExpr}
\mid \precv{\channame}{\nVar} \mid \pSilent$) 
and runtime buffer:
\begin{center}\framebox{$\begin{array}{l}\!\!\!\begin{array}{l@{\ }c@{\ }l}
\procP & \coloneqq & \ldots \mid \migoChanPref\pCont\procP \mid \close{\channame}\pCont\procP 
\mid \migoSel{\migoChanPref[i]}{\procP[i]}{i\in I} \mid \migoNewch{\channame}{\sigma}{n}\pCont\procP
\mid \buff{\vVal}{\channame}{\sigma,n}
\mid \closedbuff{\vVal}{\channame}{\sigma,n}\\
\end{array}
\end{array}$}\end{center}
Channels are ranged over by $\channamea,\channameb,\channame$, which are from 
now also included under the generic names $\nGeneric$, and sets of channels 
are ranged over by $\chanVec$.
The new syntax contains the ability to send and receive messages through channels, 
in capabilities under prefix $\migoChanPref$, and the ability to close a channel.
There is also a $\m{select}$ construct that allows selection between several 
processes guarded by channel send or receive actions, or a silent action. 
Lastly, we can create a new channel, and there are two runtime constructs denoting 
respectively open and closed channel $\channame$ with payload type $\sigma$, 
allowed buffer size $n$ and current buffered messages $\vVal$.

We add the structural congruence rules 
for queues, 
$\pRes{\channame}{\buff{\vVal}{\channame}{\sigma,n}} \equiv \zero$ 
and 
$\pRes{\channame}{\closedbuff{\vVal}{\channame}{\sigma,n}} \equiv
\zero$, 
and to the LTS the new corresponding reduction rules, along with their labels, 
shown in \figurename~\ref{fig:goldy-sem-chan}.
The rules include creating a new channel with \rulename{\ruleNewChan};
sending to and receiving from a buffered channel with 
\rulename{\ruleChanAsyncSend} and \rulename{\ruleChanAsyncReceive}; closing 
a channel with \rulename{\ruleChanClose}; synchronous communications for 
channels with buffer size 0 using rule \rulename{\ruleChanSyncCom}; and 
reducing a select construct with \rulename{\ruleSelectConstruct}.

\begin{figure}
\vspace{-5mm}
\framebox{
\begin{tabular}{l}
\hspace{-5mm}
\begin{tabular}{l|l}
$
\!\!\!\begin{array}{l}
\framebox{Channel actions}
\\
\!\!\!\begin{array}{ll}
  \rulename{\ruleChanSendReq}& 
   \psend{\channame}{\nExpr}\pCont\procP \tra{\dataAct{\ov{\channame}}{\nExpr}} \procP
  \\
  \rulename{\ruleChanReceiveReq}& 
   \precv{\channame}{\nVary}\pCont\procQ \tra{\dataAct{\channame}{\nVal}} \procQ\subs{\nVal}{\nVary}
  \\
  \rulename{\ruleChanCloseReq}& 
   \close{\channame}\pCont\procP \tra{\labclose{\channame}} \procP
  \\[0.2em]
\hline
  \rulename{\ruleChanCloseAck}& 
   \buff{\vVal}{\channame}{\sigma,n} \tra{\labclosedual{\channame}}
    \closedbuff{\vVal}{\channame}{\sigma,n}
  \\
  \rulename{\ruleChanSendAck} & 
   \inferrule
    {\buflen{\vVal} < n}
    {\buff{\vVal}{\channame}{\sigma,n} \tra{\dataAct{\astyBufSend{\channame}}{\nVal}} 
      \buff{\nVal \cdot \vVal}{\channame}{\sigma,n}}
  \\
  \rulename{\ruleChanReceiveAck} & 
   \buff{\vVal \cdot \nVal}{\channame}{\sigma,n} \tra{\dataAct{\astyBufRcv{\channame}}{\nVal}} 
    \buff{\vVal}{\channame}{\sigma, n}
  \\
  \rulename{\ruleChanClosedReceiveAck}& 
   \closedbuff{\vVal \cdot \nVal}{\channame}{\sigma,n} \tra{\dataAct{\labclsnd{\channame}}{\nVal}} 
    \closedbuff{\vVal}{\channame}{\sigma,n}
  \\
  \rulename{\ruleChanClosedDefaultValue}& 
   \closedbuff{\emptyset}{\channame}{\sigma,n} \tra{\dataAct{\labclsnd{\channame}}{\bot^\sigma}} 
    \closedbuff{\emptyset}{\channame}{\sigma,n}
\end{array}
\end{array}$
&
$\begin{array}{l}
\framebox{Synchronisation rules}
\\[0.5em]
\begin{array}{c}\rulename{\ruleChanClose}
   \inferrule
    {\procP\tra{\labclose{\channame}}\procPi \quad \procQ\tra{\labclosedual{\channame}}\procQi}
    {\pPar{\procP}{\procQ} \tra{\acttau} \pPar{\procPi}{\procQi}}
  \\
  \rulename{\ruleChanSyncCom}
   \inferrule
    {\procP \tra{\dataAct{\ov{\channame}}{\nExpr}} \procPi \quad 
     \procQ \tra{\dataAct{\channame}{\nVal}} \procQi \quad 
     \nExpr \conv \nVal}
    {\pPar{\pPPar{\procP}{\procQ}}{\buff{\emptyset}{\channame}{\sigma,0}} \tra{\labsync{\channame}} 
     \pPar{\pPPar{\procPi}{\procQi}}{\buff{\emptyset}{\channame}{\sigma,0}}}
  \\
  \rulename{\ruleChanAsyncSend}
   \inferrule
    {\procP \tra{\dataAct{\ov{\channame}}{\nExpr}} \procPi \quad 
     \procQ \tra{\dataAct{\astyBufSend{\channame}}{\nVal}} \procQi \quad
     \nExpr \conv \nVal}
    {\pPar{\procP}{\procQ}\tra{\labsync{\channame}}\pPar{\procPi}{\procQi}}
  \\
  \rulename{\ruleChanAsyncReceive} 
   \inferrule
    {\procP \tra{\dataAct{\channame}{\nVal}} \procPi \quad 
     \procQ \tra{\dataAct{\astyBufRcv{\channame}}{\nVal}} \procQi \mbox{ or } 
     \procQ \tra{\dataAct{\labclsnd{\channame}}{\nVal}} \procQi}
    {\pPar{\procP}{\procQ} \tra{\labsync{\channame}} \pPar{\procPi}{\procQi}}
\\
\rulename{\ruleSelectConstruct}
  \inferrule
  {\pPar{\migoChanPref[j]\pCont\procP[j]}{\procP} \tra{\actgen} \procPi \quad 
   \actgen\in\{\acttau,\labsync{\channame}\}}
  {\pPar{\migoSel{\migoChanPref[i]}{\procP[i]}{i\in I}}{\procP} \tra{\actgen} \procPi}
\\[1.5em]
\end{array}\hspace*{-5mm}
\end{array}\hspace*{-5mm}$
\end{tabular}\hspace*{-5mm}
\\
\hline
\\[-0.7em]
$\hspace*{-5mm}\begin{array}{l}
\begin{array}{l}
\framebox{Runtime creation}
\hspace*{4em}\rulename{\ruleNewChan}\
  {\migoNewch{\nVary}{\sigma}{n}\pCont\procP \tra{\acttau}
   \migoRes{\channame}{\pPPar{\procP\subs{\channame}{\varnamey}}{\buff{\emptyset}{\channame}{\sigma,n}}} 
   \quad ({\channame \notin \fn{\procP}})}
\end{array}\hspace*{-5mm}
\end{array}$
\end{tabular}
}
\caption{Remaining LTS Semantics of Processes.}\label{fig:goldy-sem-chan}
\vspace{-5mm}
\end{figure}

\begin{example}[Processes from \figurename~\ref{fig:chan:running-ex}] 
\label{ex:syntax-chan}
The following process represents the unsafe version of the code in 
\figurename~\ref{fig:chan:running-ex}. 
As in Example~\ref{ex:syntax}, we 
separate the \lstgo{main} function in two parts, 
the part that instantiates the variable and channel, and spawns the side process in parallel to the 
continuation; and two called processes $\procP$ and $\procQ$. 
\[
\procProg[\m{c-race}] \coloneqq
    \pIn
	{\left\{\begin{array}{@{}l@{\ }l@{\ }l}
	  \pEntryProc & = & \pNewMem{\nVar}{\m{int}}\pCont
	   \migoNewch{\channame}{\m{int}}{{\color{blue}2}}\pCont
	   \pPPar{\pCall{\procP}{\nVar,\channame}{}}
	         {\pCall{\procQ}{\nVar,\channame}{}}\\
	  \procP(\nVary,\nVarz) & = & \psend{\nVarz}{\m{Lock}}\pCont
	   \pLoad{\nVary}{\nTmp[1]}\pCont
	   \pStore{\nVary}{\nTmp[1]+10}\pCont
	   \precv{\nVarz}{\nGeneric[1]}\pCont\\
	  & & \psend{\nVarz}{\m{Lock}}\pCont
	   \pLoad{\nVary}{\nTmp[2]}\pCont
	   \pSilent\pCont
	   \precv{\nVarz}{\nGeneric[2]}\pCont
	   \zero\\
          \procQ(\nVary,\nVarz) & = & \psend{\nVarz}{\m{Lock}}\pCont
	   \pLoad{\nVary}{\nTmp[0]}\pCont
	   \pStore{\nVary}{\nTmp[0]+20}\pCont
	   \precv{\nVarz}{\nGeneric[0]}\pCont
	\zero
	\end{array}\right\}}
     {\pEntryPt}
\]

The safe version $\procProg[\m{c-safe}]$ is the same, replacing the {\color{blue}2} 
for a {\color{blue}1} in the channel instanciation.

This example reduces like the one with a rwlock, allowing to see the possible data race:
\[
\begin{array}{r@{\ }c@{\ }l}
\procProg[\m{c-race}] 
  & \tra{}^{6} & \pRes{\nVar\nLock}{
      \left(\begin{array}{@{}l@{\ }l}
        \pPar{\pPar{\pPar{
	  & 
	   \pStore{\nVar}{10}\pCont
	   \precv{\channame}{\nGeneric[1]}\pCont\\
	  & \ \psend{\channame}{\m{Lock}}\pCont
	   \pLoad{\nVar}{\nTmp[2]}\pCont
	   \pSilent\pCont
	   \precv{\channame}{\nGeneric[2]}\pCont
	   \zero\\}{
          & 
	   \pStore{\nVar}{20}\pCont
	   \precv{\channame}{\nGeneric[0]}\pCont
	\zero}}
	  {\stMem{\nVar}{\m{int}}{0}}}
	 {\buff{\m{Lock}\cdot\m{Lock}}{\channame}{\m{int},2}}
	\end{array}
      \right)} = \pRes{\nVar\nLock}{\procPi}
\end{array}
\]
\end{example}

\subsection{Liveness and Safety for Channels}
\label{subsect:chan-live-safe}

To define the liveness and safety properties for channels, we 
first extend the barbs as follows:
\begin{definition}[Process barbs]\label{def:barb-chan}\rm
The barbs are expanded as follows:
\begin{description}
 \item[prefix actions:] 
$\precv{\channame}{\varname}\barb{\channame};
\ 
\psend{\channame}{\nExpr}\barb{\ov{\channame}}$.
 \item[select:] we add the rule:
$\inferrule
  {\forall i\in\{1,...,n\}:\migoChanPref[i]\migoCont\migoP[i]\tra{\actname[i]}\migoP[i] \wedge \actname[i]\neq\acttau}
  {\migoSel{\migoChanPref[i]}{\migoP[i]}{i\in\{1,...,n\}}\barb{\{\actname[1],...,\actname[n]\}}}$
\end{description}
\end{definition}
The rest is unchanged, but takes into account end actions, as well as buffer actions.

Next is extending the safety and liveness properties to channels, by adding 
the following definitions: (1) \textbf{\emph{Channel Safety}}: 
A channel can be 
closed only once, and when closed should not be used to send a message. A 
closed channel can be used to receive an unbounded number of times though, 
and will yield a default value of the channel's type when the queue is
empty; and 
(2) 
\textbf{\emph{Channel Liveness}}: 
no channel action blocks indefinitely, ie. 
all channel actions lead to synchronisation on the channel eventually (or on a 
channel of the list of guarding actions for a select construct that has no 
silent action guard).

\begin{definition}[Channel Safety]~\label{def:chansafe}\rm
Program $\procProg$ 
is {\em channel safe} if for all $\procP$ such that 
$\procProg\tra{}^*\migoRes{\vGeneric}{\procP}$, if 
$\procP\barb{\cclbarb{\channame}}$ 
then 
$\neg(\procP\wbarb{\clbarb{\channame}})$ and 
$\neg (\procP\wbarb{\ov{\channame}})$.
\end{definition}

\begin{definition}[Channel Liveness]~\label{def:chanlive}\rm
Program $\procProg$ satisfies \emph{channel liveness} if for all $\procP$ such that
$\procProg \tra{}^* \migoRes{\vGeneric}{\procP}$, 
(a) if ${\procP}\barb{\channame}$ or 
${\procP}\barb{\ov{\channame}}$ 
then 
$\procP\wbarb{\labsync{\channame}}$; and 
(b) if ${\procP}\barb{\actnameVec}$ 
then $\procP\wbarb{\labsync{\channame[i]}}$ for some $\channame[i]\in
\fn{\actnameVec}$.
\end{definition}

\begin{figure}[h]
\vspace{-5mm}
\begin{center}
\framebox{$
\hspace*{0.4em}
\begin{array}{r@{\ }l@{\qquad}r@{\ }l}
\hbrulename{\ruleSendReceiveAsync}&
\inferrule
{\procP\barb{\channame}
  \quad \procP\barb{\ov{\channame}}
  \quad \buflen{\vVal} = n}
{\hb[\pPar{\procP}{\buff{\vVal}{\channame}{\sigma,n}}]{\channame}{\ov{\channame}}}
&
\hbrulename{\ruleSendSelectReceiveAsync}&
\inferrule
{\procP\barb{\ov{\channame}}
  \quad \exists j\in I:\migoChanPref[j]\barb{\channame}
  \quad \buflen{\vVal} = n}
{\hb[\pPar{\pPPar{\procP}{\migoSel{\migoChanPref[i]}{\procQ[i]}{i \in I}}}{\buff{\vVal}{\channame}{\sigma,n}}]{\channame}{\ov{\channame}}}
\\[1.2em]
\hbrulename{\ruleReceiveFromClosed}&
\inferrule
{\procP\barb{\channame}}
{\hb[\pPar{\procP}{\closedbuff{\vVal}{\channame}{\sigma,n}}]{\labclsnd{\channame}}{\channame}}
&
\hbrulename{\ruleSelectSendReceiveAsync}&
\inferrule
{\procP\barb{\channame}
  \quad \exists j\in I:\migoChanPref[j]\barb{\ov{\channame}}
  \quad \buflen{\vVal} = n}
{\hb[\pPar{\pPPar{\procP}{\migoSel{\migoChanPref[i]}{\procQ[i]}{i \in I}}}{\buff{\vVal}{\channame}{\sigma,n}}]{\channame}{\ov{\channame}}}
\\[1.2em]
\hbrulename{\ruleSendReceiveSync}&
\inferrule
{\procP\barb{\channame}
  \quad \procP\barb{\ov{\channame}}}
{\hb[\pPar{\procP}{\buff{\emptyset}{\channame}{\sigma,n}}]{\ov{\channame}}{\channame}}
&
\hbrulename{\ruleSelectSendReceiveSync}&
\inferrule
{\procP\barb{\channame}
  \quad \exists j\in I:\migoChanPref[j]\barb{\ov{\channame}}}
{\hb[\pPar{\pPPar{\procP}{\migoSel{\migoChanPref[i]}{\procQ[i]}{i \in I}}}{\buff{\emptyset}{\channame}{\sigma,n}}]{\ov{\channame}}{\channame}}
\\[1.2em]
\hbrulename{\ruleCloseBeforeDefault}&
\inferrule
{\procP\barb{\labclose{\channame}}}
{\hb[\procP]{\labclose{\channame}}{\labclsnd{\channame}}}
&
\hbrulename{\ruleSendSelectReceiveSync}&
\inferrule
{\procP\barb{\ov{\channame}}
  \quad \exists j\in I:\migoChanPref[j]\barb{\channame}}
{\hb[\pPar{\pPPar{\procP}{\migoSel{\migoChanPref[i]}{\procQ[i]}{i \in I}}}{\buff{\emptyset}{\channame}{\sigma,n}}]{\ov{\channame}}{\channame}}
\end{array}
$}
\\[1mm]
We omit the symmetric rules for most rules ending in a parallel process
$\procP\parr \procQ$.
\end{center}
\vspace{-3mm}
\caption{Rest of Go's Happens-Before Relation}\label{fig:hbrel-chan}
\vspace{-4mm}
\end{figure}

The channel synchronisations for the happens-before relation are listed in 
\figurename~\ref{fig:hbrel-chan}. They consist of channel communication 
according to the official Go memory model: a send happens-before the 
corresponding receive, and if the channel buffer size is $n$, then the $k$-th 
receive happens-before the $k+n$-th send. We add on top of that that closing 
a channel happens-before any default value is received from it, and when a 
channel is closed, default values are emmited by the closed buffer before the 
corresponding receive reads it.

We extend our behavioural types with the following constructs, mirroring process constructs, 
and using the syntax and semantics from \cite{LNTY2017,LNTY:Icse18}:

\begin{center}\framebox{$\begin{array}{l}\!\!\!\begin{array}{l@{\ }c@{\ }l}
\tyS,\tyT & \coloneqq & \ldots \mid \kappa\tyCont\tyT \mid \labclose{\channame}\tyCont\tyT 
   \mid \tbr{}{\kappa_i \tyCont \tyT[i]}{i\in I} \mid
   \astynew{\channame}{n}{\tyT}
\mid \astyOpenbuf{\channame}{k}{n}
   \mid \tyclosedbuffer{\channame} \hfill
   \quad\quad \kappa \coloneqq \tysend{\channame} \mid \tyrecv{\channame} \mid \tySilent
\end{array}
\end{array}$}\end{center}
We show the typing rules for added channel constructs, 
which contain the new type primitives, in \figurename~\ref{fig:typing-chan}.
\begin{figure}[ht]
\vspace{-6mm}
\begin{center}\framebox{
\begin{tabular}{l}
$
\!\!\!\!\begin{array}{l}
\framebox{$\G \vdash \procP \ts \tyT$}
\hfill
\trulename{\ruleNewChan}
\inferrule
{\G\ctxComp\typing{\nVary}{\m{ch}(\sigma,n)} \vdash \procP \ts \tyT \quad \channame \not\in \domain{\G} \cup \fn{\tyT}}
{\G \vdash {\migoNewch{\nVary}{\sigma}{n}}\pCont\procP \ts \astynew{\channame}{n}\tyT\subs{\channame}{\nVary}}
\\[1.5em]
\trulename{\ruleChanSendReq}
\inferrule
{\G \vdash \typing{\channame}{\m{ch}(\sigma,n)} \quad \G \vdash \typing{\nExpr}{\sigma} \quad \G \vdash \procP \ts \tyT}
{\G \vdash \psend{\channame}{\nExpr}\pCont\procP \ts \tsend{\channame}{\sigma}\tyCont\tyT}
\hfill
\trulename{\ruleChanReceiveReq}
\inferrule
{\G \vdash \typing{\channame}{\m{ch}(\sigma,n)} \quad \G\ctxComp\typing{\nVar}{\sigma} \vdash \procP \ts \tyT}
{\G \vdash \precv{\channame}{\nVar}\pCont\procP \ts \trecv{\channame}{\sigma}\tyCont\tyT}
\\[1.5em]
\trulename{\ruleSelectConstruct}
\inferrule
{\G \vdash \migoChanPref[i]\pCont\procP[i] \ts \typrefix[i]\tyCont\tyT[i]}
{\G \vdash \migoSel{\migoChanPref[i]}{\procP[i]}{i\in I} \ts \tbr{}{\typrefix[i]\tyCont\tyT[i]}{i\in I}}
\hfill
\trulename{\ruleChanCloseReq}
\inferrule
{\G \vdash \procP \ts \tyT}
{\G \vdash \close{\channame}\pCont\procP \ts \End[\channame]\tyCont\tyT}
\\[1.1em]
\hline
\\[-0.7em]
\framebox{$\G \vdash_B \procP \ts \tyT$}
\hspace*{3.4em}
\trulename{\ruleTypingChan}
\inferrule
{\G \vdash \typing{\channame}{\m{ch}(\sigma,n)} \quad \lvert \vVal \rvert = k}
{\G \vdash_{\{\channame\}}\buff{\vVal}{\channame}{\sigma,n} \ts \astyOpenbuf{\channame}{k}{n} }
\hspace*{3.4em}
\trulename{\ruleTypingChanClosed}
\inferrule
{\G \vdash \typing{\channame}{\m{ch}(\sigma,n)}}
{\G \vdash_{\{\channame\}}\closedbuff{\vVal}{\channame}{\sigma} \ts \tyclosedbuffer{\channame}}
\end{array}\!\!\!\!
$
\end{tabular}
}
\caption{Typing Rules for Channels.~\label{fig:typing-chan}}
\vspace{-2mm}
\end{center}
\end{figure}

We also add the structure rules 
$\tyRes{\channame}{\astyOpenbuf{\channame}{k}{n}} \equiv \zero$
and 
$\tyRes{\channame}{\tyclosedbuffer{\channame}} \equiv \zero$; and  
the LTS semantics for the communication primitives (\figurename~\ref{fig:type-sem-chan}). They correspond to the ones found for the processes.
\begin{figure}[ht]
\vspace{-4mm}
\begin{center}
\framebox{
\begin{tabular}{l}
\hspace*{-5mm}
\begin{tabular}{l|l}
$
\!\!\!\begin{array}{l}
\framebox{Channel actions}
\\
\!\!\!\begin{array}{r@{\ }l}
    \ltsrulename{\ruleChanSendReq}& 
    {\tsend{\channame}{\sigma}\tyCont\tyT}\tra{\labtysnd{\channame}}{\tyT}
\\
    \ltsrulename{\ruleChanReceiveReq}& 
    {\trecv{\channame}{\sigma}\tyCont\tyT}\tra{\labtyrcv{\channame}}{\tyT}
\\
    \ltsrulename{\ruleChanCloseReq}&
     {}
     {\End[\channame]\tyCont\tyT}\tra{\labclose{\channame}}{\tyT}
\\[0.2em]
\hline
     \ltsrulename{\ruleChanClosedDefaultValue}&
     {}
     {\tyclosedbuffer{\channame}}
         \tra{\labclsnd{\channame}}{\tyclosedbuffer{\channame}}
\\
     \ltsrulename{\ruleChanCloseAck}&
     {}
     {\astyOpenbuf{\channame}{k}{n}}
         \tra{\labclosedual{\channame}}{\tyclosedbuffer{\channame}}
\\
     \ltsrulename{\ruleChanReceiveAck}&
     \inferrule
     {k \geq 1}
     {{\astyOpenbuf{\channame}{k}{n}}
         \tra{\astyBufRcv{\channame}}{\astyOpenbuf{\channame}{k-1}{n}}}
\\
     \ltsrulename{\ruleChanSendAck}&
     \inferrule
     {k < n}
     {{\astyOpenbuf{\channame}{k}{n}}
         \tra{\astyBufSend{\channame}}{\astyOpenbuf{\channame}{k+1}{n}}}
\hspace*{4.2em}~
\\[1.5em]
\end{array}\hspace*{-3em}
\\
\hline
\\[-0.7em]
\framebox{Runtime creation}
\\
    \ltsrulename{\ruleNewChan}\ 
    {\astynew{\channame}{n}\tyT \tra{\acttau}
    \tyRes{\channame}{\tyPPar{\tyT}{\astyOpenbuf{\channame}{0}{n}}}}
\end{array}\hspace*{-.5em}$
&
$\begin{array}{l}
\framebox{Synchronisation rules}
\\[0.5em]
\begin{array}{c}
     \ltsrulename{\ruleChanClose}
     \inferrule
     {\tyT \semty{\labclose{\channame}} \tyTi 
       \quad
       \tyS \semty{\labclosedual{\channame}} \tySi
     }
     {\tyPar{\tyT}{\tyS}\tra{\acttau}
     \tyPar{\tyTi}{\tySi}}
\\[1.1em]
    \ltsrulename{\ruleChanSyncCom}
    \inferrule
    {\tyT \semty{\labtysnd{\channame}} \tyTi 
      \quad
      \tyS \semty{\labtyrcv{\channame}} \tySi
    }
    {\tyPar{\tyPPar{\tyT}{\tyS}}{\astyOpenbuf{\channame}{0}{0}}\tra{\labsync{\channame}}
    \tyPar{\tyPPar{\tyTi}{\tySi}}{\astyOpenbuf{\channame}{0}{0}}}
\\[1.1em]
     \ltsrulename{\ruleChanAsyncSend}
     \inferrule
     {\tyT \semty{\labtysnd{\channame}} \tyTi 
      \quad
      \tyS \semty{\astyBufSend{\channame}} \tySi
     }
     {\tyPar{\tyT}{\tyS}\tra{\labsync{\channame}}
      \tyPar{\tyTi}{\tySi}}
\\[1.1em]
     \ltsrulename{\ruleChanAsyncReceive}
     \inferrule
     {\tyT \semty{\labtyrcv{\channame}} \tyTi 
      \quad
      \tyS \semty{\actname} \tySi
      \quad
      \actname\in\{\astyBufRcv{\channame},\labclsnd{\channame}\}
     }
     {\tyPar{\tyT}{\tyS}\tra{\labsync{\channame}}
      \tyPar{\tyTi}{\tySi}}
\\[1.1em]
    \ltsrulename{\ruleSelectConstruct}
    \inferrule
    {\tyPar{\typrefix_j\tyCont\tyT[j]}{\tyT} \semty{\actname}
      \tyTi \quad \actname\in\{\acttau,\labsync{\channame}\}}
    {\tyPar{\tbr{}{\typrefix_i \tyCont \tyT[i]}{i\in I}}{\tyT} 
      \tra{\actname} \tyTi}
\\[1.5em]
\end{array}\hspace*{-2em}
\end{array}$
\end{tabular}
\end{tabular}
}
\caption{Remaining LTS Semantics of Types.}\label{fig:type-sem-chan}
\end{center}
\vspace{-9mm}
\end{figure}

All results in \S~\ref{sec:types} hold as-is with the new definitions. We only 
add the new barbs, like for processes (identical definition), and the
following 
type properties:

\begin{definition}[Channel Safety]~\label{def:type-safe-chan}\rm
Type $\tyEqn$ 
is {\em channel safe} if for all $\tyT$ such that 
$\tyEqn\tra{}^*\tyRes{\vGeneric}{\tyT}$, 
if 
$\tyT\barb{\cclbarb{\channame}}$ 
then 
$\neg(\tyT\wbarb{\clbarb{\channame}})$ and 
$\neg (\tyT\wbarb{\ov{\channame}})$.
\end{definition}

\begin{definition}[Channel Liveness]~\label{def:type-live-chan}\rm
Type $\tyEqn$ is \emph{channel live} if for all $\tyT$ such that
$\tyEqn \tra{}^* \tyRes{\vGeneric}{\tyT}$,
(a) if ${\tyT}\barb{\channame}$ or 
${\tyT}\barb{\ov{\channame}}$ 
then 
$\tyT\wbarb{\labsync{\channame}}$; and 
(b) if ${\tyT}\barb{\actnameVec}$ 
then $\tyT\wbarb{\labsync{\channame[i]}}$ for some $\channame[i]\in
\fn{\actnameVec}$.
\end{definition}
\noindent They correspond to
the ones added for processes, and are integrated in 
other theorems of \S~\ref{sec:types}.

\subsection{Modal $\mu$-Calculus Properties for Channels}
\label{subsect:mod-mu-chan}

\begin{wrapfigure}{l}{0.55\linewidth}
\vspace{-3mm}
\scalebox{0.84}{
\framebox{
\begin{minipage}{\linewidth}
\begin{enumerate}
\item\label{item:mu-chan-safe}
Channel safety:\\
$\muCol{\psi_{s}=}\muConj{\channame}{\muSome{\barb{\cclbarb{\channame}}}{\muTrue} \muImplies \muCol{\Psi\left(\muAll{\barb{\ov{\channame}}+\barb{\clbarb{\channame}}}{\muFalse}\right)}}$
\\[-0.7em]
\item\label{item:mu-chan-live1}
Channel liveness (a):\\
$\muCol{\psi_{l_a}=}\muConj{\channame}{\muSome{\barb{\channame}+\barb{\ov{\channame}}}{\muTrue} \muImplies \muCol{\Phi\left(\muSome{\labsync{\channame}}{\muTrue}\right)}}$
\\[-0.7em]
\item\label{item:mu-chan-live2}
Channel liveness (b):\\
$\muCol{\psi_{l_b}=}\muConj{\tilde{\channame}}{\muSome{\barb{\tilde{\channame}}}{\muTrue} \muImplies \muCol{\Phi\left(\muSome{\sum_{\channame[i]\in\tilde{\channame}}\labsync{\channame[i]}}{\muTrue}\right)}}$
\end{enumerate}
\end{minipage}
}}
\caption{Modal $\mu$-calculus properties for channels}\label{table:mu-prop-chan}
\vspace*{-4mm}
\end{wrapfigure}

With extending to the channel primitives, 
all definitions in \S~\ref{sect:form-hb-modmu} 
still hold with added properties in the modal $\mu$-calculus for channel liveness and safety.
These are defined in \figurename~\ref{table:mu-prop-chan}.

The model-checking result is also extended as the following
theorem to capture the situation where shared memory and
message passing co-exist. 

\begin{theorem}[Model Checking of \gol processes]
\label{the:mod-mu-proc-check-full}
Suppose $\G \vdash \procProg \ts \tyEqn$.
\begin{enumerate}
\item If $\sat[\ltsT_{\tyEqn}]{\tyEqn}{\muCol{\Psi(\muPred)}}$ for 
  $\muPred\in\{\muCol{\psi_{s_a}},\muCol{\psi_{s_b}},\muCol{\psi_{s}},\muCol{\psi_{d}}\}$, then 
  $\sat[\ltsT_{\procProg}]{\procProg}{\muCol{\Psi(\muPred)}}$.
\item If $\sat[\ltsT_{\tyEqn}]{\tyEqn}{\muCol{\Psi(\muPred)}}$ for 
  $\muPred\in\{\muCol{\psi_{l}},\muCol{\psi_{l_a}},\muCol{\psi_{l_b}}\}$ and
either (a)
  $\procProg\in\convm$ or (b) $\procProg\not\in\ic$ or 
(c) $\procProg\in\ac$, then 
$\sat[\ltsT_{\procProg}]{\procProg}{\muCol{\Psi(\muPred)}}$. 
\end{enumerate}
\end{theorem}

This extension to our framework allows us not only to 
integrate the previous framework by 
\cite{LNTY2017,LNTY:Icse18}, but also show to some extent the 
modularity of our memory-based approach. 
With channels, 
this extension of \gol is implementing a significant range of the 
concurrency features of Go, allowing for a range of programs to be model-checked 
for data races, liveness issues and other safety issues in the use of locks and channels.

\subsection{Types and process (program) liveness}
\label{subsec:gap}
There are several categories of processes for which
the equivalence between types and process (program) liveness is not ensured:
(3) programs that have an infinite conditional that is not an alternating 
    conditional, {\em if they do not always have a termination path available}. 
    They can be checked by the model checker if they are not in (3),
    however the result may not coincide with the process liveness; 
(2) programs that neither have an infinite conditional, nor always have a 
    potential path for termination (e.g. a program that recurses indefinitely 
    without ever having an ending branch available through a \texttt{select} construct, without 
    the need of a conditional in the recursing selection); and 
(3) programs that are not finite control -- 
    i.e. programs that spawn an unbounded amount of new processes -- 
    because the model-checker 
    will not be able to generate a linear representation of them (see \S~\ref{sec:implementation}).
    
Note that for (1) and (2), the tool
returns ``live'' if the types are live, 
though it may be the case that the programs are not live.

\section{Implementation and Evaluation}\label{sec:implementation}
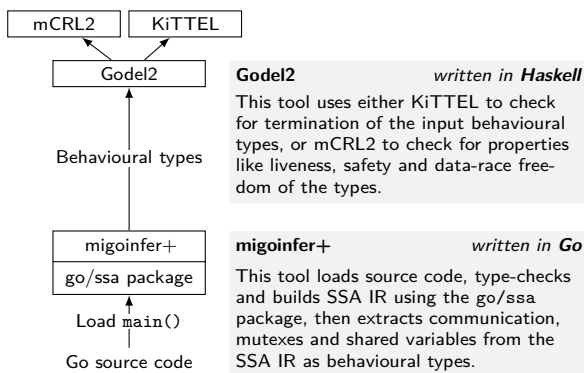
\begin{wrapfigure}{l}{0.55\linewidth}
\vspace*{-6mm}
\begin{tikzpicture}[
    mtool/.style={draw,text width=5em,text centered},
    tool/.style={draw,text width=7em,text centered},
    arrowlabel/.style={inner sep=0,fill=pagecolor},
    comment/.style={text width=\linewidth-12em,anchor=north west,
    fill=gray!10 
    }
  ]
  \scriptsize\sffamily
  \node (godel) [tool] {\godeltwo};
  \node (mcrl) [above left=0.3 and -0.9 of godel,mtool] {mCRL2};
  \node (kittel) [right=0.1 of mcrl,mtool] {KiTTEL};
  \node (migo) [rectangle split,rectangle split parts=2,below=1.9 of godel,tool]
        {migoinfer+ \nodepart{two}\scriptsize{\sf go/ssa} package};
  \node (source) [below=0.7 of migo]
        {Go source code};
  \path [draw,->,>=latex] (source) -- node [arrowlabel] {Load \texttt{main()}} (migo);
  \path [draw,->,>=latex] (migo)  -- node [arrowlabel] {Behavioural types} (godel);
  \path [draw,->,>=latex] (godel) -- (mcrl);
  \path [draw,->,>=latex] (godel) -- (kittel);
  \node [rectangle split,rectangle split parts=2,right=0.3 of godel.north east,comment]
        {\textbf{\godeltwo}\hfill\textit{written in \textbf{Haskell}}\nodepart{two}
         This tool uses either KiTTEL to check for termination of the
         input behavioural types, or mCRL2 to check for properties
         like liveness, safety and data-race freedom of the types.};
  \node [rectangle split,rectangle split parts=2,right=0.3 of migo.north east,comment]
        {\textbf{migoinfer+} \hfill\textit{written in \textbf{Go}}\nodepart{two}
         This tool loads source code, type-checks and builds SSA IR using the
         \texttt{go/ssa} package, then extracts communication, mutexes and
         shared variables from the SSA IR as behavioural types.};
\end{tikzpicture}
\vspace{-6.5mm}
  \caption{Workflow of the verification toolchain.\label{fig:workflow}}
\vspace{-3mm}
\end{wrapfigure}
{\textbf{The tool chain.\ }
Our implementation tool 
(shown in \figurename~\ref{fig:workflow}) 
consists of a type inference tool 
and
a type verifier. 
The type inference tool (\textsf{migoinfer+})~\cite{web:migoinfer} 
extracts behavioural types, including 
eight new primitives related to shared memory:
creating a new lock (called mutex in the tool, in reference to the name of the 
mutual exclusion lock implementation in Go)
or shared address, exclusive write-locking or unlocking of a
lock or a read-write lock, read-locking/unlocking a read-write lock, and
reading or writing a shared variable. 
This new inference tool supports both
channel-based communication primitives from \cite{LNTY:Icse18} and shared memory
primitives. 

\textsf{migoinfer+} currently supports a subset of the Go language syntax,
extracting only variables and mutexes created explicitly inside the body of a function, and
does not support embedding or mutexes in $\mathtt{struct}$. These usage 
patterns of mutexes can be transformed to the flat representation we support, 
allowing us to analyse the examples in our benchmark~\cite{web:benchmark}. 
Note that it is advised to avoid the non-declared sharing of
variables, channels and mutexes to a 
nameless child goroutine, as it may not extract 
the parameter passing properly, and 
this is a good practice in Go to specify shared parameters. 
Programs that spawn an unbounded number of goroutines such as our 
prime-sieve example 
can be extracted by \textsf{migoinfer+} if 
they respect the above limitations.
Lastly, the use of some (non-default) packages, such as the \texttt{net} package, is known to break 
\textsf{migoinfer+} under certain conditions, 
making it not extract the types correctly.

The type verifier (\textsf{\godeltwo})~\cite{web:godel2} analyses the new extracted
primitives, implements the theory presented in this paper, and uses
the mCRL2~\cite{web:mcrl2-website,mcrl2-book} model checker as a backend
to check safety and data race properties. 
Regarding the liveness properties, 
as discussed 
after Theorem~\ref{the:prog} and  
in \cite{LNTY2017,LNTY:Icse18}, 
liveness of types does not imply liveness of processes,
due to conditionals behaving differently in the types and the processes. 
In Theorem \ref{the:mod-mu-proc-check}, we 
identified the three classes of Go programs where 
both liveness properties coincide. 
One such class  
is a set of terminating processes, as defined in 
Definition~\ref{def:term-prog}, which is a strict subset of 
may converging processes (Proposition~\ref{prop:term-is-may-conv}).   
To make sure liveness coincides on types and processes, 
we combine the termination checker 
KITTeL~\cite{web:kittel-github} to our tool (see also \cite[\S~5]{LNTY:Icse18}).
This tool can check processes that are not 
terminating under certain conditions, 
namely they should not spawn an unbounded number of threads. 
However, such programs 
may, in rare cases, lead to false positives or negatives regarding liveness 
(and possibly safety), because of the approximations the model checker has to 
make when running against models with cycles.

\noindent{\textbf{Evaluations.\ } We evaluate our tool for reference on an 8-core Intel i7-7700K machine
with 16 GB memory, in a 64-bit Linux environment running go 1.12.2.
\tablename~\ref{tbl:benchmark} shows the results for a range of programs
that mix shared memory with either channels or mutexes as locking mechanism. 
The sources for those examples can be found in the benchmark repository~\cite{web:benchmark}.
Programs \texttt{no-race} and \texttt{simple-race} are programs made to
test the behaviour of mutexes and check that liveness errors are properly reported.
The channel version of our running example, from \figurename~\ref{fig:chan:running-ex} is named
\texttt{channel-as-lock}, and \texttt{channel-as-lock-bad} is a variation of the
\texttt{-fixed} version but with channel sends and receive switched, hence
the program deadlocks on the first attempt to lock of each thread as there is nothing to
receive. 

\begin{wraptable}{l}{0.65\linewidth}
\vspace*{-3mm}
\caption{Go Programs Verified by the Toolchain.}\label{tbl:benchmark}
\noindent
\renewcommand{\arraystretch}{1}
{\hspace*{-4mm}
\footnotesize
  \renewcommand{\tabcolsep}{0.15cm}
\begin{tabular}{l|lr|cccr}
  \toprule
  Programs &
  LoC &
  Sum &
  Safe &
  Live &
  DRF &
  \hspace{-2mm}time (ms)\\
  \midrule
  \texttt{no-race}
    &15&9&\yes&\yes&\yes&691.45\\
  \texttt{no-race-mutex}
    &24&33&\yes&\yes&\yes&785.57\\
  \texttt{no-race-mut-bad}
    &23&20&\yes&\noo&\yes&721.77\\
  \texttt{simple-race}
    &13&8&\yes&\yes&\noo&701.93\\
  \texttt{simple-race-fix}
    &19&17&\yes&\yes&\yes&731.73\\
  \texttt{deposit-race}\textsuperscript{1}
    &18&14&\yes&\yes&\noo&697.90\\
  \texttt{deposit-fix}\textsuperscript{1}
    &24&27&\yes&\yes&\yes&727.43\\
  \texttt{ch-as-lock-race}\textsuperscript{2}
    &19&20&\yes&\yes&\noo&753.99\\
  \texttt{ch-as-lock-fix}\textsuperscript{2}
    &19&20&\yes&\yes&\yes&745.64\\
  \texttt{ch-as-lock-bad}
    &19&20&\yes&\noo&\yes&749.97\\
  \texttt{prod-cons-race}
    &38&156&\yes&\yes&\noo&1,903.52\\
  \texttt{prod-cons-fix}
    &40&188&\yes&\yes&\yes&1,971.26\\
  \texttt{dine5-unsafe}
    &35&106&\noo&\yes&\yes&6,996.27\\
  \texttt{dine5-deadlock}
    &35&106&\yes&\noo&\yes&12,278.33\\
  \texttt{dine5-fix}
    &35&106&\yes&\yes&\yes&8,998.04\\
  \texttt{dine5-chan-race}
    &59&2672&\yes&\yes&\noo&$\sim$ 185mn\\
  \texttt{dine5-chan-fix}
    &59&2688&\yes&\yes&\yes&$\sim$ 645mn\\
\bottomrule
\end{tabular}\\[1mm]
}
  \textsuperscript{1}{\cite{go-book}},
  \textsuperscript{2}{\figurename~\ref{fig:chan:running-ex}},
LoC: Lines of Code,
DRF: Data Race Free,\\
Sum: Summands,
  \yes: Formula is true,
  \noo: Formula is false
\vspace{-4mm}
\end{wraptable}

The \texttt{deposit} implementation is taken from~\cite{go-book}
(the example to present data races and locking mechanisms), and
\texttt{prod-cons} is a shared memory implementation of the
classic producer-consumer algorithm, where two producers race against each other
and one consumer takes whichever product is available first. In this
example, all three
threads share a single memory heap, supposed to be protected by a mutex.
Finally, \texttt{dine5} is an implementation of the Dining Philosophers problem as 
explained in \S~\ref{sect:safe-live}, and 
\texttt{dine5-chan} is a channel variant adapted slightly 
to allow for a potential shared-memory 
data race%
\full{(code in Appendix~\ref{app:example}).}{\!\!.}

We note that the Prime Sieve algorithm~\cite{LNTY2017,NY2016} is not analysed by our tool, as 
it continually spawns new threads, making the state space too big for the {mCRL2} model-checker.

Future work for applying this approach to 
real-world Go programs are:  working around the explosion seen with \texttt{select}+channels 
in \texttt{dine5-chan}, for which using a different model for \texttt{select} constructs and channel 
actions than the one in our implementation might be sufficient; 
working on the  
implementation for a wider range of extractions for channels, shared memory 
and mutexes embedded in $\mathtt{struct}$s, or to implement a parser that flattens those structs upstream 
of \textsf{migoinfer+}; and
working on analysis of programs that dynamically spawn new goroutines 
-- this would require non-trivial approximations to be leveraged. 
Note that it should represent only a small fraction of programs, as most daily-use 
protocols should be implementable without the need for such
unbounded growth in memory usage.

All examples in Table~\ref{tbl:benchmark} are analysed by our tool, and the time
given as an indication scales exponentially with the number of summands 
(and possibly action labels) and their ordering, in
the linear process specification that represents the types in the model checker.
Those directly depend from the source code of the analysed program.

\section{Conclusion and Related Work}\label{sec:related}

The Go language provides a unique programming environment where both explicit
communication and shared memory concurrency primitives co-exist.
This work introduces \gol 
as an abstraction layer for Go code, as well as behavioural types to propose a 
static verification framework for detecting concurrency bugs in Go. 
These include deadlocks and safety for both mutual 
exclusion locks and channel communication, as well as 
data race detection for shared memory primitives. 

Shared memory locks and channels cover by themselves a substantial 
amount of Go's concurrency features. The former is 
a low-level, standard library provision
and the latter is a high-level, built-in 
language feature. 
Go only features these two basic building blocks because one can use them to implement
most higher levels of concurrency abstraction, for example actors models.

The works~\cite{LNTY2017,LNTY:Icse18} built behavioural types 
for verification of 
concurrency bugs for channel-based message passing. 
We integrate with their asynchronous calculus 
(a.k.a.~\asyncfgo) for our channel-related extension in 
\S~\ref{sect:chan-extension}. 
These works, however, were lacking more 
shared memory concurrency with locks and shared pointers, and did not 
tackle data races for shared pointers, which we do. It does not 
study happens before relations either (for channels). 
It furthermore was lacking 
complete proofs on their equivalence theorems for liveness, which is also 
addressed in this paper. We also proved  
\gol satisfies  
the properties of the types characterised by the modal $\mu$-calculus
(Theorems \ref{the:mod-mu-proc-check},\ref{the:mod-mu-proc-check-full}).
The paper \cite{LNTY:Icse18} has informally 
described them, but these have never been formalised nor proved. 

The work~\cite{Stadmuller16} defines forkable behaviours (ie. regular 
expressions with a fork construct) to capture goroutine spawning in synchronous 
Go programs. They develop a tool based on this model to analyse directly Go
programs. Their approach is sound, but suffers from several
limitations, which were overcome by ~\cite{LNTY2017,LNTY:Icse18}; 
their tool does not treat shared memory concurrency primitives and locks. 

The work~\cite{pretendsync} observed that asynchronous distributed systems can
be verified by only modelling \textit{synchronisations} in the core protocol,
and introduces a language \textsc{\sf IceT} similar to \gol for specifying
synchronisation in message-passing programs. Their focus was to verify
functional correctness of the input protocol, and requires 
input programs to be synchronisable (i.e.~no deadlocks nor spurious sends
in the input programs). Their approach allows 
for checking correctness of an implementation, given a reasonable amount 
of annotations.  
It is orthogonal to our work in which we only need to check for 
runtime sanity. Both approaches independently benefit the user, 
and should be run individually on testing code in order to check both for 
concurrency behavioural bugs and for implementation bugs.

Recent works 
\cite{Tu2019,go-survey} provide empirical studies of Go programs,
which show that almost half of concurrency bugs in Go are non-blocking bugs,
mostly shared memory problems, and the remaining blocking bugs are mostly 
related to channel and lock misuse. 
That gives an incentive to make tools and implementations
built on the concurrent behavioural theory,  
for easy detection of such bugs. Our work is part of that effort.

A large body of race detection tools targeting other languages such 
as Java are available. ThreadSanitizer 
(TSan)~\cite{tsan,web:tsan,SPIV11:tsan} 
which is included in LLVM/Clang is one of the most widely 
deployed dynamic race detectors. The runtime 
race detector of Go~\cite{web:go-race-detector} 
uses TSan's runtime library.  

The work \cite{NanoGo:2018} proposes a subset of the Go language
akin to~\gol, along with a modular
approach to statically analyse processes. Their approach combines
lattice-valued regular expressions and a shuffle operator allowing
for separate analysis of single threads, and they prove their
theory to be sound. They have a prototype implementation
in OCaml 
to check deadlocks in synchronous message-passing programs. 
The work \cite{CHJNY2019} uses a protocol description 
language, Scribble \cite{scribble}, which is a practical incarnation of 
multiparty session types \cite{HYC08} to generate Go APIs,   
ensuring deadlock freedom and liveness of communications 
by construction. 
Neither \cite{NanoGo:2018} nor \cite{CHJNY2019}
treat either communication error
or data race detection, both handled in this paper, nor do they treat 
shared variables, which our approach extends upon.

The main difference in code writing between Go and \gol is the handling of 
continuations for select and if-then-else constructs, where Go allows for 
standard continuation while \gol restrains the user to use tail calls. 
This is handled by our extraction tool, as it extracts the Go 
code to \gol by building an SSA representation before extracting relevant 
primitives from it, see \figurename~\ref{fig:workflow} in \S~\ref{sec:implementation}.

The idea to use the LTS of behavioural types for 
programming analysis 
dates back to \cite{Nielson1994HCP} 
for Concurrent ML, and since then, it has been applied to many works 
\cite{ABBCCDGGGHJMMMNNPVY2017}. Some tackle mutual exclusion locks, but 
systematically lack support for read-write mutual exclusion locks, including 
works \cite{Igarashi:2001:GTS,AFF:Java-Race06,HBK19}.
The work~\cite{Kobayashi:2008} aims to guarantee
liveness with termination of a typed $\pi$-calculus.  
We study wider classes in the theory, aiming termination to use the existing 
tool ({KITTeL}) in order to integrate with our tool-chain
to scale -- thus the main 
aim and the target (real {Go} programs in our case) differ from \cite{Kobayashi:2008}. 

Type-level model-checking for message-passing 
programming was first addressed in \cite{DBLP:conf/popl/ChakiRR02}. 
Recent applications using mCRL2 include 
verifications of multiparty session typed $\pi$-calculus \cite{SY2019} and 
the Dotty programming language (the future Scala 3)  \cite{SYB2019}. 

Our future works include 
studying the soundness and completeness of the happens-before relation 
provided by the Go memory model, ie. studying if the definition of data race 
given by it covers all data races that can happen in Go, and whether it does 
not provide false positives; 
speeding-up the analysis 
using more mCRL2 options and the extension to 
an incremental analysis based on happens-before relations, as  
taken in other languages, e.g.~\cite{d4,zh16:echo}; as well as possibly 
counter-example extraction for code failing verification, to provide direct 
access to the detected bugs to developers. There is also the 
possibility to work on handling dynamic process creation, widening the 
analysis scope of our current tool and model.



\bibliography{main}

\full{
\appendix
\section{Additional Definitions}
\label{app:def}
This section lists the additional definitions. 

\subsection{Free Names}
\label{app:freename}
\figurename~\ref{fig:names} lists 
the set of free names. 

\begin{figure}
\framebox{
\begin{tabular}{l|l}
$
\!\!\begin{array}{r@{\ }c@{\ }l}
\fn{\zero} & = & \emptyset\\
\fn{\close{\nGeneric}\pCont\procP} & = & \fn{\procP}\cup\{\nGeneric\}\\
\fn{\sel{\migoChanPref[i]\pCont\procP[i]}{i\in I}} & = & \bigcup_{i\in I}{\fn{\migoChanPref[i]\pCont\procP[i]}}\\
\fn{\pITE{\nExpr}{\procP}{\procQ}} & = & \fn{\nExpr} \cup \fn{\procP} \cup \fn{\procQ}\\
\fn{\pPar{\procP}{\procQ}} & = & \fn{\procP} \cup \fn{\procQ}\\
\fn{\migoNewch{\channame}{\sigma}{n}\pCont\procP} & = & \fn{\procP} \backslash \{\channame\}\\
\fn{\pNewMem{\nVar}{\sigma}\pCont\procP} & = & \fn{\procP} \backslash \{\nVar\}\\
\fn{\pRes{\nGeneric}{\procP}} & = & \fn{\procP} \backslash \{\nGeneric\}\\
\fn{\pCall{\procX}{\vExpr}{\vGeneric}} & = & \fn{\vExpr}\cup\fn{\vGeneric}\\
\fn{\buff{\vVal}{\channame}{\sigma,n}} & = & \{\channame\}\cup\fn{\vVal}\\
\fn{\closedbuff{\vVal}{\channame}{\sigma,n}} & = & \{\channame\}\cup\fn{\vVal}\\
\fn{\stMem{\nVar}{\sigma}{\nVal}} & = & \fn{\nVal}\cup\{\nVar\}\\
\fn{\pNewLock{\nLock}\pCont\procP} & = & \fn{\procP} \backslash \{\nLock\}\\
\fn{\pNewRWLock{\nLock}\pCont\procP} & = & \fn{\procP} \backslash \{\nLock\}\\
\fn{\stLock{\nLock}} = \fn{\stLockL{\nLock}} & = & \{\nLock\}\\
\fn{\stRWLock{\nLock}{i}} = \fn{\stRWLockL{\nLock}{i}} & = & \{\nLock\}\\
\fn{\stRWLockW{\nLock}{i}} & = & \{\nLock\}
\end{array}\hspace*{0.4em}
$
&
$\hspace*{0.4em}
\begin{array}{r@{\ }c@{\ }l}
\fn{\pLock{\nLock}\pCont\procP} & = & \fn{\procP}\cup\{\nLock\}\\
\fn{\pUnlock{\nLock}\pCont\procP} & = & \fn{\procP}\cup\{\nLock\}\\
\fn{\pRLock{\nLock}\pCont\procP} & = & \fn{\procP}\cup\{\nLock\}\\
\fn{\pRUnlock{\nLock}\pCont\procP} & = & \fn{\procP}\cup\{\nLock\}\\
\fn{\pLoad{\nVar}{\nVary}\pCont\procP} & = & (\fn{\procP}\cup\{\nVar\})\backslash\{\nVary\}\\
\fn{\pStore{\nVar}{\nExpr}\pCont\procP} & = & \fn{\procP} \cup \fn{\nExpr} \cup \{\nVar\}\\
\fn{\psend{\channame}{\nExpr}\pCont\procP} & = & \fn{\procP} \cup \fn{\nExpr} \cup \{\channame\}\\
\fn{\precv{\channame}{\nVary}\pCont\procP} & = & (\fn{\procP} \cup \{\channame\}) \backslash \{\nVary\}\\
\fn{\pSilent\pCont\procP} & = & \fn{\procP}\\
\fn{\num} & = & \emptyset\\
\fn{\true} = \fn{\false}  & = & \emptyset\\
\fn{\mathsf{not}(\nExpr)} & = & \fn{\nExpr}\\
\fn{\mathsf{succ}(\nExpr)} & = & \fn{\nExpr}\\
\fn{\nVar} & = & \{\nVar\}\\
\fn{\channame} = \fn{\ov{\channame}} & = & \{\channame\}\\
\fn{\nLock} & = & \{\nLock\}\\
\fn{\{\actname[1],\ldots,\actname[n]\}} & = & \bigcup_{1\leq i\leq n}\fn{\actname[i]}
\end{array}\!\!
$
\end{tabular}
}
\caption{Definition of free names.}\label{fig:names}
\vspace{-0.5cm}
\end{figure}

\section{Proofs}
\label{app:proofs}

{\noindent\large\bf Theorem~\ref{the:drace}.} 

We first prove the if-direction.
\begin{proof}
Suppose $\procProg\tra{}^{\ast}\pRes{\vGeneric}{\procP}$,
$\procP\barb{\stbarb{\nVar}{\occurgen}}\wedge \procP\barb{\actname[2]}$ with $\actname[2]=\stbarb{\nVar}{\occurgeni},\ldbarb{\nVar}{\occurgeni}$, 
and prove that $\nothb[\procP]{\stbarb{\nVar}{\occurgen}}{\actname[2]}\wedge\nothb[\procP]{\actname[2]}{\stbarb{\nVar}{\occurgen}}$.
We suppose that $\hb[\procP]{\stbarb{\nVar}{\occurgen}}{\actname[2]}$ or $\hb[\procP]{\actname[2]}{\stbarb{\nVar}{\occurgen}}$, and prove 
there is a contradiction by induction on the happens-before relation.
Barring the structural congruence, there are four rules that allow for read and write
events to shared variables, so we only need to take these four into
account. We prove only the case for $\hb[\procP]{\stbarb{\nVar}{\occurgen}}{\actname[2]}$, as the other case is symmetric:
\begin{itemize}
\item Rule \hbrulename{con} is the base case, with $\procP=\prefMem\pCont\procPi$, $\prefMem\barb{\stbarb{\nVar}{\occurgen}}$
and $\procPi\barb{\actname[2]}$. The contradiction here is that then, $\neg \procP\barb{\actname[2]}$.
\item Rules \hbrulename{res} and \hbrulename{par} suppose $\procP=\pRes{\nGeneric}{\procPi}$ and $\procP=\procPi\parr \procQ$
respectively, with $\hb[\procPi]{\stbarb{\nVar}{\occurgen}}{\actname[2]}$, so the induction is on $\procPi$ to prove
that $\neg \procPi\barb{\actname[2]}$.
\item Rule \hbrulename{tra} supposes there is an intermediate action $\actname$, with
$\hb[\procP]{\stbarb{\nVar}{\occurgen}}{\actname}$ and $\hb[\procP]{\actname}{\actname[2]}$. We only need to prove $\neg \procP\barb{\actname[2]}$,
which we get by induction on the second relation $\hb[\procP]{\actname}{\actname[2]}$.
\end{itemize}
Which ends the proof by contradiction and induction.
\end{proof}

Next we prove the only-if-direction.
\begin{proof}
Suppose $\procProg\tra{}^{\ast}\pRes{\vGeneric}{\procP}$,
$\procP\wbarb{\stbarb{\nVar}{\occurgen}}\wedge \procP\wbarb{\actname[2]}$ with $\actname[2]=\stbarb{\nVar}{\occurgeni},\ldbarb{\nVar}{\occurgeni}$
and $\nothb[\procP]{\stbarb{\nVar}{\occurgen}}{\actname[2]}\wedge\nothb[\procP]{\actname[2]}{\stbarb{\nVar}{\occurgen}}$,
and prove that there exists $\procP\tra{}^{\ast}\procPi$ such that
$\procPi\barb{\stbarb{\nVar}{\occurgen}}\wedge \procPi\barb{\actname[2]}$.
We suppose that for all $\procP\tra{}^{\ast}\procPi$ such that
$\procPi\wbarb{\stbarb{\nVar}{\occurgen}}\wedge \procPi\wbarb{\actname[2]}$, it holds that
$\neg \procPi\barb{\stbarb{\nVar}{\occurgen}}\vee \neg \procPi\barb{\actname[2]}$ and prove there is a contradiction
by induction on the structure of $\procPi$ and the happens-before relation.
We can safely suppose that $\procPi\barb{\stbarb{\nVar}{\occurgen}}$ and $\neg \procPi\barb{\actname[2]}$
(or the symmetric case), as if none holds, because
$\procPi\wbarb{\stbarb{\nVar}{\occurgen}}\wedge \procPi\wbarb{\actname[2]}$ , we can reduce until one of them holds.
\begin{itemize}
\item if $\procPi=\pStore{\nVar}{\nExpr}\pCont\procPii$, with necessarily $\procPii\wbarb{\actname[2]}$, then
by construction of the happens-before relation we have $\hb[\procPi]{\stbarb{\nVar}{\occurgen}}{\actname[2]}$,
thus $\hb[\procP]{\stbarb{\nVar}{\occurgen}}{\actname[2]}$, which contradicts the hypothesis.
\item if $\procPi=\pPar{\procP[1]}{\procP[2]}$, then if both actions come from the same $\procP[i]$, the induction
on this $\procP[i]$ proves the contradiction, if not we can suppose $\procP[1]\barb{\stbarb{\nVar}{\occurgen[1]}}$
and $\pPar{\procP[1]}{\procP[2]}\wbarb{\actname[2]}$ with $\occurgen=\occurleft{\occurgen[1]}$ and $\occurgeni=\occurright{\occurgeni[2]}$.\\
Then for all $\procP[2]\tra{}^{\ast}\procPi[2]$ 
such that $\pPar{\procP[1]}{\procPi[2]}\wbarb{\actname[2]}$,
$\neg \procPi[2]\barb{\actnamei[2]}$ with $\actnamei[2]=\stbarb{\nVar}{\occurgeni[2]},\ldbarb{\nVar}{\occurgeni[2]}$
(otherwise the contradiction is on both actions happening at the same time).\\
Suppose $\procP[2]\tra{}^{\ast} \procPi[2] \not \rightarrow$ 
and 
$\pPar{\procP[1]}{\procPi[2]}\wbarb{\actname[2]}$.
Then $\procPi[2]\barb{\actname}$ with 
$\actname\in\{\channame,\ov{\channame},\lckbarb{\nLock},\rlckbarb{\nLock}\}$ one of the blocking
actions. We can suppose it is the last blocking action before $\actname[2]$ on this branch,
as we can just induce on the chain of blocking actions if not.
Because it is the last blocking action, it cannot be unlocked by $\procP[1]$ unless it reduces
the suspended $\stbarb{\nVar}{\occurgen[1]}$ action (otherwise we can reach a state where
both $\stbarb{\nVar}{\occurgen}$ and $\actname[2]$ can execute at the same time).\\
By this, we have $\hb[{\procP[1]}]{\stbarb{\nVar}{\occurgen[1]}}{\actnamei}$ and $\hb[{\procP[2]}]{\actname}{\actnamei[2]}$, with $\actnamei$ the action that unblocks $\actname$ in $\procP[2]$ ($\actnamei=\ov{\channame}$ if $\actname=\channame$, $\actnamei=\channame$ if $\actname=\ov{\channame}$, $\actnamei=\lckbarb{\nLock}$ if
$\actname=\rlckbarb{\nLock}$, and $\actnamei\in\{\ulckbarb{\nLock},\rulckbarb{\nLock}\}$ if $\actname=\lckbarb{\nLock}$).
Since $\actnamei$ is unblocking $\actname$, we also have the corresponding rule from happens-before that
holds at some point in the reduction, so $\hb[\pPar{\procP[1]}{\procPi[2]}]{\actnamei}{\actname}$, which with two
applications of the transitivity rule gives $\hb[\pPar{\procP[1]}{\procPi[2]}]{\stbarb{\nVar}{\occurgen}}{\actname[2]}$,
thus $\hb[\procPi]{\stbarb{\nVar}{\occurgen}}{\actname[2]}$ and finally $\hb[\procP]{\stbarb{\nVar}{\occurgen}}{\actname[2]}$, contradicting
the hypothesis.
\end{itemize}
\end{proof}

{\noindent\large\bf Proposition~\ref{prop:cong}.}
\begin{proof}
Suppose $\G \vdash \procP \ts \tyT$ and $\procP\equiv \procPi$.
Then from the typing judgments, we have either:
\begin{itemize}
\item $\procP=\pPar{\procPi}{\zero} \equiv \procPi$, then we can just remove
the last \trulename{\rulePar} rule along with its child
\trulename{\ruleTypingZero} rule: $\tyT=\tyPar{\tyTi}{\zero} \equiv \tyTi$.
\item $\procP=\pPar{\procP[1]}{\procP[2]} \equiv \pPar{\procP[2]}{\procP[1]}=\migoPi$.
Then $\tyT=\tyPar{\tyT[1]}{\tyT[2]}$ and we can use $\tyTi=\tyPar{\tyT[2]}{\tyT[1]}\equiv\tyT$,
the last rule of the typing judgment being \trulename{\rulePar} with
the premises switched.
\item $\procP=\pPar{\procP[1]}{\pPPar{\procP[2]}{\procP[3]}}\equiv\pPar{\pPPar{\procP[1]}{\procP[2]}}{\procP[3]}=\procPi$.
Then $\tyT=\tyPar{\tyT[1]}{\tyPPar{\tyT[2]}{\tyT[3]}}$ and we can use $\tyTi=\tyPar{\tyPPar{\tyT[1]}{\tyT[2]}}{\tyT[3]}
\equiv\tyT$, the bottom of the tree looking like the following:
\[
\trulename{\rulePar}
\inferrule
{\G \vdash \procP[1] \ts \tyT[1] \quad
	\trulename{\rulePar}
	\inferrule
	{\G \vdash \procP[2] \ts \tyT[2] \quad
		\G \vdash \procP[3] \ts \tyT[3]}
	{\G \vdash \pPar{\procP[2]}{\procP[3]} \ts \tyPPar{\tyT[2]}{\tyT[3]}}}
{\G \vdash \pPar{\procP[1]}{\pPPar{\migoP[2]}{\migoP[3]}} \ts \tyPar{\tyT[1]}{\tyPPar{\tyT[2]}{\tyT[3]}}}
\]
which turns into:
\[
\trulename{\rulePar}
\inferrule
{\trulename{\rulePar}
	\inferrule
	{\G \vdash \procP[1] \ts \tyT[1] \quad
	\G \vdash \procP[2] \ts \tyT[2]}
	{\G \vdash \pPar{\procP[1]}{\procP[2]} \ts \tyPPar{\tyT[1]}{\tyT[2]}}\quad
	\G \vdash \migoP[3] \ts \tyT[3]}
{\G \vdash \pPar{\pPPar{\procP[1]}{\procP[2]}}{\procP[3]} \ts \tyPar{\tyPPar{\tyT[1]}{\tyT[2]}}{\tyT[3]}}
\]
\end{itemize}
If we suppose $\G \vdash_{B} \procP \ts \tyT$ and $\procP\equiv \procPi$, the
proof is of the same form, with congruence and typing rules
about name restriction. The congruence rules from processes
are matched exactly with congruence rules from the typing system,
and the only typing rules involved other than the ones above are
\trulename{\ruleTypingPar}, \trulename{\ruleRestrictChan}, 
\trulename{\ruleRestrictMem} and \trulename{\ruleRestrictMut}.
\end{proof}

{\noindent\large\bf Theorem~\ref{the:red}.}
\begin{proof}
Suppose $\G \vdash \procP \ts \tyT$ and $\procP\tra{} \procPi$. We prove by induction on
the reductions semantics from the processes that there exists $\tyTi$ such
that $\G \vdash \procPi \ts \tyTi$ and $\tyT \tra{} \tyTi$.
We first treat the base cases:

\begin{itemize}
\item Rule \rulename{\ruleChanSyncCom} corresponds to the send case of 
\ltsrulename{\ruleChanSyncCom}.
\item Rule \rulename{\ruleChanAsyncSend} corresponds to 
\ltsrulename{\ruleChanAsyncSend}.
\item Rule \rulename{\ruleChanAsyncReceive} corresponds to 
\ltsrulename{\ruleChanAsyncReceive} in the pop case,
and \ltsrulename{\ruleChanSyncCom}'s closed channel case in the closed case.
\item Rules \rulename{\ruleMemLoad} and \rulename{\ruleMemStore} both fit with 
rule \ltsrulename{\ruleMemAction}, as there is no need for matching data in the 
types word.
\item Rules \rulename{\ruleMutLock}, \rulename{\ruleMutUnlock}, 
\rulename{\ruleRMutRLock} and \rulename{\ruleRMutRUnlock} respectively 
correspond to \ltsrulename{\ruleMutLock}, \ltsrulename{\ruleMutUnlock}, 
\ltsrulename{\ruleRMutRLock} and \ltsrulename{\ruleRMutRUnlock}, 
and \rulename{\ruleRMutLockStage} corresponds to 
\ltsrulename{\ruleRMutLockStage}.
\item Finally rules \rulename{\ruleITETrue} and \rulename{\ruleITEFalse} both 
translate to rule \ltsrulename{\ruleITEConstruct}, reducing with $\acttau$.
\item Rule \rulename{\ruleSelectConstruct} translates to rule 
\ltsrulename{\ruleSelectConstruct}.
\end{itemize}
All other base cases correspond 
one-on-one to their similarly-named rule in the types semantics.
See now the induction cases:
\begin{itemize}
\item Rule \rulename{\ruleDefinition} uses \ltsrulename{\ruleDefinition}, 
with \ltsrulename{\ruleParLeft} ot \ltsrulename{\ruleParRight}
in the case of a $\acttau$, $\labsync{\nLock}$ or memory action, by induction hypothesis;
or with one of the synchronisation rules from above in other cases,
by induction hypothesis as well.
\item Rules \rulename{\ruleParLeft} and \rulename{\ruleParRight} correspond to 
\ltsrulename{\ruleParLeft} and \ltsrulename{\ruleParRight} by direct induction.
\item Rules \rulename{\ruleRestrictFree} and \rulename{\ruleRestrictBind} 
correspond to \ltsrulename{\ruleRestrictFree} and 
\ltsrulename{\ruleRestrictBind}, both cases calling induction hypothesis.
\item Rule \rulename{\ruleAlphaEquiv} uses \ltsrulename{\ruleAlphaEquiv} 
directly with induction hypothesis as well.
\end{itemize}
\end{proof}

{\noindent\large\bf Theorem~\ref{the:prog}.}

To prove this main theorem, we first prove the following lemma. 

\begin{lemma}[Correspondence of Barbs on Types and
  Processes]\label{lem:barbs}
Suppose $\G \vdash \procP \ts \tyT$. Then 
if $\tyT\barb{\actname}$ with
$\actname \neq \heapbarb{\nVar}$, then $\procP\barb{\actname}$.
If $\tyT\barb{\heapbarb{\nVar}}$
then $\procP\barb{\labloaddual{\nVar}}$ or $\procP\barb{\labloaddual{\nVar}}$.
\end{lemma}
\begin{proof}
Suppose $\G \vdash \procP \ts \tyT$, and $\tyT\barb{\actname}$.
Let us case on $\actname$, and then for each case prove the conclusion by induction
on the structure of $\tyT$.
\begin{itemize}
\item if $\actname=\heapbarb{\nVar}$, then $\tyT$ can be:
	\begin{itemize}
	\item $\tyT=\tyheap{\nVar}$, then $\procP=\stMem{\nVar}{\sigma}{\nVal}$ and
	    it follows $\procP\barb{\labstoredual{\nVar}}$ and 
	    $\procP\barb{\labloaddual{\nVar}}$.
	\item $\tyT=\tyPar{\tyT[1]}{\tyT[2]}$ and $\tyT[1]\barb{\actname}$,
		by the typing rules there exist $\procP[1]$ and $\procP[2]$ such that
		$\G\vdash \procP[i]\ts \tyT[i]$ for $i=1,2$ and $\procP=\pPar{\procP[1]}{\procP[2]}$;
		then by induction on $\tyT[1],\procP[1]$ we get
		$\procP[1]\barb{\labloaddual{\nVar}}$ or $\procP[1]\barb{\labstoredual{\nVar}}$,
		and finally $\procP\barb{\labloaddual{\nVar}}$ or $\procP\barb{\labstoredual{\nVar}}$.
	\item $\tyT=\tyRes{\nGeneric}{\tyT[0]}$ and $\nGeneric\notin\fn{\actname}$, then there exists $\procP[0]$
		such that $\G\vdash \procP[0]\ts \tyT[0]$ and $\procP=\pRes{\nGeneric}{\procP[0]}$;
		by induction on $\tyT[0],\procP[0]$ we get
		$\procP[0]\barb{\labloaddual{\nVar}}$ or $\procP[0]\barb{\labstoredual{\nVar}}$,
		and finally $\procP\barb{\labloaddual{\nVar}}$ or $\procP\barb{\labstoredual{\nVar}}$.
	\item $\tyT=\tdefcall{\tyCallee}{\vGeneric}$, then
		$\tyT[0]\subs{\vGeneric}{\vVar}
		\barb{\actname}$ and
		$\tyCallee(\vVar) = \tyT[0]$; we have
		$\procP=\pCall{\procX}{\vExpr,\vGeneric}{}$, $\procX(\vVar)=\procQ$
		and so $\G\vdash \procQ\subs{\vExpr,\vGeneric}{\vVar}
		\ts \tyT[0]\subs{\vGeneric}{\vVar}$,
		so by induction 
		$\procQ\subs{\vExpr,\vGeneric}{\vVar}
		\barb{\labloaddual{\nVar}}$ or 
		$\procQ\subs{\vExpr,\vGeneric}{\vVar}
		\barb{\labstoredual{\nVar}}$, and finally
		$\procP\barb{\labloaddual{\nVar}}$ or $\procP\barb{\labstoredual{\nVar}}$.
	\item $\tyT\equiv_\alpha \tyT[0]$ and $\tyT[0]\barb{\actname}$,
		then as $\alpha$-conversion only renames bound variables we
		have $\G\vdash \procP\ts \tyT[0]$ and we can use the former points on
		$\tyT[0],\procP$, such that $\procP\barb{\labloaddual{\nVar}}$ or $\procP\barb{\labstoredual{\nVar}}$.
	\end{itemize}
\item if $\actname$ is anything else, the barb construction rules from the types and the
processes correspond exactly, with each types construct corresponding to the process
that can fire actions the same way.
\end{itemize}
\end{proof}
Now we prove the main theorem. 

\begin{proof}
Suppose $\G \vdash \procP \ts \tyT$ and $\tyT\tra{\actname}\tyTi$ with $\actname=\acttau,
\labsync{\nGeneric}$.
\begin{itemize}
\item if $\actname = \labsync{\nGeneric}$, then by structural congruence we can make sure
	the two actions that can sync are directly parallel to each other in
	a subprocess $\procP[0]$ of a process $\procPi\equiv \procP$,
	providing the ability for $\procP[0]$ to reduce, thus $\procPi$ can reduce,
	and finally $\procP$ can reduce.
\item if $\actname = \acttau$, then, $\tyT$ can be:
	\begin{itemize}
	\item $\tyT=\tySilent\tyCont\tyTi$, then
		$\procP=\pSilent\pCont\procPi\tra{} \procPi$.
	\item $\tyT=\tbr{}{ \typrefix[i]\tyCont\tyT[i]}{i\in I}$ and for a certain $j$,
		$\typrefix[j]=\tySilent$, then $\procP=\sel{\migoChanPref[i]
		\pCont\procP[i]}{i\in I}$
		and $\migoChanPref[j]=\pSilent$,
		then $\procP\tra{}\procP[j]$.
		Note that in this case, $\tyT$ and $\procP$ do not have a barb, because
		barbs for select constructs are only defined when no prefix is a 
		$\pSilent$ prefix.
	\item $\tyT=\tch{}{\tyT[1], \tyT[2]}{}\tra{\acttau}\tyT[j]$, then $\procP=\pITE{\nExpr}{\procP[1]}{\procP[2]}$ and,
depending on the value of $\nExpr$, $\procP\tra{}\procP[1]$ and $\tyTi=\tyT[1]$ or $\procP\tra{}\procP[2]$ and $\tyTi=\tyT[2]$.
	\item $\tyT=\tyPar{\tyT[1]}{\tyT[2]}$ and $\tyT[1]\tra{\acttau}\tyTi[1]$,
		by the typing rules there exists $\procP[1]$ and $\procP[2]$ such that
		$\G\vdash \procP[i]\ts \tyT[i]$ for $i=1,2$ and $\procP=\pPar{\procP[1]}{\procP[2]}$;
		then by induction on $\tyT[1],\procP[1]$ we get $\procPi[1]$ such that
		$\procP[1]\tra{}\procPi[1]$, and finally
		$\procP\tra{}\procPi=\pPar{\procPi[1]}{\procP[2]}$.
	\item Rules \ltsrulename{\ruleNewChan}, \ltsrulename{\ruleNewMem}, 
	    \ltsrulename{\ruleNewMut}
		and \ltsrulename{\ruleNewRMut} all correspond to the same constructs in
		the process world and reduce with the corresponding rules.
	\item Rule \ltsrulename{\ruleChanClose} corresponds to its similarly-named 
	    rule \rulename{close} as well.
	\item $\tyT=\tyRes{\nGeneric}{\tyT[0]}$ and $\tyT[0]\tra{\acttau}\tyTi[0]$, then there exists $\procP[0]$
		such that $\G\vdash \procP[0]\ts \tyT[0]$ and $\procP=\pRes{\nGeneric}{\procP[0]}$;
		by induction on $\tyT[0],\procP[0]$ we get
		$\procP[0]\tra{}\procPi[0]$, and finally
		$\procP\tra{}\procPi=\pRes{\nGeneric}{\procPi[0]}$.
	\item $\tyT=\tyRes{\nGeneric}{\tyT[0]}$ and $\tyT[0]\tra{\labsync{\nGeneric}}\tyTi[0]$, then there exists $\procP[0]$
		such that $\G\vdash \procP[0]\ts \tyT[0]$ and $\procP=\pRes{\nGeneric}{\procP[0]}$;
		by induction on $\tyT[0],\procP[0]$ using the sync case, we get
		$\procP[0]\tra{}\procPi[0]$, and finally
		$\procP\tra{}\procPi=\pRes{\nGeneric}{\procPi[0]}$.
	\item $\tyT=\tdefcall{\typevara}{\vGeneric}$, then
		$\typesemF{\tyT[0]}\subs{\vGeneric}{\vVar}
		\semty{\acttau}\typesemFdash{\tyTi}$ and
		$\typevara(\vVar) = \tyT[0]$; we have
		$\procP=\defcall{\procX}{\vExpr,\vGeneric}$, $\procX(\vVar)=\procQ$
		and so $\G\vdash \procQ\subs{\vExpr,\vGeneric}{\vVar}
		\ts \typesemF{\tyT[0]}\subs{\vGeneric}{\vVar}$,
		so by induction $\procQ\subs{\vExpr,\vGeneric}{\vVar}
		\tra{}\procQi$ for some $\procQi$, and finally
		$\procP\tra{}\procPi$ for some $\procPi$ using \rulename{\ruleDefinition}.
	\item $\tyT\equiv_\alpha \tyT[0]$ and $\tyT[0]\tra{\acttau}\tyTi$,
		then as $\alpha$-conversion only renames bound variables we
		have $\G\vdash \procP\ts \tyT[0]$ and we can use the former points on
		$\tyT[0],\procP$, such that $\procP\tra{}\procPi$ for some $\procPi$ such that
		$\G\vdash \procPi\ts \tyTi$.
	\end{itemize}
\end{itemize}
\end{proof}

{\noindent\large\bf Theorem~\ref{the:proc-safe-drf}.}

We need to formalise the Inversion Lemma in our model 
in order to prove this Theorem:

\begin{lemma}[Inversion]\label{lem:inv}
\begin{enumerate}
\item If $\G \vdash_{B} \procP \ts \tyT$ and 
  $\procP \equiv \pRes{\name}{\procPi}$ then 
  $\tyT \equiv \tyRes{\name}{\tyTi}$, with 
  $\Gi \vdash_{B'} \procPi \ts \tyTi$ for some $\Gi$ and $B'$, with 
  $\G\subseteq\Gi$ and $B\subseteq B'$.
\item If $\G \vdash_{B} \procP \ts \tyT$ and 
  $\procP \equiv \pPar{\procP[1]}{\procP[2]}$ then 
  $\tyT \equiv \tyPar{\tyT[1]}{\tyT[2]}$, with 
  $\G \vdash_{B_1} \procP[1] \ts \tyT[1]$ and 
  $\G \vdash_{B_2} \procP[2] \ts \tyT[2]$, with 
  $B = B_1 \cup B_2$.
\item If $\G\vdash_{B}\procP\ts\tyT$ and $P\wbarb{\actname}$ then:
  \begin{itemize}
  \item if $\actname\notin\{\labloaddual{\nVar},\labstoredual{\nVar}\}$ 
    then $\tyT\wbarb{\actname}$.
  \item if $\actname\in\{\labloaddual{\nVar},\labstoredual{\nVar}\}$ 
    then $\tyT\wbarb{\heapbarb{\nVar}}$.
  \end{itemize}
\item If $\G\vdash_{B}\procP\ts\tyT$ and $P\barb{\actname}$ then:
  \begin{itemize}
  \item if $\actname\notin\{\labloaddual{\nVar},\labstoredual{\nVar}\}$ 
    then $\tyT\barb{\actname}$.
  \item if $\actname\in\{\labloaddual{\nVar},\labstoredual{\nVar}\}$ 
    then $\tyT\barb{\heapbarb{\nVar}}$.
  \end{itemize}
\end{enumerate}
\end{lemma}

\begin{proof}
This is straight from the typing rules, much like Subject Congruence 
in Proposition~\ref{prop:cong}.
\end{proof}

We now prove the Safety Theorem:

\begin{proof}
We decompose safety in its three parts. 
Suppose $\pCall{\procX[0]}{}{} \tra{}^{\ast} \pRes{\vGeneric}\procQ$ and 
\begin{enumerate}
\item $\procQ\barb{\cclbarb{\channame}}$. Then, by Lemma~\ref{lem:inv}, 
  there exists $\Gi,\tyT$ such that $\Gi\vdash\procQ\ts\tyT$, and 
  we have $\tyT\barb{\cclbarb{\channame}}$. Safety of $\tyEqn$ entails safety 
  of $\tyT$ as a subterm of a reduced term from $\tyEqn$, thus 
  $\neg(\tyT\wbarb{\clbarb{\channame}})$ and $\neg(\tyT\wbarb{\ov{\channame}})$. 
  This implies, by applying the third point of Lemma~\ref{lem:inv} again, 
  $\neg(\procQ\wbarb{\clbarb{\channame}})$ and 
  $\neg(\procQ\wbarb{\ov{\channame}})$.
\item \begin{enumerate}
  \item $\procQ\barb{\ulckbarb{\nLock}}$. Then, by Lemma~\ref{lem:inv}, 
    there exists $\Gi,\tyT$ such that $\Gi\vdash\procQ\ts\tyT$, and 
    we have $\tyT\barb{\ulckbarb{\nLock}}$. Safety of $\tyEqn$ entails safety 
    of $\tyT$ as a subterm of a reduced term from $\tyEqn$, thus 
    $\tyT\barb{\lckmutbarb{\nLock}}$, and by Lemma~\ref{lem:barbs}, 
    $\procQ\barb{\lckmutbarb{\nLock}}$.
  \item $\procQ\barb{\rulckbarb{\nLock}}$. Then, by Lemma~\ref{lem:inv}, 
    there exists $\Gi,\tyT$ such that $\Gi\vdash\procQ\ts\tyT$, and 
    we have $\tyT\barb{\rulckbarb{\nLock}}$. 
    Safety of $\tyEqn$ entails safety 
    of $\tyT$ as a subterm of a reduced term from $\tyEqn$, thus 
    $\tyT\barb{\waitrmutbarb{\nLock}{i}}$, and by Lemma~\ref{lem:barbs}, 
    $\procQ\barb{\waitrmutbarb{\nLock}{i}}$.
  \end{enumerate}
\item $\procQ\barb{\stbarb{\nVar}{\occurgen}}$. Then, by Lemma~\ref{lem:inv}, 
    there exists $\Gi,\tyT$ such that $\Gi\vdash\procQ\ts\tyT$, and 
    we have $\tyT\barb{\stbarb{\nVar}{\occurgen}}$. 
    Data race freedom of $\tyEqn$ entails data race freedom 
    of $\tyT$ as a subterm of a reduced term from $\tyEqn$, thus 
    $\neg(\tyT\barb{\stbarb{\nVar}{\occurgeni}})$ and 
    $\neg(\tyT\barb{\ldbarb{\nVar}{\occurgeni}})$, for any 
    $\occurgeni\neq\occurgen$, and by Lemma~\ref{lem:inv}, 
    $\neg(\procQ\barb{\stbarb{\nVar}{\occurgeni}})$ and 
    $\neg(\procQ\barb{\ldbarb{\nVar}{\occurgeni}})$, for any 
    $\occurgeni\neq\occurgen$.
\end{enumerate}
This closes the proof of the Safety Theorem.
\end{proof}

{\noindent\large\bf Proposition~\ref{prop:may-conv-live}.}

\begin{proof}
Assume $\G\vdash\procProg\ts\tyEqn$ and $\tyEqn$ is live.
\begin{enumerate}
\item[(1)] Suppose by contradiction that 
  $\pCall{\procX[0]}{}{}\tra{}^{\ast}\procP\not\tra{}$ but 
  $\procP\not\equiv\zero$. Then there exists $\procQ$ such that 
  $\procP\equiv\pRes{\vGeneric}\procQ$ and either $\procQ\barb{\actname}$ 
  with $\actname\in\{\channame,\ov{\channame},
  \lckbarb{\nLock},\rlckbarb{\nLock}\}$, or $\procQ\barb{\actnameVec}$ for 
  some $\actnameVec$ (containing only blocking channel actions by definition). 
  Then, by Lemma~\ref{lem:inv}, this contradicts liveness for $\tyEqn$.
\item[(2)] As there is always a path to term $\zero$, which is live, then 
  all blocking actions available at any given point can be fired on any 
  available path to termination, hence $\procProg$ is live.
\end{enumerate}
\end{proof}

{\noindent\large\bf Proposition~\ref{prop:finite-branch-live}.}

We first prove a lemma for the conditional-free case:

\begin{lemma}\label{cond-free-live}
Suppose $\G\vdash\procProg\ts\tyEqn$, $\tyEqn$ is live and $\procProg$ 
is conditional-free, then $\procProg$ is live.
\end{lemma}

\begin{proof}
Since there is no conditional, all moves are strongly matched between types 
and processes, ie. if 
$\pCall{\procX[0]}{}{}\tra{}^{\ast}\pRes{\vGeneric}{\procP}$, using the 
Inversion Lemma we get $\Gi\vdash\procP\ts\tyT$ and we have 
$\procP\wbarb{\actname}$ iff $\tyT\wbarb{\actname}$. Thus liveness of $\tyEqn$ 
induces liveness of $\procProg$. 
\end{proof}

We now prove the proposition:

\begin{proof}
Suppose $\procProg\notin\ic$ and 
$\pCall{\procX[0]}{}{}\tra{}^{\ast}\pRes{\vGeneric}{\procP}$. 
By Inversion Lemma there exists $\Gi,\tyT$ such that $\Gi\vdash\procP\ts\tyT$. 
Since $\procProg\notin\ic$, we can always reduce to a term that is 
conditional-free, and along with the inversion Lemma again, there is 
$\procPi,\tyTi,\Gii$ such that $\procP\tra{}^{\ast}\procPi$, 
$\tyT\tra{}^{\ast}\tyTi$, $\Gii\vdash\procPi\ts\tyTi$ and $\procPi$ 
is conditional-free. We can then use Lemma~\ref{cond-free-live} to conclude.
\end{proof}

{\noindent\large\bf Theorem~\ref{the:alt-cond-live}.}

We first need a simple lemma again:

\begin{lemma}\label{lem:mapast-live}
If $\G\vdash\procP\ts\tyT$, then for $\actname\in\{\acttau,\labsync{\nGeneric}\}$, 
  $\tyT\wbarb{\actname}$ iff $\MAPAST{\procP}\wbarb{\actname}$.
\end{lemma}

\begin{proof}
As for the conditional free case, we have a strong match for actions between 
the types and the conditional mapping, by removing the determinism inherent 
from $\pITE{\nExpr}{\procP}{\procQ}$ constructs in \gol.
\end{proof}

We now prove the Theorem:

\begin{proof}
Suppose $\pCall{\procX[0]}{}{}\tra{}^{\ast}\pRes{\vGeneric}{\procP}$, 
Then by Inversion Lemma and subject reduction, we have $\tyT,\Gi$ 
such that $\Gi\vdash\procP\ts\tyT$ 
and $\tyCallee[0]\tra{}^{\ast}\tyRes{\vGeneric}{\tyT}$.
By Lemma~\ref{lem:mapast-live} we have that $\tyT\wbarb{\actname}$ iff 
$\MAPAST{\procP}\wbarb{\actname}$, and by $\procProg\in\ac$ we can conclude 
that $\tyT\wbarb{\actname}$ implies $\procP\wbarb{\actname}$. Thus, $\tyEqn$ 
being live entails $\procProg$ is live.
\end{proof}

\section{Go implementations of examples}
\label{app:example}
This section gives two implementations of the Dining Philosophers problem with shared memory, 
used in our benchmarks, and the implementation of the concurrent Prime Sieve algorithm 
we based the Example in \S~\ref{sect:safe-live} on.

\begin{wrapfigure}[14]{l}{0.9\linewidth}
{
\begin{lstlisting}[language=Go,style=golang,firstline=3]
func Generate(ch chan<- int) {
  for i :=2; ; i++ { ch <- i }
}

func Filter(in <-chan int, out chan<- int, prime int) {
  for { i := <-in
    if i%prime != 0 { out <- i }
  }
}

func main() {
  ch := make(chan int)
  go Generate(ch)
  for i := 0; ; i++ {
    prime := <-ch
    ch1 := make(chan int)
    go Filter(ch, ch1, prime)
    ch = ch1
  }
}
\end{lstlisting}
}
\vspace{-3mm}
\caption{Go implementation of a concurrent Prime Sieve algorithm~\cite{LNTY2017,NY2016}}
\label{fig:prime-sieve}
\end{wrapfigure}

\begin{wrapfigure}{l}{0.9\linewidth}
{
\begin{lstlisting}[language=Go,style=golang,firstline=10]
func Fork(fork *int, ch chan int) {
	for {
		*fork = 1
		<-ch
		ch <- 0
	}
}

func phil(fork1, fork2 *int, ch1, ch2 chan int, id int) {
	for {
		select {
		case ch1 <- *fork1:
			select {
			case ch2 <- *fork2:
				fmt.Printf("phil %d got both fork\n", id)
				<-ch1
				<-ch2
			default:
				<-ch1
			}
		case ch2 <- *fork2:
			select {
			case ch1 <- *fork1:
				fmt.Printf("phil %d got both fork\n", id)
				<-ch1
				<-ch2
			default:
				<-ch2
			}
		}
	}
}

func main() {
	var fork1, fork2, fork3, fork4, fork5 int
	ch1 := make(chan int)
	ch2 := make(chan int)
	ch3 := make(chan int)
	ch4 := make(chan int)
	ch5 := make(chan int)
	go phil(&fork1, &fork2, ch1, ch2, 0)
	go phil(&fork2, &fork3, ch2, ch3, 1)
	go phil(&fork3, &fork4, ch3, ch4, 2)
	go phil(&fork4, &fork5, ch4, ch5, 3)
	go phil(&fork5, &fork1, ch5, ch1, 4)
	go Fork(&fork1, ch1)
	go Fork(&fork2, ch2)
	go Fork(&fork3, ch3)
	go Fork(&fork4, ch4)
	go Fork(&fork5, ch5)
	time.Sleep(10*time.Second)
}
\end{lstlisting}
}
\vspace{-3mm}
\caption{Go implementation of the Dining Philosophers problem (unsafe)}
\label{fig:dinephil5-race}
\vspace{-5mm}
\end{wrapfigure}
\begin{wrapfigure}{l}{0.9\linewidth}
{
\begin{lstlisting}[language=Go,style=golang,firstline=10]
func Fork(fork *int, ch chan int) {
	for {
		*fork = 1
		ch <- 0
		<-ch
	}
}

func phil(fork1, fork2 *int, ch1, ch2 chan int, id int) {
	for {
		select {
		case <-ch1:
			select {
			case <-ch2:
				fmt.Printf("phil %d got both fork\n", id)
				ch1 <- *fork1
				ch2 <- *fork2
			default:
				ch1 <- *fork1
			}
		case <-ch2:
			select {
			case <-ch1:
				fmt.Printf("phil %d got both fork\n", id)
				ch2 <- *fork2
				ch1 <- *fork1
			default:
				ch2 <- *fork2
			}
		}
	}
}

func main() {
	var fork1, fork2, fork3, fork4, fork5 int
	ch1 := make(chan int)
	ch2 := make(chan int)
	ch3 := make(chan int)
	ch4 := make(chan int)
	ch5 := make(chan int)
	go phil(&fork1, &fork2, ch1, ch2, 0)
	go phil(&fork2, &fork3, ch2, ch3, 1)
	go phil(&fork3, &fork4, ch3, ch4, 2)
	go phil(&fork4, &fork5, ch4, ch5, 3)
	go phil(&fork5, &fork1, ch5, ch1, 4)
	go Fork(&fork1, ch1)
	go Fork(&fork2, ch2)
	go Fork(&fork3, ch3)
	go Fork(&fork4, ch4)
	go Fork(&fork5, ch5)
	time.Sleep(10*time.Second)
}
\end{lstlisting}
}
\vspace{-3mm}
\caption{Go implementation of the Dining Philosophers problem (safe)}
\label{fig:dinephil5-fix}
\vspace{-5mm}
\end{wrapfigure}

}{}

\end{document}